\documentclass[paper=a4,twoside=true,abstract=true,toc=bibliography,toc=index]{scrreprt}
\usepackage[a4paper,left=3.5cm,right=2.5cm,bottom=3.5cm,top=3cm]{geometry}
\usepackage{microtype}
\usepackage{authblk}
\usepackage{blindtext}
\usepackage[utf8]{inputenc}
\usepackage[T1]{fontenc}
\usepackage[german,english]{babel}
\usepackage[numbers,sort&compress,square]{natbib}
\usepackage{amsmath,amssymb,amsthm,mathrsfs,amsfonts,dsfont}
\usepackage[binary-units=true]{siunitx}
\usepackage[np]{numprint}
\usepackage{graphicx}
\usepackage{svg}
\usepackage{bbm}
\usepackage{esvect}
\usepackage{nicefrac}
\usepackage{color}
\usepackage{algorithm}
\usepackage{algpseudocode}
\usepackage{framed}
\usepackage{svg}
\usepackage{float}
\usepackage{tikz}
\usepackage{microtype}
\usepackage{todonotes}
\usepackage{epsfig}
\usepackage[onehalfspacing]{setspace}
\usepackage{varwidth}
\usepackage{placeins}
\usepackage{pdfpages}
\usepackage{varioref}
\usepackage{hyperref}
\usepackage[nameinlink,capitalise]{cleveref} 

\author{Dennis Rohde}
\title{Coresets for $(k,l)$-Clustering under the Fréchet Distance}

\DeclareMathOperator{\eqdef}{\stackrel{\text{def}}{=}}
\DeclareMathOperator{\aqdef}{\stackrel{\text{def}}{\Leftrightarrow}}

\DeclareMathOperator{\R}{\mathbb{R}}
\DeclareMathOperator{\N}{\mathbb{N}}
\DeclareMathOperator{\pcost}{cost}
\DeclareMathOperator{\pfrechet}{d_F}
\DeclareMathOperator{\peucl}{d_E}
\DeclareMathOperator{\pdfrechet}{d_{dF}}

\DeclareMathOperator{\FF}{\mathcal{F}}
\DeclareMathOperator*{\argmax}{arg\,max}
\DeclareMathOperator*{\argmin}{arg\,min}

\DeclareMathOperator{\affinespace}{\mathbb{A}}
\DeclareMathOperator{\euclideanspace}{\mathbb{E}}

\DeclareMathOperator{\dist}{dist}
\DeclareMathOperator{\E}{Ex}
\DeclareMathOperator{\Var}{Var}

\DeclareMathOperator{\Reach}{Reach}
\DeclareMathOperator{\pball}{B}

\DeclareMathOperator{\pcluster}{\mathcal{C}}
\DeclareMathOperator{\pdeg}{deg}
\DeclareMathOperator{\pindeg}{in-deg}
\DeclareMathOperator{\poutdeg}{out-deg}
\DeclareMathOperator{\pswap}{sw}
\DeclareMathOperator{\pgrid}{gr}
\DeclareMathOperator{\pcube}{cu}
\DeclareMathOperator{\pcell}{cl}
\DeclareMathOperator{\prange}{ra}
\DeclareMathOperator{\pangle}{ang}
\DeclareMathOperator{\protate}{Ro}
\DeclareMathOperator{\ptrans}{Tr}
\DeclareMathOperator{\paxisangle}{ax-ang}
\DeclareMathOperator{\ppoly}{poly}

\renewcommand{\P}{\ensuremath{P}}
\newcommand{\range}[2]{\ensuremath{\prange\left(#1,#2\right)}}
\newcommand{\cell}[2]{\ensuremath{\pcell\left(#1,#2\right)}}
\newcommand{\cube}[1]{\ensuremath{\pcube\left(#1\right)}}
\newcommand{\grid}[1]{\ensuremath{\pgrid\left(#1\right)}}
\newcommand{\eucl}[2]{\ensuremath{\peucl\left(#1,#2\right)}}

\newcommand{\norm}[1]{\ensuremath{\lVert#1\rVert}}
\newcommand{\normlz}[1]{\ensuremath{\norm{#1}_{2}}}
\newcommand{\normlzs}[1]{\ensuremath{\norm{#1}_{2}^2}}
\newcommand{\centroid}[1]{\ensuremath{\mu\left(#1\right)}}
\newcommand{\Tau}{\ensuremath{\mathrm{T}}}
\newcommand{\eqcfre}[1]{\ensuremath{\Delta_{#1}}}
\newcommand{\frechet}[2]{\ensuremath{\pfrechet\left(#1,#2\right)}}
\newcommand{\dfrechet}[2]{\ensuremath{\pdfrechet(#1,#2)}}
\newcommand{\cost}[2]{\ensuremath{\pcost(#1,#2)}}
\newcommand{\optcost}[1]{\ensuremath{\pcost(#1)}}
\newcommand{\On}[1]{\ensuremath{\mathcal{O}\left(#1\right)}}
\newcommand{\Om}[1]{\ensuremath{\Omega\left(#1\right)}}
\newcommand{\kk}{\ensuremath{k}}
\newcommand{\el}{\ensuremath{l}}
\newcommand{\klcenter}[1][ ]{\ensuremath{(k,l)}-\textsc{Center}#1}
\newcommand{\klmedian}[1][ ]{\ensuremath{k}-\textsc{Median}#1}
\newcommand{\klmeans}[1][ ]{\ensuremath{(k,l)}-\textsc{Means}#1}
\newcommand{\approxcost}[1]{\ensuremath{\widehat{\pcost}(#1)}}
\newcommand{\scalarm}[2]{\ensuremath{\left\langle #1, #2 \right\rangle}}

\newcommand{\Rd}{\ensuremath{\R^d}}

\newcommand{\powerset}[1]{\ensuremath{2^{\left(#1\right)}}}
\newcommand{\ball}[2]{\ensuremath{\pball\left(#1,#2\right)}}

\newcommand{\cluster}[3]{\ensuremath{\pcluster(#1,#2,#3)}}
\newcommand{\degree}[1]{\ensuremath{\pdeg(#1)}}
\newcommand{\indegree}[1]{\ensuremath{\pindeg(#1)}}
\newcommand{\outdegree}[1]{\ensuremath{\poutdeg(#1)}}
\newcommand{\nearest}[2]{\ensuremath{\eta\left(#1,#2\right)}}
\newcommand{\swap}[3]{\ensuremath{\pswap\left(#1,#2,#3\right)}}
\newcommand{\angl}[2]{\ensuremath{\pangle\left(#1,#2\right)}}
\newcommand{\rotate}[3]{\ensuremath{\protate\left(#1,#2,#3\right)}}
\newcommand{\trans}[2]{\ensuremath{\ptrans\left(#1,#2\right)}}
\newcommand{\axisangle}[2]{\ensuremath{\paxisangle\left(#1,#2\right)}}
\newcommand{\conferre}{cf.~}
\newcommand{\ie}{i.e.,~}
\newcommand{\coreset}[1][ ]{\ensuremath{\epsilon}-coreset#1}
\newcommand{\poly}[1]{\ensuremath{\ppoly\left(#1\right)}}

\def\multiset#1#2{\ensuremath{\left(\kern-.3em\left(\genfrac{}{}{0pt}{}{#1}{#2}\right)\kern-.3em\right)}}

\let\epsilon\relax
\newcommand{\epsilon}{\varepsilon}

\newcounter{theo}

\theoremstyle{plain}
\newtheorem{theorem}[theo]{Theorem}
\newtheorem{lemma}[theo]{Lemma}
\newtheorem{proposition}[theo]{Proposition}
\newtheorem{corollary}[theo]{Corollary}
\newtheorem{observation}[theo]{Observation}

\theoremstyle{definition}
\newtheorem{definition}[theo]{Definition}

\defcitealias{langbergschulman}{``Universal $\epsilon$-approximators for integrals''}

\algrenewcommand{\algorithmiccomment}[1]{// #1}
\begin{document}
	\setlength{\parindent}{0pt}
	\pagenumbering{roman}
	\begin{titlepage}
\definecolor{TUGreen}{rgb}{0.517,0.721,0.094}
\vspace*{-2cm}
\newlength{\links}
\setlength{\links}{-1.5cm}
\sffamily
\hspace*{\links}
\begin{minipage}{12.5cm}
\includegraphics[width=8cm]{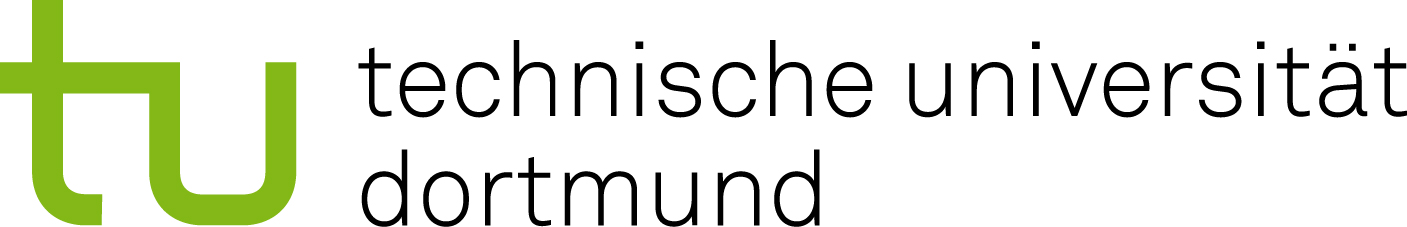}
\end{minipage}

\vspace*{4cm}

\hspace*{\links}
\hspace*{-0.2cm}
\begin{minipage}{9cm}
\large
\begin{center}
{\Large Masterthesis} \\
\vspace*{1cm}
\textbf{Coresets for $(k,l)$-Clustering under the Fréchet Distance} \\
\vspace*{1cm}
Dennis Rohde\\
December 2018
\end{center}
\end{minipage}
\normalsize
\vspace*{5.5cm}

\vspace*{2.1cm}

\hspace*{\links}
\begin{minipage}[b]{5cm}
\raggedright
Supervisors: \\
Jun.-Prof. Dr. Maike Buchin \\
M.Sc. Hendrik Fichtenberger \\
\end{minipage}

\vspace*{2.5cm}
\hspace*{\links}
\begin{minipage}[b]{10cm}
\raggedright
Technical University Dortmund \\
Department of Computer Science \\
Chair of Efficient Algorithms and Complexity Theory \\
\url{http://ls2-www.cs.tu-dortmund.de}
\end{minipage}

\end{titlepage}

	\begin{abstract}
		Clustering is the task of partitioning a given set of geometric objects. This is thoroughly studied when the objects are points in the euclidean space. There are also several approaches for points in general metric spaces. In this thesis we consider clustering polygonal curves, \ie curves composed of line segments, under the Fréchet distance. We obtain clusterings by minimizing an objective function, which yields a set of centers that induces a partition of the input.
		
		The objective functions we consider is the so called \klcenter[], where we are to find the $k$ center-curves that minimize the maximum distance between any input-curve and a nearest center-curve and the so called \klmedian[], where we are to find the $k$ center-curves that minimize the sum of the distances between the input-curves and a nearest center-curve.
		
		Given a set of $n$ polygonal curves, we are interested in reducing this set to an \coreset[], \ie a notably smaller set of curves that has a very similar clustering-behavior. We develop a construction method for such \coreset[s] for the \ensuremath{(k,2)}-\textsc{Center}, that yields \coreset[s] of size of a polynomial of $\nicefrac{1}{\epsilon}$, in time linear in $n$ and a polynomial of $\nicefrac{1}{\epsilon}$, for line segments. Also, we develop a construction technique for the \klcenter that yields \coreset[s] of size exponential in $m$ with basis $\nicefrac{1}{\epsilon}$, in time sub-quadratic in $n$ and exponential in $m$ with basis $\nicefrac{1}{\epsilon}$, for general polygonal curves, if the given curves are of ``good'' structure. Finally, we develop a construction method for the \klmedian[], that yields \coreset[s] of size polylogarithmic in $n$ and a polynomial of $\nicefrac{1}{\epsilon}$, in time linear in $n$ and a polynomial of $\nicefrac{1}{\epsilon}$.
	\end{abstract}
	This thesis is dedicated to my mother Antonia and my wife Lia who encouraged (and sometimes pushed) me to live out my curiosity in the scientific disciplines, which lead to this thesis and thus to the completion of my studies, eventually.
	\vspace{1em}
	
	Also, I want to thank Maike and Hendrik, not only for their good supervision, but also for their general support. Furthermore, I owe them the opportunity to begin my doctoral studies, for which I am very grateful.
	\tableofcontents
	\chapter{Introduction}
\pagenumbering{arabic}
Clustering is an old topic that appears in many shapes and variations. For example, classification may be one of the oldest tasks that relates to clustering. One of the first scientific classification tasks is due to Aristotle\footnote{This claim originates from \citet{hansen_aristotle}.}, \conferre \cite{aristotle_biology}, who classified the ``living things'' into animals with and without blood, according to the number of their legs and other observables. According to \citet{jain_50_years}, clustering is the task of discovering the ``natural grouping(s) of a set of patterns, points, or objects''. These groupings are the \emph{clusters} we are looking for. 

Clearly, the family of clustering problems stems from human intuition. For most of the real-world problems the human approach is non-uniform, non-general and very adaptive, \ie it almost completely depends on the given setting, the experience with such problems and the goal that is to be achieved. This introduces a fuzziness, \ie there are no exact instructions on how to obtain a clustering, which makes it hard to define a general approach that can be used by computers. Therefore, it is not very surprising that an exact definition of what clustering generally is, is nowhere to be found. 

Although clustering has a rich history to be told, we focus on mathematical formulations for clustering-related problems, which we call \emph{objective functions} or short \emph{objectives}. Our goal is to minimize the objective function at hand, \ie we are treating clustering problems as optimization problems. Minimizing one of the objective function in the focus of this work yields a set of \emph{centers}, which are the representatives of the clusters. One of the first and surely the most popular objective function is the \emph{$k$-means} objective, \conferre \cite{k_means_origin}. The aim of the $k$-means problem is to find a set of centers, such that the squared distances of the members of the clusters to their respective centers are minimal. Other formulations are a few decades young, \conferre \cite{rao_cluster, hansen_aristotle}. Since the introduction of $k$-means, \emph{computational} clustering had its triumphant march to one of the most fundamental topics of data analysis, \conferre \cite{jain_50_years}. It is now thoroughly studied and the approaches are widely employed. 

There are numerous applications for clustering in a broad range of fields, such as image segmentation, \conferre \cite{image_seg_chuang, image_seg_wu, images_seg_pappas, image_seg_coleman}, document hierarchization/information retrieval, \conferre \cite{docment_steinbach, document_fung, document_larsen, document_kummamuru, document_leuski, document_jardine}, and genome informatics, \conferre \cite{genome_kirzhner, genome_hatfull, genome_arima, genome_michaels}. Moreover, particle physicists are currently utilizing clustering techniques when the ouput-data of large particle detectors, such as ATLAS, CMS\footnote{Most commonly known for the fact that in 2012 scientists at ATLAS and CMS had a break-through that spread beyond the scientific community when they discovered the Higgs boson.} or ALICE, \conferre \cite{cluster_chekanov, cluster_higgs, cluster_jet, cluster_jet_origin}, shall be interpreted. Here interpreting means distinguishing signals from events to be observed from background noise. 

\begin{figure}
	\centering
	\includegraphics[width=\textwidth]{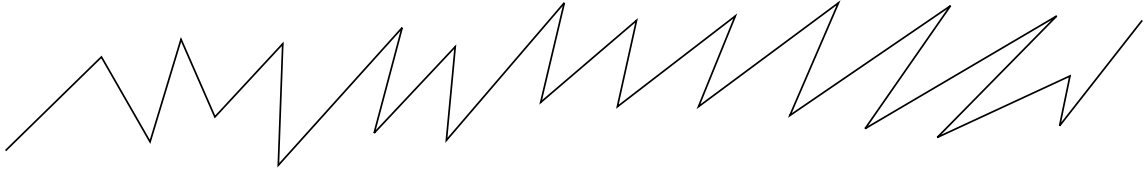}
	\caption{Simple polygonal curve in two-dimensional euclidean space that may arise from various applications. E.g., this could be an interpolation of a time-series of a physical quantity that was being measured, a simplified GPS-trajectory or a simplified backbone of a protein.}
	\label{fig:simple_poly_curve}
\end{figure}

\paragraph{Motivation for studying \coreset[s]} At this point the main topic of this work comes into play, which are \emph{\coreset[s]}. \coreset[s] -- as the name suggests -- are, generally speaking, sets of objects that are \emph{notably} smaller than a given set of objects, but reflect the core of the given set. To be precise, we allow an error of at most an $\epsilon$-fraction of the value of the objective function at hand, \ie we are dealing with \coreset[s] that have a similar clustering behavior. Our aim will be to obtain \coreset[s] of small cardinality, preferably independent of the input-set.

In the current time, with machines as the ATLAS detector producing data at about 3.2 Petabytes per year, \conferre \cite{cluster_atlas_data}, we reached the point where we may be able\footnote{Sometimes we are not!} to keep the whole data, but are not able to analyze it efficiently. These enormous amounts of data do not only arise in a few scientific fields like in the previous example, moreover it has become common to collect these amounts of data. The term for this phenomenon has established well: Big Data. Here we will face clustering problems under the aspect that the input is big data.

The reader may wonder why there seems to be no interest yet in polynomial-time approximation schemes for the clustering problems to be considered, or even in linear-time approximation schemes, for after all these problems are NP-hard, \conferre \cite{approx_k_l_center, clustering_time_series}. Also, one has to look at every input instance either way. 

In fact for geometric problems (quasi-)polynomial-time approximation schemes, or even (quasi-)linear-time approximation schemes, are often realized by constructing an \coreset for the problem at hand and analyzing it with an exact algorithm or by brute force\footnote{Normally this would be too expensive, but if the cardinality of the \coreset is independent of the input(-set) it can become cheap enough, though.}, \conferre \cite{kmeans_ptas, kmeans_ptas2}. There are other approaches, \conferre \cite{clustering_time_series, ptas_local, ptas_lloyd}. Nevertheless, these approaches are unique and have little in common while \coreset[s] yield a general approach for obtaining approximation schemes. Also, such ''direct`` approaches are mostly non-intuitive and non-trivial, while the techniques for constructing \coreset[s] can usually be derived from the problem at hand, thus are easier and therefore more beneficial to study.

\begin{figure}
	\centering
	\includegraphics[width=0.85\textwidth]{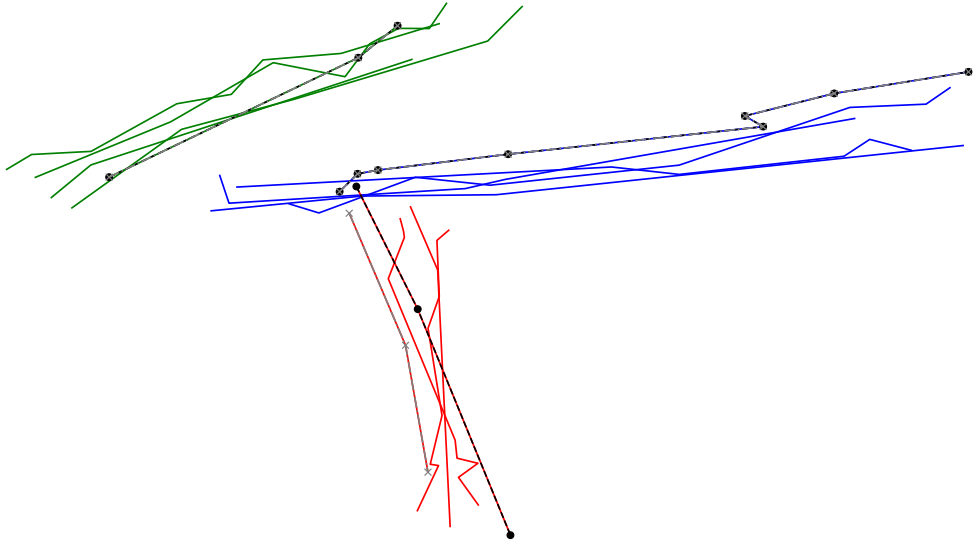}
	\caption{Randomly generated set of curves with three clusters in green, blue and red. Every cluster consists of five curves out of which the black one got selected by an implementation of \cref{algo:klcenter_approx_mod} and the gray one got selected by an implementation of \cref{algo:klmedian_approx} as center. Both implementations utilize the discrete Fréchet distance. Here both algorithms chose nearly the same centers, only the red cluster has different centers under the different objectives.}
	\label{fig:example1}
\end{figure}

\paragraph{The objects under investigation} As we already mentioned, clustering techniques are thoroughly studied -- in terms of points in the euclidean space\footnote{There are also approaches for general metric spaces, but the number of these is outweighed by the number of approaches for points in the euclidean space. We will use some of the approaches for metric spaces in this work, though.}. Here we are dealing with \emph{polygonal curves}, \conferre \cref{fig:simple_poly_curve}, \ie one dimensional objects in the $d$-dimensional euclidean space, with the restriction that these curves are composed of line segments. In every problem we study, we are given a set of $n$ polygonal curves of \emph{complexity} at most $m$ each, \ie every curve is composed of at most $m-1$ line segments. 

There are several similarity measures for curves, such as the Hausdorff distance or the Fréchet distance. We study the latter, because additional to the shape of the curves, this measure takes the order of the points on the curves into account -- contrary to the former, \conferre \cite{alt_godau_frechet}. The reason why we restrict our studies to polygonal curves is because mostly the data is given in such a representation, additionally the Fréchet distance between two polygonal curves can be computed efficiently, \conferre \cite{alt_godau_frechet}.

\begin{figure}
	\centering
	\includegraphics[width=0.85\textwidth]{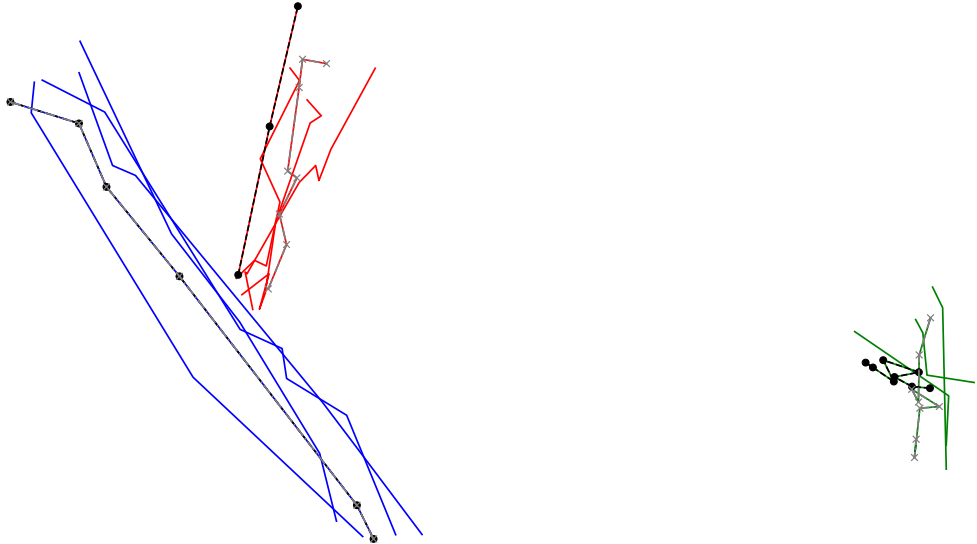}
	\caption{Randomly generated set of curves with three clusters in green, blue and red. Every cluster consists of five curves out of which the black one got selected by an implementation of \cref{algo:klcenter_approx_mod} and the gray one got selected by an implementation of \cref{algo:klmedian_approx} as center. Both implementations utilize the discrete Fréchet distance. Here we have two clusters, red and green, that have different centers under the different objectives.}
	\label{fig:example2}
\end{figure}

\paragraph{The objective functions of interest} We are mainly interested in two objective functions. The first objective, which we call \klcenter[], is to find a set of $k$ center-curves of complexity at most $l$ each, such that the maximum distance from any curve in the input-set to a nearest curve in the center-set is minimal. Usually we will refer to the center-set as the \emph{clustering}. In \cref{fig:example1}, \cref{fig:example2} and \cref{fig:example3} the black center-curves were computed by \cref{algo:klcenter_approx_mod} which yields an approximate solution to the center objective. The second objective, which we call \klmedian[], is to find a set of $k$ center-curves, such that the sum of the distances of the curves in the input-set to a nearest center in the center-set is minimal. Here we work with the restriction that the center-set is a subset of the input-set to make the problem less hard to compute. Thus, the center-curves have complexity at most $m$\footnote{This is the maximum complexity of the input-curves.} each. We will see that this restriction allows the application of sophisticated sampling techniques. In \cref{fig:example1}, \cref{fig:example2} and \cref{fig:example3} the gray center-curves were computed by \cref{algo:klmedian_approx} which yields an approximate solution to the median objective. Lastly we take a look at an objective, which we call \klmeans[], that is to find a set of $k$ centers of complexity at most $l$ each, such that the sum of the squared distances of the curves in the input-set to a nearest center in the center-set is minimal. This objective is derived from the already mentioned $k$-means, which was defined with respect to point-sets. With respect to point-sets this objective is particularly efficient to compute, but we will show that this may not be the case for polygonal curves.

\begin{figure}
	\centering
	\includegraphics[width=0.85\textwidth]{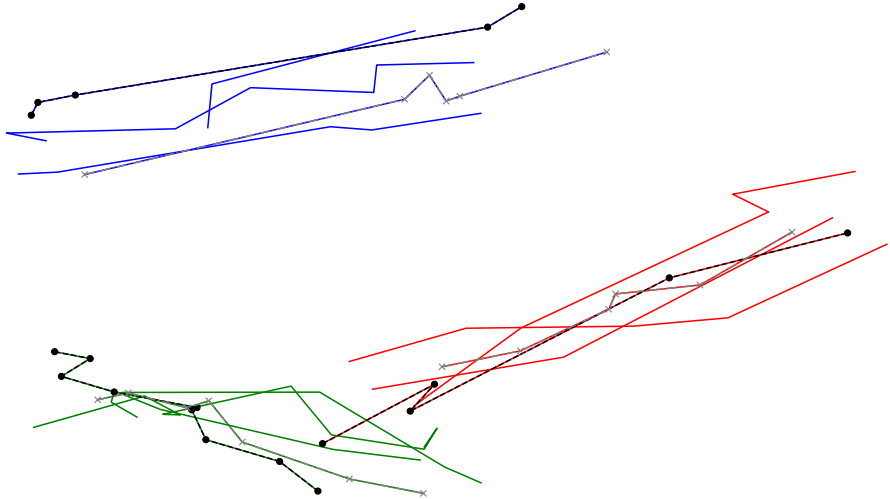}
	\caption{Randomly generated set of curves with three clusters in green, blue and red. Every cluster consists of five curves out of which the black one got selected by an implementation of \cref{algo:klcenter_approx_mod} and the gray one got selected by an implementation of \cref{algo:klmedian_approx} as center. Both implementations utilize the discrete Fréchet distance. Here every cluster has different centers under the different objectives. Comparing to the previous figures it can be observed that the less similar the curves within a cluster look, the more the choices of centers for the different objectives differ.}
	\label{fig:example3}
\end{figure}

\paragraph{On the running-times of the following algorithms} In the following we are studying clustering problems where we are given a set of polygonal curves, in $d$-dimensional euclidean space, of cardinality $n$ and two integers $k, l \in \mathbb{N}_{> 0}$ (or one integer $k \in \mathbb{N}_{>0}$). We assume $n$ and $m$ to be part of the input of the algorithm, while (as the nomenclature already suggests, e.g., \klcenter[]) $k$ and $l$ are part of the problem and are therefore fixed. Also, we assume that $d$ is fixed, \ie we are looking at a euclidean space of specific number of dimensions. Finally, for the \coreset[s] to be constructed, we assume the parameter $\epsilon$ to be a part of the input. 

In this work we assume the real RAM model, \conferre \cite{shamos}. Here simple geometric operations like translations and rotations can be done in constant time. Because we are mainly interested in the relation between the running-times of the following algorithms and their input, the running-times will be given in $\mathcal{O}$-notation.

\section{Main Results}
Our first result is an algorithm that constructs \coreset[s] for the center objective, for a set of line segments, in time linear in the number of given curves and polynomial in the reciprocal of the accepted error. The cardinality of a resulting \coreset is polynomial in the reciprocal of the accepted error.
\begin{theorem}
	\label{theo:algo_coreset_center_ls}
	There exists an algorithm that, given a set of $n$ line segments in $d$-dimensional euclidean space and a parameter $\epsilon \in (0,1)$, computes an \coreset for the \ensuremath{(k,2)}-\textsc{Center} objective of cardinality $\On{\frac{1}{\epsilon^{2d}}}$, in time $\On{\frac{n}{\epsilon^{2d}}}$.
\end{theorem}
Building upon this algorithm we develop an algorithm that constructs \coreset[s] for the center objective, for a set of polygonal curves, in time sub-quadratic in the number of given curves and exponential in their maximum complexity, with maximum complexity and reciprocal of the accepted error as base. The cardinality of the resulting \coreset is sub-linear in the number of given curves and exponential in their maximum complexity, with maximum complexity and reciprocal of the accepted error as base. Unfortunately the algorithm does not work for every setting. It is designed for curves of much larger number than maximum complexity, such that $k \cdot 2^{3m} \cdot \sqrt{n} \cdot \frac{l^{12d^2m}}{\epsilon^{dm}} + 2^m m^m \ll n$ holds, and that have edges short compared to the distances among the curves. If the input does not fit into this setting the algorithm fails and is not able to return an \coreset[].
\begin{theorem}
	\label{theo:algo_coreset_center_pc}
	There exists an algorithm that, given a set of $n$ polygonal curves of complexity at least $3$ and at most $m$ each, in $d$-dimensional euclidean space and a parameter $\epsilon \in (0,1)$, computes an \coreset for the \klcenter objective of cardinality $\On{2^{3m} \cdot \sqrt{n} \cdot \frac{l^{12d^2m}}{\epsilon^{dm}} + 2^m m^m}$ in time \par $\On{\left(2^{3m} \cdot n^{1.5} \cdot \frac{l^{12d^2m}m}{\epsilon^{dm}} + 2^m m^{m+1} n \right) + nm \log(m) + m^3 \log(m)}$, if successful. Otherwise, the algorithm fails and then has running-time $\On{nm \log(m) + m^3 \log(m)}$.
\end{theorem}
Finally, we develop an algorithm that constructs a set of polygonal curves, for a given set of polygonal curves, that is an \coreset for the median objective with constant probability, in time quadratic in the number of given curves, sub-cubic in their maximum complexity, and quadratic in the reciprocal of the accepted error. It has cardinality logarithmic in the number of given curves and quadratic in the reciprocal of the accepted error.
\begin{theorem}
	\label{theo:algo_coreset_median}
	There exists an algorithm that, given a set of $n$ polygonal curves of complexity at most $m$ each, and a parameter $\epsilon \in (0,1)$, computes an \coreset for the \klmedian objective of cardinality $\On{\frac{\ln(n)}{\epsilon^2}}$ in time $\On{n^2 \cdot m^2 \log(m) + \frac{\ln^2(n)}{\epsilon^2}}$, with  probability at least $\nicefrac{2}{3}$.
\end{theorem}

\section{Related Work}
Most related to this work are the works of \citet{clustering_time_series} and \citet{approx_k_l_center}. The former work introduces the \klcenter and $(k,l)$-\textsc{median} objectives and provides quasi-linear-time approximation schemes for these with running-times $\widetilde{\mathcal{O}}(n\cdot m)$\footnote{This notation hides logarithmic factors.}, under the restriction that $d=1$ and $\epsilon, k$ and $l$ are fixed. They introduce the concept of signatures: A signature of a given curve is another curve which can be seen as a summary of the characteristics of the given curve. \citeauthor{clustering_time_series} express this in a more mathematical sense: The signature ``captures the critical points of the curve''. It is the polygonal curve whose vertices are critical points of the given curve. They show that the signature-vertices of the input-curves must be matched to distinct vertices of their respective nearest centers in the clustering. This eventually leads to algorithms which generate a constant size set of candidate solutions to the respective clustering objective, which contains an $(1+\epsilon)$-approximate solution. Also, they provide constant-factor approximation algorithms for the center and the median objective that have running-time near linear in the number of curves and their maximum complexity. These algorithms are extensions of the algorithms of \citet{gonzalez} and \citet{chen}. Further they prove that clustering under the Fréchet distance with respect to the center objective or to the median objective is NP-hard\footnote{Where $k$ is part of the input.}, where $l \geq 2$ and $d=1$. They do so by providing an isometry that embeds clustering-instances of point-sets in the euclidean space, which are known to be NP-hard to solve, into a Fréchet-space, thus obtaining reductions from clustering points-sets to clustering curve-sets. Finally, they show that the doubling-dimension of the Fréchet-space is unbounded. This was a motivation for their work, because existing $(1+\epsilon)$-approximation algorithms with similar running-times do only work for spaces with bounded doubling dimension.

The latter work follows the first work and provides a $3$-approximation algorithm for the \klcenter objective with running-time linear in the number of given polygonal curves and in their maximum complexity plus time sub-quartic in the maximum complexity. This algorithm, \conferre \cref{algo:klcenter_approx} in \cref{subsec:center_approx}, is also an extension of the algorithm of \citet{gonzalez}, where the resulting center-curves are simplified. They prove that, given a set of polygonal curves, the center objective is NP-hard\footnote{Where $l$ is part of the input.} to approximate within a factor of $2 - \epsilon$ for the discrete Fréchet distance and $1.5 - \epsilon$ for the continuous Fréchet distance, for $d=1$ and $\epsilon > 0$, even if $k=1$. They do so by providing reductions from the shortest common supersequence problem\footnote{This problem is NP-hard for binary strings.}, \ie we are given a set of strings and shall compute a string that has every of those strings as substring and is as short as possible, to the $(1,l)$-\textsc{center} problem\footnote{The decision version of this problem.} under the respective measure. The reductions compute a curve for every input-string, which is composed of so called letter gadgets and buffer gadgets. Now if the set of strings has a common supersequence of length $t$, then there exists a center with $2t+1$ vertices that lies within distance $1$ to the generated curves and vice versa. 

Further they prove that the center objective is NP-hard to approximate within a factor of less than $3 \sin(\nicefrac{\pi}{3}) - \epsilon$ for the discrete Fréchet distance and within a factor of $2.25 - \epsilon$ for the continuous Fréchet distance, for $d \geq 2$ and $\epsilon > 0$. These results are also achieved by reductions from the shortest common supersequence problem to the $(1,l)$-\textsc{center} problem\footnote{The decision version of this problem} under the respective measure. These reductions are more complicated, though. In each reduction a curve is computed for each string in the input-set. These curves are composed of A-gadgets and B-gadgets, that are polygonal curves which traverse ten points, which have a special alignment, in different orders. This time, if the set of strings has a common supersequence of length $t$, then there exists a center with $6t^2 + 9t$ vertices within distance $1$ to the generated curves and vice versa.

In these reductions, for $d=1$ and $d \geq 2$, the center-curves have restricted complexity. \citeauthor{approx_k_l_center} show that even if this restriction is loosened the problem remains NP-hard, \ie the minimum enclosing ball problem\footnote{In the respective Fréchet-dimension.} is NP-hard for polygonal curves under the discrete and continuous Fréchet distance, for $d=1$.  They do so by providing polynomial time truth-table reductions from the shortest common supersequence problem to the minimum enclosing ball problem\footnote{The decision version of the problem.}.

A work that is loosely related to this work is the work of \citet{coresers_methods_history}. In this work the authors summarize different techniques for constructing \coreset[s] for different exemplary applications. They describe how \coreset[s] for the well-known $k$-means can be obtained by a \emph{geometric decomposition} through $\epsilon$-ball covers, utilizing only a generalized version of the triangle-inequality. Another simple technique they describe is \emph{gradient descent}, which is derived from complex optimization. They show how this technique can be utilized to obtain \coreset[s] for the smallest enclosing ball problem, \ie given a set of points in $d$-dimensional euclidean space one shall compute a center point, such that the ball around this center that encloses the given points has minimal radius. A prerequisite for this technique is that a sub-gradient of the function that shall be minimized must exist and must be known.

Another technique they describe is \emph{random sampling}. They show how weak\footnote{Here weak means that the \coreset does not give a guarantee for all choices of parameters of the function at hand.} \coreset[s] for the geometric median, \ie we are given a set $P$ of points in $d$-dimensional euclidean space and shall find the point $c$ that minimizes $\sum_{p \in P} \normlz{p - c}$, can be obtained through uniform random sampling. Building upon this approach they show how (strong) \coreset[s] for the geometric median can be constructed by non-uniform random sampling. For this purpose they utilize the sensitivity sampling framework, which will also be used in this work. 

At last, they explain a technique called \emph{sketches and projections}. The core of the technique is that a set $P$ of points in $\R^d$ can be viewed as a matrix, or more precisely as an arbitrary matrix $A$ from a set of matrices. This set of matrices consist of the matrices from $\R^{n\times d}$ where the $i$\textsuperscript{th} row corresponds to the $i$\textsuperscript{th} point in $P$, for an arbitrary ordering. A sketch of $A$ is a linear projection obtained by multiplying $A$ with a matrix $S \in \R^{m \times n}$ with $m$ notably smaller than $n$. An \coreset can then be obtained by mapping $A \mapsto S\cdot A$. A central ingredient in the technique is the Johnson-Lindenstrauss Lemma, which states that there exists a distribution over matrices, such that a matrix drawn from this distribution yields a sketch of $A$ with constant probability.

Also, they sketch a way of obtaining streaming algorithms from \coreset constructions, the so called \emph{merge and reduce}. The key idea is that, for most of the functions that are studied, if we have a set of points $P \eqdef Q \cup W$ and an \coreset $S_Q$ for $Q$ and an \coreset $S_W$ for $W$, then $S \eqdef S_Q \cup S_W$ is an \coreset for $P$. In the streaming-framework the input is processed point by point and streaming-algorithms have the restriction that the occupied memory does at most have a poly-logarithmic dependency on the cardinality of the input-set. Now the idea is to partition the input into batches of cardinality $\On{\log(n)}$. These batches are seen as the leafs of a binary tree of at most logarithmic height. This tree is processed bottom-up: For every two child-nodes we compute an \coreset and merge these, thus obtaining an \coreset for the parent-node. The children can be deleted. When the process is finished the resulting \coreset is at the root-node, with approximation guarantee $(1+ \epsilon)^{\log(n)}$. Rescaling $\epsilon$ by $\nicefrac{1}{2 \log(n)}$ gives the desired guarantee of less than or equal to $(1+\epsilon)$. All in all at most a logarithmic number of \coreset[s] have to be saved, so the restriction on the memory is adhered to.

Finally, lower bounds on the cardinality of \coreset[s] for certain functions are provided, such as for logistic regression. They show that for any $\delta > 0$ there exist a point-set, such that an \coreset for logistic regression for this point-set must have size $\Om{n^{1-\delta}}$.

At last, the work of \citet{melanie} gives a nice introduction to clustering and \coreset[s] in general, going into further detail for the $k$-\textsc{means} problem.
	\chapter{Preliminaries}
	Here we formalize the mathematical concepts that are used throughout this work and give a short introduction about their later application.
	
	We start with graph theory which provides us a framework that can be used to analyze basic properties of functions. Then we move on to euclidean geometry which is the basis of this work. It provides the notions of points, curves, planes, distances, angles and other fundamentals of the field of geometry. We introduce clustering and \coreset[s], the main concepts in this work and finally we introduce certain aspects of probability theory, which provides us sophisticated tools to tackle some harder problems that are in the focus of this work.  
	\section{Graph Theory}
		We give some definitions of the most basic notions of graph theory. These are sufficient for the purposes of this work, though graph theory is a comprehensive theory.
		\begin{definition}[graph, \conferre\citet{diestel}]
			A directed graph is a pair $G \eqdef (V,E)$ of a set of vertices $V$ and edges $E \subseteq V \times V$. A graph is complete, iff $E = V \times V$.
		\end{definition}
		Abstractly speaking, a graph merely is a binary relation over some ground-set. For example, we can define the function-graph for a unary function $f \colon V \rightarrow W$ as follows: $G_f \eqdef (V\cup W, E_f)$ with $(v,w) \in E_f \aqdef f(v) = w$. The degree of a member of $V$, respective $W$, can be used to infer some basic properties of $f$.
		\begin{definition}[degree, \conferre{\citet{diestel}}]
			The degree of a vertex $v \in V$ is the number of edges at $v$, formally \[ \degree{v} \eqdef \indegree{v} + \outdegree{v}, \] where \[ \indegree{v} \eqdef \vert \{ (w,x) \in E \mid x = v \} \vert \] and \[ \outdegree{v} \eqdef \vert \{ (w,x) \in E \mid w = v \} \vert. \]
		\end{definition}
	\section{Euclidean Geometry}
		We define the basics of euclidean geometry, which are points and movements between points, which we call vectors. These concepts are formalized in an affine space.
		\begin{definition}[affine space, \conferre{\citet[Definition 5]{prasolov_geometry}}]
			The affine space $\affinespace^d \eqdef (\Rd, \Rd, +)$ over $\R^d$ is a triple of a set of points $p \eqdef (p_1, \dots, p_d) \in \Rd$, where the $p_1, \dots, p_d$ are the coordinates of the point, a set of vectors $\vv{pq} \eqdef (q_1 - p_1, \dots, q_d - p_d) \in \Rd$, where $p$ is the initial point and $q \eqdef (q_1, \dots, q_d)$ is the end point of the vector and an operation $+ \colon \Rd \times \Rd \rightarrow \Rd$ that applies a vector to a point: $p + \vv{pq} \eqdef (p_1 + q_1 - p_1, \dots, p_d + q_d - p_d)$. The point $o \eqdef \{0\}^d$ is the origin of coordinates and $\vv{o} \eqdef \{0\}^d$ is the zero vector.
		\end{definition}
		Note that for each point $p \in \affinespace^d$: $\vv{op} = p$. We call $\vv{op}$ the position vector of $p$. Let $\vv{x}$ be a vector, by $i(\vv{x}), e(\vv{x})$ we denote the initial point, respective end point of $\vv{x}$.
		
		An affine space is oblivious of the notions of distances and angles, which are formalized in the euclidean space.
		\begin{definition}[euclidean space, \conferre{\citet[Definition 6]{prasolov_geometry}}]
			\label{def:euclideanspace}
			The $d$-dimensional euclidean space $\euclideanspace^d$ is the affine space $\affinespace^d$, endowed with the scalar product of vectors $\scalarm{\vv{x}}{\vv{y}} \eqdef \sum\limits_{i=1}^{d} x_i \cdot y_i$, where $\vv{x} \eqdef (x_1, \dots, x_d)$, $ \vv{y} \eqdef (y_1, \dots, y_d)$ and distance measure $\eucl{p}{q} \eqdef \sqrt{\sum\limits_{i=1}^{d} (p_i - q_i)^2}$ between two points $p \eqdef (p_1, \dots, p_d)$ and $q \eqdef (q_1, \dots, q_d)$. Also the length of $\vv{x}$ is defined $\norm{\vv{x}} \eqdef \sqrt{\sum\limits_{i=1}^{d} x_i^2}$. The angle between two non-zero vectors $\vv{x}, \vv{y}$ in \emph{radians}, \ie a number in $[0, 2\pi]$, is $\angl{\vv{x}}{\vv{y}} \eqdef \arccos\left(\frac{\scalarm{\vv{x}}{\vv{y}}}{\norm{\vv{x}}\cdot\norm{\vv{y}}}\right)$.
		\end{definition}
	 	Two vectors $\vv{x},\vv{y}$ are said to be orthogonal, if $\scalarm{\vv{x}}{\vv{y}} = 0$. We will also use the scalar product for points, since they stem from the same field as vectors and sometimes we will use norms to express distances.
		\begin{definition}[distance]
			Let $p \eqdef (p_1, \dots, p_d), q \eqdef (q_1, \dots, q_d)$ be two points in $\euclideanspace^d$. The $\ell_2$-norm $\normlz{p-q} \eqdef \sqrt{\sum_{i=1}^d (p_i - q_i)^2}$ is the euclidean distance between $p$ and $q$. The squared $\ell_2$-norm $\normlzs{p-q} \eqdef \sum_{i=1}^d (p_i-q_i)^2$ is the squared euclidean distance between $p$ and $q$.
		\end{definition}
	 	Further it holds that $\normlzs{p-q} = \scalarm{p-q}{p-q}$. Now that we have defined the basics of euclidean geometry we emphasize that the euclidean space has an intrinsic property: In euclidean space the number of distinct points that can share a common nearest neighbor is not unbounded, more precisely every number $d$ of dimensions of the euclidean space has its own bound for the maximum number of points that can share a common nearest neighbor. These bounds are called the $d$-dimensional kissing numbers.
	 	\begin{theorem}[kissing number, {\citet[Theorem 1]{zeger_gersho}}]
	 		\label{theo:kiss}
	 		The maximum number of distinct points in $\euclideanspace^d$ that can have a common nearest neighbor is equal to the kissing number $\psi_d$ which is bounded as follows: \[ 2^{0.2075d(1+o(1))} \leq \psi_d \leq 2^{0.401d(1+o(1))} \] Here $o(1)$ denotes an asymptotic of the number of dimensions $d$.
	 	\end{theorem}
	 	
	 	We are ready to define some basic notions of motion in the euclidean space.
	 	\begin{definition}[motions, \conferre \citet{prasolov_geometry}]
	 	 	\label{def:motions}
	 		Let $p \in \euclideanspace^d$ be a point. A counter-clockwise \emph{rotation} by the angle $\alpha \in [-2\pi,2\pi]$ (clockwise by $-\alpha$, if $\alpha < 0$) around the origin in the plane spanned by the axes of the $i$\textsuperscript{th} and $(i+1)$\textsuperscript{th} dimension, for $i \in \{1,\dots,d-1\}$, gives the point $\rotate{p}{i}{\alpha} \eqdef (q_1, \dots, q_d)$, where $q_j \eqdef p_j$ for $j \in \{1, \dots, d\} \setminus \{i, i+1\}$, $q_i \eqdef p_i \cos(\alpha) - p_{i+1} \sin(\alpha)$ and $q_{i+1} \eqdef p_i \sin(\alpha) + p_{i+1} \cos(\alpha)$. A \emph{translation} of $p$ in the direction of a vector $\vv{x} \in \euclideanspace^d$ gives the point $\trans{p}{\vv{x}} \eqdef p + \vv{x}$.
	 	\end{definition}
	 	Let $\mathcal{P}$ be a \emph{plane}, \ie a two-dimensional subspace of $\euclideanspace^d$, \conferre \cite{prasolov_geometry}, and $p \in \euclideanspace^d$ be a point. We call the point $q \in \mathcal{P}$ the (orthogonal) \emph{projection} of $p$ onto $\mathcal{P}$, iff $q = \trans{p}{\vv{x}}$ and $\vv{x}$ is orthogonal to $\vv{rs}$, where $r,s \in \mathcal{P}$, \ie $q$ is the orthogonal translation of $p$ onto $\mathcal{P}$.

		An isometry is a function that embeds one metric space into one other. We consider embeddings of curves under the Fréchet distance, \ie $(\Tau, \pfrechet)$, into points in the euclidean space, \ie $\euclideanspace^d$.
		\begin{definition}[isometry, \conferre{\citet{isometries}}]
			\label{def:isometry}
			Let $(X, \dist_1), (Y, \dist_2)$ be two metric spaces. A mapping $f\colon X \rightarrow Y$ is called isometry, iff: 
			\[ \forall p,q \in X: \dist_2(f(p), f(q)) =  \dist_1(p,q) \]
		\end{definition}
		We say that $X$ is embedded into $Y$ by $f$. It is easy to see that if $f$ is invertible then its inverse is also an isometry. 
		
		Now we show that the motions defined in \cref{def:motions} are isometries that embed the euclidean space into itself.
		\begin{proposition}
			\label{prop:rotation_isometry}
			For a fixed $i \in \{1, \dots, d-1\}$ and a fixed $\alpha \in  \R$ the motion $\protate$ is an isometry that embeds the euclidean space into itself.
		\end{proposition}
		\begin{proof}
			Let $p \eqdef (p_1, \dots, p_d),q \eqdef (q_1, \dots, q_d) \in \euclideanspace^d$ be two arbitrary points. Let \[p^\prime \eqdef (p^\prime_1, \dots, p^\prime_d) \eqdef \rotate{p}{i}{\alpha}\] and \[q^\prime \eqdef (q^\prime_1, \dots, q^\prime_d) \eqdef \rotate{q}{i}{\alpha}.\] We have:
			\begin{align*}
			\eucl{\rotate{p}{i}{\alpha}}{\rotate{q}{i}{\alpha}}^2 ={} & \sum_{j=1}^{d} (q^\prime_j - p^\prime_j)^2 
			=  \sum_{j=1}^{i-1} (q_j - p_j)^2 + \sum_{j=i+2}^{d} (q_j - p_j)^2 \\
			& + (q_i \cos(\alpha) - q_{i+1} \sin(\alpha) - p_i \cos(\alpha) + p_{i+1} \sin(\alpha))^2 \\ 
			& + (q_i \sin(\alpha) + q_{i+1} \cos(\alpha) - p_i \sin(\alpha) - p_{i+1} \cos(\alpha))^2 \\
			={} & \sum_{j=1}^{i-1} (q_j - p_j)^2 + \sum_{j=i+2}^{d} (q_j - p_j)^2 \\
			& + (\cos(\alpha)(q_i - p_i) +  \sin(\alpha)(p_{i+1} - q_{i+1}))^2 \\
			& + (\sin(\alpha)(q_i - p_i) + \cos(\alpha)(q_{i+1} - p_{i+1}))^2 \\
			={} & \sum_{j=1}^{i-1} (q_j - p_j)^2 + \sum_{j=i+2}^{d} (q_j - p_j)^2 \\
			& + \cos^2(\alpha)(q_i - p_i)^2 + \sin^2(\alpha)(q_i - p_i)^2 \\
			& + \cos^2(\alpha)(q_{i+1} - p_{i+1})^2 + \sin^2(\alpha)(p_{i+1} - q_{i+1})^2 \\
			& + 2\cos(\alpha)(q_i - p_i)\sin(\alpha)(p_{i+1} - q_{i+1}) \\
			& + 2\sin(\alpha)(q_i - p_i)\cos(\alpha)(q_{i+1} - p_{i+1}) \\ 
			={} & \sum_{j=1}^{i-1} (q_j - p_j)^2 + \sum_{j=i+2}^{d} (q_j - p_j)^2 + (q_i - p_i)^2 + (q_{i+1} - p_{i+1})^2 \\
			& + 2\cos(\alpha)\sin(\alpha)(q_i p_{i+1} - q_i q_{i+1} - p_i p_{i+1} + p_i q_{i+1}) \\
			& + 2\cos(\alpha)\sin(\alpha)(q_i q_{i+1} - q_i p_{i+1} - p_i q_{i+1} + p_i p_{i+1}) \\
			={} & \sum_{j=1}^{d} (q_j - p_j)^2 = \eucl{p}{q}^2
			\end{align*}
			We use the fact that for any $x \in \R$ it holds that $\cos^2(x) + \sin^2(x) = 1$ and the fact that for any $x, y \in \R$ it holds that $(x-y)^2 = (y-x)^2$.
			
			Taking the square-root of both sides proves the claim.
		\end{proof}
		\begin{proposition}
			\label{prop:translation_isometry}
			For a fixed vector $\vv{x} \in \euclideanspace^d$ the motion $\ptrans$ is an isometry that embeds the euclidean space into itself.
		\end{proposition}
		\begin{proof}
				Let $p \eqdef (p_1, \dots, p_d),q \eqdef (q_1, \dots, q_d) \in \euclideanspace^d$ be two arbitrary points. Let \[p^\prime \eqdef (p^\prime_1, \dots, p^\prime_d) \eqdef \trans{p}{\vv{x}}\] and \[q^\prime \eqdef (q^\prime_1, \dots, q^\prime_d) \eqdef \trans{q}{\vv{x}}.\] We have:
				\begin{align*}
					\eucl{\trans{p}{\vv{x}}}{\trans{q}{\vv{x}}} = \sqrt{\sum_{i=1}^d (q^\prime_i - p^\prime_i)^2} = \sqrt{\sum_{i=1}^{d} (q_i + \vv{x} - p_i - \vv{x})^2} = \eucl{p}{q}
				\end{align*}
		\end{proof}
		Finally we show that these motions are also invertible.
		\begin{proposition}
			\label{prop:rotation_inverse}
			For a fixed $i \in \{1, \dots, d-1\}$ and a fixed $\alpha \in  \R$ the motion $\protate$ is invertible.
		\end{proposition}
		\begin{proof}
			We show that for every $p \eqdef (p_1, \dots, p_d) \in \euclideanspace^d$ and $q = \rotate{p}{i}{\alpha}$ it holds that $q^\prime \eqdef (q^\prime_1, \dots, q^\prime_d) \eqdef \rotate{q}{i}{-\alpha} = p$.
			
			By \cref{def:motions} we have \[\forall j \in \{1, \dots, d\}: (j \neq i \wedge j \neq i + 1) \Rightarrow (p_{j} = q_j)\] and \[\forall j \in \{1, \dots, d\}: (j \neq i \wedge j \neq i + 1) \Rightarrow (q_{j} = q^\prime_j),\]
			hence \[\forall j \in \{1, \dots, d\}: (j \neq i \wedge j \neq i + 1) \Rightarrow (p_{j} = q^\prime_{j}).\]
			Further we have \[q_i =\cos(\alpha) p_i - \sin(\alpha) p_{i+1}\] and \[q_{i+1} = \sin(\alpha) p_i + \cos(\alpha) p_{i+1}.\]
			Using the fact that $\cos(x) = \cos(-x)$ and $-\sin(x) = \sin(-x)$, we obtain
			\begin{align*}
				q_i^\prime ={} & \cos(-\alpha) q_i - \sin(-\alpha) q_{i+1} = \cos(\alpha) q_i + \sin(\alpha) q_{i+1} \\
				={} & \cos(\alpha) \cdot (\cos(\alpha) p_i - \sin(\alpha) p_{i+1}) + \sin(\alpha) \cdot (\sin(\alpha) p_i + \cos(\alpha) p_{i+1}) \\
				={} & \cos^2(\alpha) p_i + \sin^2(\alpha) p_i + \sin(\alpha)\cos(\alpha) p_{i+1} - \cos(\alpha)\sin(\alpha) p_{i+1} \\
				={} & p_i
			\end{align*}
			and
			\begin{align*}
				q_{i+1}^\prime ={} & \sin(-\alpha) q_i + \cos(-\alpha) q_{i+1} = -\sin(\alpha) q_i + \cos(\alpha) q_{i+1} \\
				={} & -\sin(\alpha) \cdot (\cos(\alpha) p_i - \sin(\alpha) p_{i+1}) + \cos(\alpha) \cdot (\sin(\alpha) p_i + \cos(\alpha) p_{i+1}) \\
				={} & \sin^2(\alpha) p_{i+1} + \cos^2(\alpha) p_{i+1} + \cos(\alpha)\sin(\alpha) p_i - \sin(\alpha)\cos(\alpha) p_i \\
				={} & p_{i+1},
			\end{align*}
			which implies that $q^\prime = p$, thus finishes the proof.
		\end{proof}
		\begin{proposition}
			\label{prop:translation_inverse}
			For a fixed vector $\vv{x} \eqdef (x_1, \dots, x_d) \in \euclideanspace^d$ the motion $\ptrans$ is invertible.
		\end{proposition}
		\begin{proof}
			We show that for every $q \eqdef (q_1, \dots, q_d) \in \euclideanspace^d$ there exists one, and only one, $p_1 \eqdef (p_{1,1}, \dots, p_{1,d}) \in \euclideanspace^d$ with $q = \trans{p_1}{\vv{x}}$.
			
			Assume there exists a $p_2 \eqdef (p_{2,1}, \dots, p_{2,d}) \in \euclideanspace^d$ with $p_2 \neq p_1$ and $q = \trans{p_2}{\vv{x}}$. By \cref{def:motions} we have 
			\begin{align*}
				&& \forall i \in \{1,\dots,d\}: q_{i} = p_{2,i} + x_i = p_{1,i} + x_i \\
				\Leftrightarrow && \forall i \in \{1,\dots,d\}: p_{2,i} = p_{1,i} \label{trbij:eq1} \tag{I}
			\end{align*}
			By \cref{trbij:eq1} it is implied that $p_2 = p_1$ which leads to a contradiction and finishes the proof. Also, it can be observed that $p_1 = p_2 = \trans{q}{-\vv{x}}$.  
		\end{proof}
		\subsection{Objects in the Euclidean Space}
		We define the central objects, that are the foundations of most of the techniques we employ to construct \coreset[s]. One object that is very central in most of the techniques utilized for the well-known \emph{$k$-means} clustering problem is the centroid.
		\begin{definition}[centroid]
			\label{def:centroid}
			Let $\P \subset \euclideanspace^d$ be a set of points. By $\centroid{\P} \eqdef \frac{1}{\vert \P \vert} \sum\limits_{p \in \P} p$ we denote the centroid of the points in $P$.
		\end{definition}
		The centroid has the feature that it is the optimal \emph{$1$-means} for a set of points, which can be shown analytically:
		\begin{proposition}[{\citet[Lemma 2.1]{zauberformel_herkunft}}]
			\label{prop:zauberformel}
			Let $P \subset \euclideanspace^d$ be a set of points. The centroid of $P$, \ie $\centroid{P}$, is the point that minimizes the sum of squared euclidean distances with respect to the points in $P$.
		\end{proposition}
		\begin{proof}
			Let $q \in \euclideanspace^d$ be an arbitrary but fixed point. We have:
			\begin{align*}
				\sum_{p \in P} \eucl{p}{q}^2 ={} & \sum_{p \in P} \normlzs{p - q} = \sum_{p \in P} \normlzs{p - \centroid{P} + \centroid{P} - q} \\
				={} & \sum_{p \in P} \scalarm{p-\centroid{P}+\centroid{P}-q}{p-\centroid{P}+\centroid{P}-q} \\
				={} & \sum_{p \in P} \left[\normlzs{p - \centroid{P}} + 2 \scalarm{p-\centroid{P}}{\centroid{P}-q} + \normlzs{\centroid{P}-q}\right] \\
				={} & \vert P \vert \cdot \normlzs{\centroid{P}-q} + \sum_{p \in P} \normlzs{p - \centroid{P}} + 2 \cdot \sum_{p \in P} \scalarm{p - \centroid{P}}{\centroid{P}-q} \\
				={} & \vert P \vert \cdot \normlzs{\centroid{P}-q} + \sum_{p \in P} \normlzs{p - \centroid{P}} + 2 \cdot (\centroid{P}-q)^T  \underbrace{\left( \sum_{p \in P} p - \sum_{p \in P} \centroid{P} \right)}_{=0} \\
				={} & \vert P \vert \cdot \normlzs{\centroid{P}-q} + \sum_{p \in P} \normlzs{p - \centroid{P}} = \vert P \vert \cdot \eucl{\centroid{P}}{q}^2 + \sum_{p \in P} \eucl{p}{\centroid{P}}^2 
			\end{align*}
			The claim follows by setting $q = \centroid{P}$.
		\end{proof}
		Line segments are the first objects we use that go beyond points.
		\begin{definition}[line segment, \conferre{\citet{prasolov_geometry}}]
			A line segment is a set of points $\overline{pq} \eqdef \{(1-\gamma) \cdot p + \gamma\cdot q \mid \gamma \in [0,1] \}$ with endpoints $p, q \in \euclideanspace^d$.
		\end{definition}
		Curves are the central objects in this work. We are solely interested in curves that are composed of line segments, \ie \emph{polygonal} curves which have the $d$-dimensional euclidean space as so called ambient space.
		\begin{definition}[curve, \conferre{\citet[Definition 1]{alt_godau_frechet}}]
			A (parameterized) curve is a continuous mapping $\tau \colon [a,b] \rightarrow \euclideanspace^d$, where $[a,b] \subseteq \R_{\geq 0}$. A polygonal curve is a curve $\tau$, where exist $t_1, \dots, t_m \in [a,b]$ with $t_1 \leq \dots \leq t_m$ and $t_1 = a$, $t_m = b$ such that $ \bigcup\limits_{t \in [a,b]} \{\tau(t)\} = \bigcup\limits_{i=1}^{m-1} e_i$, where the $e_i \eqdef \overline{v_i v_{i+1}}$ are line segments, which we call the edges of the curve and the $v_i \eqdef \tau(t_i)$ are the vertices. We call $m$ the complexity of $\tau$, denoted by $\vert \tau \vert$. The equivalence class of polygonal curves with complexity at least $2$ and at most $m$ is denoted \eqcfre{m}.
		\end{definition}
		For a curve $\tau \in \eqcfre{m}$ we call $\tau(0)$ the initial point of the curve and $\tau(1)$ the end point of the curve.

		Grids are objects that can easily be used to employ a simple scheme for partitioning a set of points. 
		\begin{definition}[grid, \conferre{\citet{har_peled_geo_approx}}]
			\label{def:grid}
			Let $p \eqdef (p_1, \dots, p_d) \in \euclideanspace^d$ be a point. The cube of edge length $l \in \R_{\geq 0}$ around $p$ is defined $\cube{p, l} \eqdef \{q \eqdef (q_1, \dots, q_d) \in \euclideanspace^d \mid \max\limits_{i=1}^d \vert p_i - q_i \vert = \frac{l}{2} \}$. It has volume $l^d$. The grid of cell length $m \in \R_{\geq 0}$ with respect to $\mathfrak{C} \eqdef \cube{p, l}$ is defined \[\grid{\mathfrak{C}, m} \eqdef \bigcup\limits_{i_1, \dots, i_d \in \{ -k, k \mid k \in \{ 1, \dots, \lceil \nicefrac{l}{2m} \rceil \} \}} \cell{\mathfrak{G}}{i_1, \dots, i_d},\] where \[\cell{\mathfrak{G}}{i_1, \dots, i_d} \eqdef \{ (q_1, \dots, q_d) \in \euclideanspace^d \mid \forall j \in \{1, \dots, d\}: q_j \in \range{p_j + (i_j - \nicefrac{i_j}{\vert i_j \vert}) \cdot m}{p_j + i_j \cdot m} \}\] is the grid cell of $\mathfrak{G} \eqdef \grid{\mathfrak{C}, m}$ with id $(i_1, \dots, i_d)$ and \[
				\range{x}{y} \eqdef
				\begin{cases}
					[x,y], & \text{if } x \leq y \\
					[y,x], & \text{else}
				\end{cases}
			\] is the range enclosed by $x$ and $y$.
		\end{definition}
		Note that every cell of a grid is a cube itself. The following observations lie near:
		\begin{observation}
			\label{obs:cube_grid}
			Let $p \in \euclideanspace^d$ be a point and $l \in \R_{\geq 0}$. By \cref{def:grid} it holds that $\mathfrak{C} \eqdef \cube{p, l} = \grid{\mathfrak{C}, m}$ for any $m \in \R_{\geq 0}$.
		\end{observation}
		\begin{observation}
			\label{obs:ball_grid}
			Let $p \in \euclideanspace^d$ be a point and $B \eqdef \{ q \in \euclideanspace^d \mid \eucl{p}{q} \leq r \}$ be the closed ball with radius $r \in \R_{\geq 0}$ around $p$. By \cref{def:grid} it holds that $B \subseteq \cube{p, 2r}$.
		\end{observation}
		Usually balls and spheres are defined as sets of points with respect to the euclidean distance. Here we define a ball to be a set of curves with respect to the Fréchet distance, which is formally introduced in the following subsection.
		\begin{definition}[ball]
			\label{def:ball}
			For $\tau \in \eqcfre{m}$ and $r \in \R_{\geq 0}$, we define \[\ball{\tau}{r} \eqdef \{ \tau^\prime \in \eqcfre{m} \mid \frechet{\tau}{\tau^\prime} \leq r \}.\]
		\end{definition}
		\subsection{Similarity Measures for Curves}
		Naturally, to be able to compare a set of curves among each other we need a notion of distance for curves. There are similar of these notions. We will work with the Fréchet distance which is a similarity measure that yields good results for most applications, because it takes the location and the ordering of the points on the curves into account.
		\begin{figure}
			\centering
			\def\svgwidth{\textwidth}
\begingroup%
  \makeatletter%
  \providecommand\color[2][]{%
    \errmessage{(Inkscape) Color is used for the text in Inkscape, but the package 'color.sty' is not loaded}%
    \renewcommand\color[2][]{}%
  }%
  \providecommand\transparent[1]{%
    \errmessage{(Inkscape) Transparency is used (non-zero) for the text in Inkscape, but the package 'transparent.sty' is not loaded}%
    \renewcommand\transparent[1]{}%
  }%
  \providecommand\rotatebox[2]{#2}%
  \newcommand*\fsize{\dimexpr\f@size pt\relax}%
  \newcommand*\lineheight[1]{\fontsize{\fsize}{#1\fsize}\selectfont}%
  \ifx\svgwidth\undefined%
    \setlength{\unitlength}{603.69920513bp}%
    \ifx\svgscale\undefined%
      \relax%
    \else%
      \setlength{\unitlength}{\unitlength * \real{\svgscale}}%
    \fi%
  \else%
    \setlength{\unitlength}{\svgwidth}%
  \fi%
  \global\let\svgwidth\undefined%
  \global\let\svgscale\undefined%
  \makeatother%
  \begin{picture}(1,0.11173834)%
    \lineheight{1}%
    \setlength\tabcolsep{0pt}%
    \put(0,0){\includegraphics[width=\unitlength]{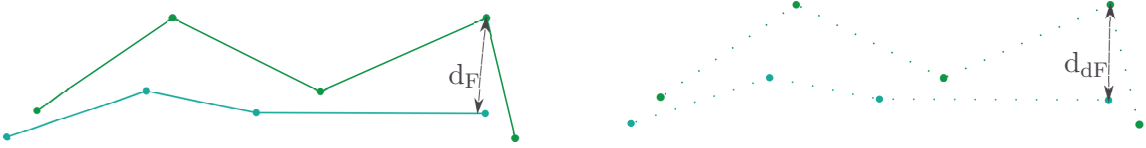}}%
    \put(0.60390429,0.13099974){\color[rgb]{0,0,0}\makebox(0,0)[lt]{\begin{minipage}{0.83713353\unitlength}\raggedright \end{minipage}}}%
    \put(0.39,0.05209105){\color[rgb]{0.28235294,0.2627451,0.28627451}\makebox(0,0)[lt]{\lineheight{1.25}\smash{\begin{tabular}[t]{l}$\pfrechet$\end{tabular}}}}%
    \put(0.93,0.06222527){\color[rgb]{0.28235294,0.2627451,0.28627451}\makebox(0,0)[lt]{\lineheight{1.25}\smash{\begin{tabular}[t]{l}$\pdfrechet$\end{tabular}}}}%
  \end{picture}%
\endgroup%

			\caption{Comparison of the (continuous) Fréchet distance to the discrete Fréchet distance. The Fréchet distance is usually visualized by a man walking the first curve and his dog, who is on a leash, walking the second one. The man varies his speed such that the maximum leash length is minimal throughout the walk. This length equals the Fréchet distance between the curves. With respect to the discrete Fréchet distance both can only hop from one vertex of their curve to another. In both cases the man and the dog are not allowed to go back to a point where they have already been. Standing still is allowed, though.}
			\label{fig:frechet_comp}
		\end{figure}
		\begin{definition}[Fréchet distance, \conferre{\citet[Definition 2]{alt_godau_frechet}}]
			\label{def:frechet_distance}
			Let $f$ be a continuous injective mapping. We call $f$ a reparamterization. 
			
			Let $\FF$ denote the set of non-decreasing reparameterizations $f\colon [0,1] \rightarrow [0,1]$ with $f(0) = 0$ and $f(1) = 1$. The Fréchet distance between two curves $\tau$ and $\sigma$ is defined as follows:
			\[ \frechet{\tau}{\sigma} \eqdef \inf_{f \in \FF} \max_{t \in [0,1]} \eucl{\tau(f(t))}{\sigma(t)}. \]
		\end{definition}
		If the domains of $\tau$ and $\sigma$ are not equal to $[0,1]$ (assume both have domain $[a,b]$), let $f\colon [0,1] \rightarrow [a,b]$ be a non-decreasing reparameterization, we obtain $\tau^\prime\colon [0,1] \rightarrow \euclideanspace^d, x \mapsto \tau(f(x))$ (resp. $\sigma^\prime \colon [0, 1] \rightarrow \euclideanspace^d, x \mapsto \sigma(f(x))$), which have the same Fréchet distance to each other, as $\tau$ and $\sigma$, because the Fréchet distance is invariant under reparameterizations. A visualization of the Fréchet distance is depicted in \cref{fig:frechet_comp}.
		
		Because the reparameterizations $f \in \FF$ have the restriction that $f(0) = 0$ and $f(1) = 1$ the following observation is immediate:
		\begin{observation}[{\citet[Observation 2.2]{clustering_time_series}}]
			\label{obs:frechet_endpoints}
			For every $\tau_1, \tau_2 \in \eqcfre{m}$ it holds that $\frechet{\tau_1}{\tau_2} \geq \max\{\eucl{\tau_1(0)}{\tau_2(0)}, \eucl{\tau_1(1)}{\tau_2(1)}\}$. It immediately follows that the Fréchet distance between two line segments $\overline{pq}, \overline{rs}$ is $\max\{ \eucl{p}{r}, \eucl{q}{s} \}$.
		\end{observation}
		We already stated that the curves of interest are composed of line segments, therefore it is beneficial to have a definition that enables us to assemble curves by concatenation.
		\begin{definition}[concatenation, \conferre {\citet[Definition 2.1]{clustering_time_series}}]
			Let $\tau_1\colon [a_1,b_1] \rightarrow \R^d$ and $\tau_2\colon [a_2,b_2] \rightarrow \R^d$ be curves with $0 \leq a_1 \leq b_1 \leq 1$ and $0 \leq a_2 \leq b_2 \leq 1$ and $\tau_1(b_1) = \tau_2(a_2)$. We call $\tau \colon [0,1] \rightarrow \R^d$ with
			\begin{align*}
				\tau(t) \eqdef (\tau_1 \oplus \tau_2)(t) \eqdef 
				\begin{cases}
					\tau_1(a_1 + (b_1 - a_1 + b_2 - a_2) \cdot t), & \text{if } t \leq \frac{b_1 - a_1}{b_1-a_1+b_2-a_2} \\
					\tau_2(b_2-(b_1-a_1+b_2-a_2)\cdot (1-t)), & \text{else}
				\end{cases}
			\end{align*}
			the \emph{concatenation} of $\tau_1$ and $\tau_2$.
		\end{definition}
		If two curves are the concatenations of each two respective (sub-)curves we can bound the Fréchet distance between the curves from the Fréchet distances of the respective sub-curves.
		\begin{observation}[{\citet[Observation 2.1]{clustering_time_series}}]
			\label{obs:subcurves}
			Let $\tau_1 \colon [a_1, b_1] \rightarrow \R^d$, $\tau_2 \colon [a_2, b_2] \rightarrow \R^d$ be curves with $\tau_1(b_1) = \tau_2(a_2)$ and $\tau \eqdef \tau_1 \oplus \tau_2$ be their concatenation. Also let $\sigma_1 \colon [a^\prime_1, b^\prime_1] \rightarrow \R^d$, $\sigma_2 \colon [a^\prime_2, b^\prime_2] \rightarrow \R^d$ be curves with $\sigma_1(b^\prime_1) = \sigma_2(a^\prime_2)$ and $\sigma \eqdef \sigma_1 \oplus \sigma_2$ be their concatenation. It holds that \[ \frechet{\tau}{\sigma} \leq \max\{\frechet{\tau_1}{\sigma_1}, \frechet{\tau_2}{\sigma_2}\}. \]
		\end{observation}
		This is similar if the curves are concatenations of more than two sub-curves.
		\begin{proposition}
			\label{prop:frechet_line_conc}
			Let \[\tau_1 \colon [t_0,t_1] \rightarrow \R^d, \dots, \tau_n \colon [t_{n-1},t_n] \rightarrow \R^d\] and \[\sigma_1 \colon [t^\prime_0,t^\prime_1] \rightarrow \R^d, \dots, \sigma_n \colon [t^\prime_{n-1},t^\prime_n] \rightarrow \R^d\] be curves with $\tau_{i-1}(t_{i-1}) = \tau_i(t_{i-1})$ and $\sigma_{i-1}(t^\prime_{i-1}) = \sigma_i(t^\prime_{i-1})$ for all $i = 2, \dots, n$ and $0 \leq t_{i-1} \leq t_i \leq 1$, as well as $0 \leq t^\prime_{i-1} \leq t^\prime_i \leq$ for $i = 1, \dots, n$.  Also for $i = 1, \dots, n$ let $f_{\tau, i}(t) \colon [0,1] \rightarrow [t_{i-1},t_i]$ and $f_{\sigma, i}(t)\colon [0,1] \rightarrow [t^\prime_{i-1}, t^\prime_i]$ be non-decreasing reparameterizations.
			
			If for all $i = 1, \dots, n$ it holds that $\frechet{\tau_i \circ f_{\tau,i}}{\sigma_i \circ f_{\sigma, i}} \leq a$, for a fixed $a \in \R_{\geq 0}$, then \[\frechet{\tau_1 \oplus \dots \oplus \tau_n}{\sigma_1 \oplus \dots \oplus \sigma_n} \leq a.\]
		\end{proposition}
		\begin{proof}
			We prove the claim by induction over $n \in \N$:
			\paragraph{Base case:} When $n = 1$ the claim follows by \cref{def:frechet_distance}.
			\paragraph{Induction step:} Let $k \in \N$ be given and assume the claim holds for $n = k$, therefore $\frechet{\tau_1 \oplus \dots \oplus \tau_k}{\sigma_1 \oplus \dots \oplus \sigma_k} \leq a$.  Now for $k+1$ we know that $\frechet{\tau_{k+1}}{\sigma_{k+1}} \leq a$ by assumption. Let $\tau_{1,k} \eqdef \tau_1 \oplus \dots \oplus \tau_k$ and $\sigma_{1, k} \eqdef \sigma_1 \oplus \dots \oplus \sigma_k$, by \cref{obs:subcurves} we obtain
			\begin{align*}
				\frechet{\tau_{1,k} \oplus \tau_{k+1}}{\sigma_{1, k} \oplus \sigma_{k+1}} \leq{} & \max\{ \frechet{\tau_1 \oplus \dots \oplus \tau_k}{\sigma_1 \oplus \dots \oplus \sigma_k}, \frechet{\tau_{k+1}}{\sigma_{k+1}} \} \\
				\leq{} & a,
			\end{align*}
			hence the claim holds for $n = k+1$, therefore it holds for all $n \in \N$ by induction.
		\end{proof}
		For the curves of interest, \ie polygonal curves, we can efficiently compute the Fréchet distance between two curves, though it is costly in terms of running-time\footnote{We are dealing with big data, thus every running-time dependency that is above linear is considered expensive.}.
		\begin{theorem}[{\citet[Theorem 6]{alt_godau_frechet}}]
			\label{theo:frechet_algo}
			For two polygonal curves $\tau, \sigma \in \eqcfre{m}$ the Fréchet distance $\frechet{\tau}{\sigma}$ can be computed in time $\On{m^2 \log(m)}$.
		\end{theorem}
		Especially for the examples that are depicted in this work we employ the discrete Fréchet distance, which can be computed more easily. Also, it has a useful relationship to the Fréchet distance.
		\begin{definition}[discrete Fréchet distance, \conferre{\citet{eiter_mannila, discrete_frechet_subquadratic}}]
			Let $P \eqdef p_1, \dots, p_m$ and $Q \eqdef q_1, \dots, q_n$ be sequences of points in $\euclideanspace^d$. Let $\delta \in \R_{\geq 0}$ be arbitrary. Let $G_{\delta} \eqdef (V, E_\delta)$ denote the directed graph with vertex-set $ V \eqdef \left(\bigcup\limits_{i=1}^m \{p_i\}\right) \times \left(\bigcup\limits_{i=1}^n \{q_i\}\right)$ and edge-set 
			\begin{align*}
			E_\delta \eqdef & \{ \left( (p_i, q_j), (p_{i+1}, q_j)\right) \mid \eucl{p_i}{q_j} \leq \delta \wedge \eucl{p_{i+1}}{q_j} \leq \delta \} \\ 
			& \cup \{ \left( (p_i, q_j), (p_{i}, q_{j+1})\right) \mid \eucl{p_i}{q_j} \leq \delta \wedge \eucl{p_{i}}{q_{j+1}} \leq \delta \}.
			\end{align*}
			The discrete Fréchet distance is defined: \[ \dfrechet{P}{Q} = \delta \aqdef \Reach((p_1,q_1), (p_m,q_n), \delta) \wedge \forall \delta^\prime > \delta: \lnot \Reach((p_1,q_1), (p_m,q_n), \delta^\prime) \] 
			Here $\Reach(x,y,\delta)$ is a tertiary predicate, that is true, iff there exists a path from vertex $x$ to vertex $y$ in $G_\delta$.
		\end{definition}
		A visualization of the discrete Fréchet distance is depicted in \cref{fig:frechet_comp}.
		
		The discrete Fréchet distance yields an upper bound on the (continuous) Fréchet distance:
		\begin{proposition}[{\citet[Lemma 2]{eiter_mannila}}]
			For all polygonal curves $\tau, \sigma \in \eqcfre{m}$ with vertices $V_\tau \eqdef v_{\tau, 1}, \dots, v_{\tau, m_{\tau}}$ and $V_{\sigma} \eqdef v_{\sigma, 1}, \dots, v_{\sigma, m_{\sigma}}$ it holds that $\frechet{\tau}{\sigma} \leq \dfrechet{V_{\tau}}{V_{\sigma}}$.
		\end{proposition}
	\subsection{Clustering}
		Clustering is the task of partitioning a set of objects with respect to optimizing some objective function. Here we consider a set of curves $\Tau \subset \eqcfre{m}$ as input and three different objective functions that are to be minimized. In particular the ``nearest''-object relationship is of main interest. In this work all objective functions are defined with respect to this concept. The following function formalizes this relationship.
		\begin{definition}
			\label{def:nearest_func}
			For a curve $\tau \in \eqcfre{m}$ and a set of curves $C \subseteq \eqcfre{m}$, the function \[ \eta(\tau, C) \eqdef \argmin\limits_{c \in C} \frechet{\tau}{c} \] returns an arbitrary but fixed nearest neighbor (especially if there is more than one) of $\tau$ in $C$.
		\end{definition} 
		All objective functions are binary. They depend on the set of curves and on a set $C \subset \eqcfre{l}$ or $C \subseteq \Tau$, of cardinality $k$, that induces the partition. The members of the partition are called \emph{clusters} and the members of $C$ are called \emph{centers}, which are the representatives of the clusters.
		\begin{definition}
			\label{def:cluster}
			The cluster $\pcluster$ of a center $c \in C$ with respect to a set of curves $\Tau$ is defined \[\cluster{\Tau}{C}{c} \eqdef \{ \tau \in \Tau \mid \nearest{\tau}{C} = c \}. \]
		\end{definition}
		Now that we have defined the basics of clustering we define the objective functions that will be in the focus of this work.
		\subsubsection{Clustering Objectives}
		The first objective function we consider is the one that is most fuzzy, in the sense that for a given set of curves there can be \emph{many} center-sets for which the objective function has equal value. Specifically only the farthest curve to a respective nearest center determines the value of the objective function.
		\begin{figure}
			\centering
			\includegraphics[width=\textwidth]{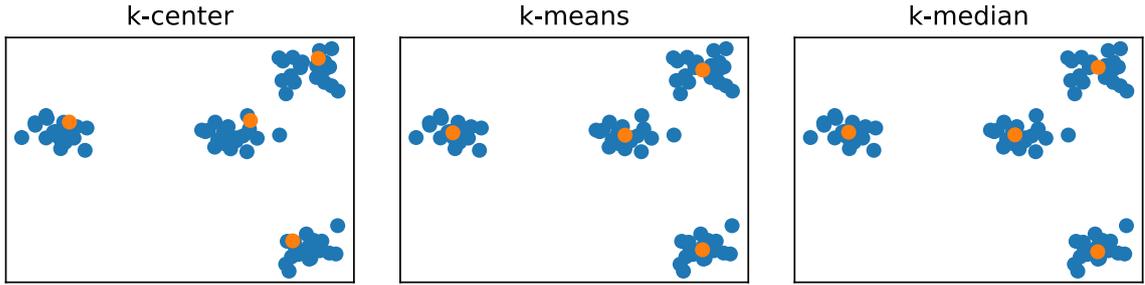}
			\caption{Comparison of approximate solutions for different clustering objectives with respect to point-sets. It can be observed that $k$-means is most sensitive to outliers, while $k$-median is less sensitive. $k$-center is least sensitive to outliers, but most efficient to approximate. In fact one has to look carefully to see the little differences between the centers with respect to the means objective and the centers with respect to the median objective.}
			\label{fig:cluster_comparison}
		\end{figure}
		\begin{definition}[\klcenter objective]
			\label{def:klcenter}
			Given a set of polygonal curves $\Tau \eqdef \{ \tau_1, \dots, \tau_n \} \subset \eqcfre{m}$ and $k, l \in \mathbb{N}_{>0}$, return the optimal cost, \ie $\optcost{\Tau} \eqdef \min\limits_{C \subset \eqcfre{l}, \vert C \vert = k} \cost{\Tau}{C}$, of a clustering with \kk\ centers of complexity at most $\el$, where $ \cost{\Tau}{C} \eqdef \max\limits_{\tau \in \Tau} \frechet{\tau}{\nearest{\tau}{C}}$.
		\end{definition}
		In contrast to the \klcenter objective, the \klmedian objective takes all distances from the curves to the respective nearest centers into account. To be precise, here we are working with the \emph{discrete} median objective because, different to the \klcenter objective and the \klmeans objective, the possible center-sets stem from a countably finite set, the input set. Generally it is easier to obtain a clustering with respect to the \klcenter objective than with respect to the \klmedian objective.
		\begin{definition}[\klmedian objective]
			\label{def:klmedian}
			Given a set of polygonal curves $\Tau \eqdef \{ \tau_1, \dots, \tau_n \} \subset \eqcfre{m}$ and $k \in \mathbb{N}_{>0}$, return the optimal cost, \ie $\optcost{\Tau} \eqdef \min\limits_{C \subseteq \Tau, \vert C \vert = k} \cost{\Tau}{C}$, of a clustering with \kk\ centers of complexity at most \el, where $ \cost{\Tau}{C} \eqdef \sum\limits_{\tau \in \Tau} \frechet{\tau}{\nearest{\tau}{C}}$.
		\end{definition}
		The following definition is different from \cref{def:klmedian} in the sense that curves that lie further from the respective nearest center are more strongly taken into account while the nearer curves contribute less strongly to the value of the objective function.
		\begin{definition}[\klmeans objective]
			\label{def:klmeans}
			Given a set of polygonal curves $\Tau \eqdef \{ \tau_1, \dots, \tau_n \} \subset \eqcfre{m}$ and $k, l \in \mathbb{N}_{>0}$, return the optimal cost, \ie $\optcost{\Tau} \eqdef \min\limits_{C \subset \eqcfre{l}, \vert C \vert = k} \cost{\Tau}{C}$, of a clustering with \kk\ centers of complexity at most $\el$, where $ \cost{\Tau}{C} \eqdef \sum\limits_{\tau \in \Tau} \frechet{\tau}{\nearest{\tau}{C}}^2$.
		\end{definition} 
		In particular, with respect to point-sets the \klmeans objective has an advantage over the \klmedian objective, because the optimal $1$-means can be determined by a simple calculation, \conferre \cref{def:centroid}. In \cref{fig:cluster_comparison} a comparison of the different clustering objectives with respect to point-sets is depicted.
	\section{$\epsilon$-coresets}
		Let $\Tau \subset \eqcfre{m}$ be a set of curves of cardinality $n$. In this work we assume that $n$ is a large number, such that any algorithms that work on $\Tau$ have high running-time. We assume that especially algorithms of super-polynomial running-time are not efficient enough to examine $\Tau$ in reasonable time. 
		
		If there is no known polynomial-time approximation scheme for a geometric problem then there may be the possibility of reducing the input-set to an \coreset and analyzing the \coreset instead of the input-set. Generally speaking, an \coreset approximates the ``shape'' of a set of objects up to an $\epsilon$-fraction. If the \coreset has cardinality dependent only on $\poly{\frac{1}{\epsilon}}$, \ie a polynomial of $\nicefrac{1}{\epsilon}$, and independent of $n$, we can use an exact algorithm on the \coreset which then has running-time $\On{\poly{\nicefrac{1}{\epsilon}}}$. In such a case it may even be beneficial to analyze the \coreset by brute force. If the \coreset can be constructed in polynomial time this yields a polynomial-time approximation scheme. If the \coreset can be constructed in linear time, we even obtain a pseudo-linear-time approximation scheme.
		\begin{definition}[{\coreset[] for curves}, \conferre{\citet[Definition 1]{coresers_methods_history}}]
			\label{def:coreset}
			Let $A$ be the set of finite sets of curves $\tau \in \eqcfre{m}$ and $B$ the set of sets, of cardinality $k$, of curves $\sigma \in \eqcfre{m}$, also called the candidate solutions. Further, let $f\colon A \times B \rightarrow \R_{\geq 0}$ be a non-negative measurable function. Then $S \in A$ is an \coreset for $\Tau \in A$ with respect to $f$, if $\forall C \in B:$
			\[ (1-\epsilon) f(\Tau,C) \leq f(S,C) \leq (1+\epsilon) f(\Tau,C) \]
			Further a set $S \in A$ with weights $w \colon S \rightarrow \R$ is a weighted \coreset for $\Tau \in A$, if $\forall C \in B:$ 
			\[ (1-\epsilon) f(\Tau,C) \leq g[w, S, C] \leq (1+\epsilon) f(\Tau, C) \]
			Here $g$ is obtained through $f$ by taking the weights $w$ into account.
		\end{definition}
		In the setting of this work, the set $A$ is the set of all possible input-sets and $B$ is the set of all possible center-sets. The function $f$ is the clustering objective at hand. The big benefit of \coreset[s] is, that they give a guarantee for every possible center-set, respective for one fixed input-set.
	\section{Probability Theory}
		Previously we stated that it is easier to obtain a clustering with respect to the \klcenter objective than with respect to the \klmedian objective. The same holds for \coreset constructions, though probability theory yields tools which make it possible to obtain \coreset[s] \emph{with constant probability} and reasonable effort.  
		
		We start by defining the basic notions of probability theory.
		\begin{definition}[\conferre\citet{mitzenmacher}]
			\label{def:probability_space}
			Let $ \Omega $ be a sample space, which is the set of all possible outcomes of a random process and let $ \mathcal{E} $ be the set of sets that represent the allowable events in $ \Omega $, so it holds that \( \forall E \in \mathcal{E}: E \subseteq \Omega \). Further let \( \Pr\colon \mathcal{E} \rightarrow \mathbbm{R} \) be a probability function, \ie it satisfies all the following conditions:
			\begin{itemize} 
				\item \( \forall E \in \mathcal{E} : 0 \leq \Pr[E] \leq 1 \) 
				\item \( \Pr[\Omega] = 1 \)
				\item for any finite or countably infinite sequence of pairwise mutually disjoint events \( E_1, E_2, E_3, \dots \) it holds that \( \Pr\left[ \bigcup\limits_{i \in \{1, 2, 3, \dots\}} E_i \right] = \sum\limits_{i \in \{1, 2, 3, \dots\}} \Pr[E_i] \)
			\end{itemize}
			The triplet \( (\Omega, \mathcal{E}, \Pr) \) defines a probability space, which is the basis for every random process.
		\end{definition}
		We call the $E \in \Omega$ simple events. Let $ E \in \mathcal{E} $ be an event, the contrary event of $ E $ is $ E^C = \Omega \backslash E $, its complement, with probability \( \Pr[E^C] = 1 - \Pr[E] \). Repeatedly choosing elements of the sample space, according to the given distribution, \ie the given probability function, we call sampling. The resulting set of this sampling-process is denoted a sample.
		
		If the occurrence of two arbitrary events does not influence each other's probability, we call those events independent. In particular two events are independent, if they occur in distinct samples.
		\begin{definition}[\conferre\citet{mitzenmacher}]
			\label{def:independence}
			Two events $ E_1 $ and $ E_2 $ are independent, iff \( \Pr[E_1 \cap E_2] = \Pr[E_1] \cdot \Pr[E_2] \).
		\end{definition}
		The union bound is a simple yet strong and widely employed tool to obtain upper bounds on the occurrence of general events.
		\begin{proposition}[union bound, {\citet[Lemma 1.2]{mitzenmacher}}]
			\label{prop:union_bound}
			Let \( E_1, \dots, E_n \) be arbitrary events in $ \mathcal{E} $, then it holds that \( \Pr \left[ \bigcup_{i=1}^n E_i \right] \leq \sum_{i=1}^{n} \Pr[E_i] \).
		\end{proposition}
		A benefit of probability spaces are random variables. Here the value of the variable is determined by a sample from the probability space.
		\begin{definition}[random variable, \conferre\citet{mitzenmacher}]
			\label{def:random_variable}
			A real random variable \( X \colon \Omega \rightarrow \R \) is a function from a sample space to the real numbers.
		\end{definition}
		The expected value of a random variable over a discrete probability space is the sum of all possible outcomes with respect to their probabilities.
		\begin{definition}[expected value, \conferre{\citet{mitzenmacher}}]
			\label{def:expected_value}
			Let $ X $ be a real random variable that takes the values \( x_1, \dots, x_n \) with probabilities \( p_1, \dots, p_n \). The expected value of $X$ is defined as
			\begin{equation*}
			\E[X] \eqdef \sum_{i=1}^n x_i \cdot p_i.
			\end{equation*}
		\end{definition}
		The expected squared deviation of $ X $ from its expected value is its variance.
		\begin{definition}[variance, \conferre\citet{mitzenmacher}]
			\label{def:variance}
			Let $X$ be a random variable. Its variance is defined as
			\begin{equation*}
			\Var[X] \eqdef \E\left[(X-\E[X])^2\right].
			\end{equation*}
		\end{definition}
		When we are using more than one probability space, e.g. $(\Omega, \mathcal{E}_\varphi,\varphi)$ and $(\Omega, \mathcal{E}_\psi, \psi)$, we may write $\E_{\psi}[X]$ or $\Var_\varphi[Y]$ to emphasize the underlying probability distribution.
		
		Consider a function $f$ that maps a set of sets $\mathcal{S} \eqdef \{ S_1, S_2, S_3, \dots \}$ to the real numbers. We may have the case that every $S \in \mathcal{S}$ has very high cardinality, such that we can not efficiently evaluate $f$. We have the possibility to set $\Omega \eqdef S \in \mathcal{S}$ and use a random variable $X$, with $\E[X] = f(S)$, as so called estimator for $f$.
		\subsection{Concentration Inequalities}
		In this work we use so called tail-inequalities to determine the probability that a given random variable deviates from its expected value more than a certain constant.

		\begin{theorem}[Markov's inequality, {\citet[Theorem 3.1]{mitzenmacher}}]
			\label{theo:markov}
			Let $ X $ be a random variable that assumes only non-negative values. Then for $ a \in \R_{>0}$:
			\begin{equation*}
			\Pr[X \geq a] \leq \frac{\E[X]}{a}.
			\end{equation*}
		\end{theorem} 
		\begin{theorem}[Hoeffding's inequality, {\citet[Theorem 2]{hoeff}}]
			\label{theo:hoeff}
			Let \( X_1, \dots , X_\ell \) be independent random variables and let \( \bar{X} = \frac{1}{\ell} (X_1 + \dots + X_\ell) \) be their mean.
			
			If \( \forall i \in \{ 1, \dots, \ell\} \exists a_i, b_i \in \R : a_i \leq X_i \leq b_i\), then for \( t \in \R_{\geq 0} \):
			\begin{equation*}
			\Pr \left [ \bar{X} - \E[\bar{X}] \geq t \right ] \leq \exp \left ({ \frac{ -2 \ell^2 t^2 }{ \sum_{i=1}^\ell (b_i - a_i)^2 } } \right ).
			\end{equation*}
			Similarly, for $- \bar{X} + \E[X]$ this gives an upper bound for \( \Pr[ -\bar{X} + \E[\bar{X}] \geq t ] \), which implies the upper bound: \[ \Pr[ \vert{\bar{X}-\E[\bar{X}]}\vert \geq t ] \leq 2 \exp \left({ \frac{ -2 \ell^2 t^2 }{ \sum_{i=1}^\ell (b_i - a_i)^2 } } \right).\]
		\end{theorem}
		\begin{theorem}[Bernstein's inequality, {\citet{bernstein}}]
			\label{theo:bern}
			Let \( X_1, \dots , X_\ell \) be independent random variables with $\E[X_i] = 0$ for all $i=1,\dots, \ell$ and $\vert X_i \vert \leq M$ almost surely for all $i=1,\dots, \ell$, then for $t \in \R_{> 0}$: \[ \Pr\left[\sum_{i=1}^\ell X_i > t\right] \leq \exp\left(\frac{-t^2}{2\sum_{i=1}^\ell \E[X_i^2] + \nicefrac{2}{3}Mt}\right). \]
			For any independent random variables \( Y_1, \dots, Y_\ell \) we can apply this inequality to $Y_i - \E[Y_i]$ and $-Y_i + \E[Y_i]$, if the absolute value of these random variables are bounded, since they have zero mean. We get
			\[ \Pr\left[\left\vert \sum_{i=1}^\ell Y_i - \E[Y_i]\right\vert > t\right] \leq 2 \exp\left(\frac{-t^2}{2 \sum_{i=1}^{\ell} \Var[Y_i] + \nicefrac{2}{3}Mt}\right).\] 
		\end{theorem}
	\subsection{Sensitivity Sampling}
	\label{subsec:sensitivity}
	Sometimes it is useful to transform one probability space into one other. For example a probability space endowed with the uniform distribution has many benefits in designing estimators for various tasks. Sometimes these estimators may have high variance, such that a sample needs to have very high cardinality for the estimator to be precise. Especially in clustering, so called outliers, \ie objects that lie far from most other objects, can introduce a lot of error when uniform sampling schemes are applied.
	
	Following the work of \citet{langbergschulman}, for a probability space $(\Omega, \mathcal{E}, \varphi)$ and some non-negative random variables $X_1, \dots, X_w$, every $\omega \in \Omega$ is assigned a sensitivity value, which expresses the highest contribution that $\omega$ has on any $\E_\varphi[X_i]$ for $i=1,\dots,w$. The set of random variables $W \eqdef \{ X_1, \dots, X_w\}$ is then assigned a total sensitivity value based on the sensitivity values of the $\omega \in \Omega$.
	
	\citeauthor{langbergschulman} show that these values suffice to construct a probability function $\psi\colon \mathcal{E} \rightarrow \R$, such that the random variables $Y_i(\omega) \eqdef X_i(\omega) \frac{\varphi(\omega)}{\psi(\omega)}$, for $i \in \{1, \dots, w \}$, have the same expected values under $\psi$ as the $X_i$ under $\varphi$ and have low variance under $\psi$. Hence, they are precise for a relatively small sample from $(\Omega, \mathcal{E}, \psi )$.
	
	\begin{definition}[\conferre\citet{langbergschulman}]
		\label{def:sensitivity}
		The sensitivity for $\omega \in \Omega$ with respect to $W$ is defined \[\xi(\omega) \eqdef \max\limits_{i \in \{1,\dots,w\}} \frac{X_i(\omega)}{\E_\varphi[X_i]}.\]
		The total sensitivity of $W$ is defined \[ \Xi \eqdef \sum_{\omega \in \Omega} \xi(\omega) \varphi(\omega). \]
		Since $\xi(\cdot)$ is sometimes very hard to calculate we will use an upper bound $\lambda(\cdot)$ on $\xi(\cdot)$ and an upper bound \[\Lambda \eqdef \sum_{\omega \in \Omega} \lambda(\omega) \varphi(\omega) \] on $\Xi$ to construct the probability function $\psi$, which is then defined for $\omega \in \Omega$:
		\[ \psi(\omega) \eqdef \frac{\lambda(\omega)}{\Lambda} \varphi(\omega). \]
		For $E \in \mathcal{E}$ we define $\psi(E) \eqdef \sum_{\omega \in E} \psi(\omega)$.
	\end{definition}
	First we show that $\psi$ indeed is a probability function.
	\begin{proposition}
		\label{prop:sum_one}
		$\psi$ is a probability function for $\Omega$ and $\mathcal{E}$.
	\end{proposition}
	\begin{proof}
		Let $E \in \mathcal{E}$.
		\begin{align*}
			\psi(E) = \sum_{\omega \in E} \psi(\omega) = \sum\limits_{\omega \in E} \frac{\lambda(\omega)}{\Lambda} \varphi(\omega) = \frac{\sum\limits_{\omega \in E} \lambda(\omega) \varphi(\omega)}{\sum\limits_{\omega \in \Omega} \lambda(\omega) \varphi(\omega)} \leq 1
		\end{align*}
		Since all random variables are non-negative $\psi(E) \geq 0$. For $\Omega$ the calculation is similar:
		\begin{align*}
			\psi(\Omega) = \frac{\sum\limits_{\omega \in \Omega} \lambda(\omega) \varphi(\omega)}{\sum\limits_{\omega \in \Omega} \lambda(\omega) \varphi(\omega)} = 1
		\end{align*}
		The claim follows by these two arguments and the definition of $\psi$.
	\end{proof}
	Also, the expected values of $X_i$ under $\varphi$ and $Y_i$ under $\psi$ are equal for $i=1,\dots,w$.
	\begin{proposition}
		\label{prop:exp_equal}
		Let $X_i \in W$ be arbitrary. It holds that $\E_\varphi[X_i] = \E_{\psi}[Y_i]$.
	\end{proposition}
	\begin{proof}
		\begin{align*}
			\E_\psi[Y_i] = \sum_{\omega \in \Omega} Y_i(\omega) \psi(\omega) = \sum_{\omega \in \Omega} X_i(\omega) \frac{\varphi(\omega)}{\psi(\omega)} \psi(\omega) = \E_\varphi[X_i]
		\end{align*}
	\end{proof}
	Finally we show that the random variables $Y_1, \dots, Y_w$ have low variance under $\psi$.
	\begin{proposition}[\citet{langbergschulman}]
		\label{prop:var_sens}
		Let $X_i \in W$ be arbitrary. It holds that $\Var_\psi[Y_i] \leq (\Lambda-1) \cdot \E_\psi[Y_i]^2$.
	\end{proposition}
	\begin{proof}
		\begin{align*}
			\Var_\psi[Y_i] ={} & \sum_{\omega \in \Omega} \left( X_i(\omega) \frac{\varphi(\omega)}{\psi(\omega)} - \E_\psi[Y_i]  \right)^2 \cdot \psi(\omega) \\
			={} & \sum_{\omega \in \Omega} \left( X_i(\omega) \frac{\varphi(\omega)}{\psi(\omega)} - \E_\psi[Y_i]  \right)^2 \cdot \varphi(\omega) \frac{\lambda(\omega)}{\Lambda} \\
			={} & \sum_{\omega \in \Omega} \left[X_i(\omega)^2 \frac{\Lambda}{\lambda(\omega)} \varphi(\omega) - 2 X_i(\omega) \varphi(\omega) \E_\psi[Y_i] + \E_\psi[Y_i]^2 \varphi(\omega) \frac{\lambda(\omega)}{\Lambda}\right] \\
			={} & \sum_{\omega \in \Omega} \frac{\Lambda}{\lambda(\omega)} X_i(\omega)^2  \varphi(\omega) - 2 \E_\psi[Y_i] \underbrace{\sum_{\omega \in \Omega} X_i(\omega) \varphi(\omega)}_{=\E_\varphi[X_i]} + \E_\psi[Y_i]^2 \underbrace{\sum_{\omega \in \Omega} \frac{\lambda(\omega)}{\Lambda} \varphi(\omega)}_{=1} \\
			={} & \sum_{\omega \in \Omega} \frac{\Lambda}{\lambda(\omega)} X_i(\omega)^2  \varphi(\omega) -  \E_\psi[Y_i]^2 \\
			\leq{} & \sum_{\omega \in \Omega} \frac{\Lambda}{\lambda(\omega)} X_i(\omega) \cdot \E_\psi[Y_i] \cdot \lambda(\omega) \cdot \varphi(\omega) -\E_\psi[Y_i]^2 \\
			={} & (\Lambda - 1) \cdot \E_\psi[Y_i]^2
		\end{align*}
		The inequality is obtained through the fact that $\frac{X_i(\omega)}{\E_\varphi[X_i]} \leq \xi(\omega) \Leftrightarrow X_i(\omega) \leq \E_\varphi[X_i] \xi(\omega)$ and that $\lambda(\omega)$ is an upper bound for $\xi(\omega)$.
	\end{proof}
	\chapter{$\epsilon$-coresets for Clustering Objectives}
\label{chap:coresets}
Before we start investigating \coreset constructions, we seek to answer one obvious question: Can we directly apply techniques for constructing \coreset[s] of point-sets to curve sets?

Since a vast number of sophisticated techniques for constructing \coreset[s] of point sets exists, \ie random sampling, geometric decompositions (grids, balls, half-spaces, etc.), gradient descent as well as sketches and projections, \conferre \cite{coresers_methods_history}, it would be a great benefit to utilize these techniques when we are dealing with curve-sets. Here a first remark that comes to mind is that, even if this was possible, we would only be able to obtain the value of the clustering. To construct center-curves out of the center-points we would still need additional techniques -- these are even imaginable. 

\section{Why Transforming Curves into Points Fails}
Assume we are given a set of curves and want to apply one of the aforementioned techniques, without customizing or changing it in any way. A possibility that lies near: Transform the curves into points through an isometry, \ie a function that maps $(\Tau, \pfrechet)$ into $\euclideanspace^d$, maintaining the distances between the objects, \conferre \cref{def:isometry}. But is there a suitable isometry for every set of curves?

The answer is no, which is not surprising. Otherwise it would not make sense to deal with the Fréchet distance at all. The following proposition shows where the crucial point lies in the transformation.
\begin{proposition}
	\label{prop:embedding}
	For every integer $d \geq 1$, there exists a set $\Tau \eqdef \{\tau_1,\dots, \tau_n\} \subset \eqcfre{m}$ of curves, such that there exists no isometry $f\colon \Tau \rightarrow \euclideanspace^d$ that embeds $(\Tau, \pfrechet)$ into $\euclideanspace^d$.
\end{proposition}
\begin{proof}
	Let $\psi_d$ denote the $d$-dimensional kissing number (\conferre \cref{theo:kiss}), further let $n \eqdef \psi_{d} + 2$ and $m \eqdef 2$, such that $\tau_1, \dots, \tau_{\psi_{d} + 2}$ are line segments. Let $I \eqdef \{ p_1, \dots, p_{\psi_d} \} \subset \euclideanspace^d$ and let $E \eqdef \{ q_1, \dots, q_{\psi_d} \} \subset \euclideanspace^d$ be sets of points. Let $\tau_1$ be an arbitrary but fixed curve. 
	
	Let $a \in \R_{\geq 0}$. We arrange $p_1, \dots, p_{\psi_d}$, such that they share $\tau_1(0)$ as a common nearest neighbor with distance $a$. Similarly, we arrange $q_1, \dots, q_{\psi_d}$, such that they share $\tau_1(1)$ as a common nearest neighbor with distance $a$.
	
	The main part is to define $\tau_2, \dots, \tau_{\psi_{d} + 2}$, such that they share $\tau_1$ as common nearest neighbor under the Fréchet distance, which is done as follows (\conferre \cref{fig:embedding}):
	
	For $i = 1, \dots, \psi_d$ we define $ \tau_{i+1} \eqdef \overline{p_i q_i}$. The crucial point is to define $\tau_{\psi_{d} + 2} \eqdef \overline{p_1 q_2}$ or similar. The main thing is that $\tau_{\psi_{d} + 2}$ does not equal any of $\tau_1, \dots, \tau_{\psi_{d} + 1}$. From \cref{obs:frechet_endpoints} it immediately follows that $\tau_1$ is a common nearest neighbor of $\tau_2, \dots, \tau_{\psi_d+2}$ with Fréchet distance $a$.
	
	For $m > 2$ let the edges of $\tau_{i}$, for $i \in \{1, \dots, \psi_{d}+2\}$, lie very close to the respective previously defined line segment (we want the curves to be \emph{nearly} equal to the line segments), such that the Fréchet distances among the curves is equal to the Fréchet distances among the line segments. It is easy to see, that such a construction is possible.
	
	\begin{figure}[h]
		\begin{center}
			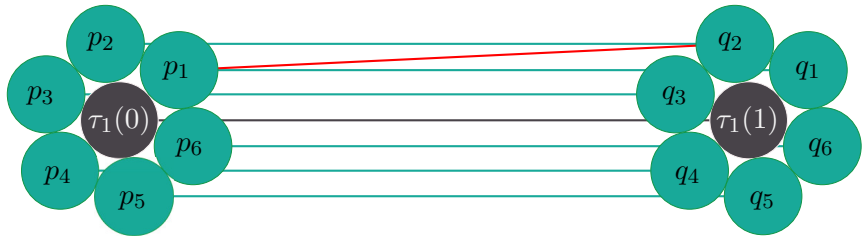
			\caption{Exemplary construction of $\Tau$ in the proof of \cref{prop:embedding} for $d=2$. The curves are defined with respect to the centers of the balls. The red curve, \ie $\tau_{\psi_d + 2} = \overline{p_1 q_2}$, breaks every embedding.}
			\label{fig:embedding}
		\end{center}
	\end{figure}

	Now assume there exists an isometry $f\colon \Tau \rightarrow \euclideanspace^d$ that embeds $(\Tau,\pfrechet)$ into $\euclideanspace^d$. Then the point $f(\tau_1)$ is a common nearest neighbor of $\psi_{d} + 1$ other points $f(\tau_2), \dots, f(\tau_{\psi_{d} + 2})$. Because $\psi_{d} + 1 > \psi_d$ we have a contradiction, which finishes the proof.
\end{proof}
\cref{prop:embedding} shows that even for a small number of input-curves there might not exist an embedding of $(\Tau, \pfrechet)$ into the euclidean space, therefore we need to check carefully if an adaption from an \coreset construction for point-sets does work for curve-sets. In the following we will examine and rigorously analyze such adaptions, but before we do, we state a corollary that can be derived from \cref{prop:embedding} surprisingly.
\begin{corollary}
	\label{coro:kiss_lines}
	The maximum number of line segments in $d$-dimensional euclidean space that can share a common nearest neighbor is less than or equal to $\psi_d^2$ (\conferre \cref{theo:kiss}).
\end{corollary}
\begin{proof}
	Let $\tau$ be an arbitrary line segment. Assume there exists a set of $\psi_d^2 + 1$ line segments $\Tau^\prime \eqdef \{\tau_1^\prime, \dots, \tau_{\psi_{d}^2 + 1}^\prime\}$ that have $\tau$ as common nearest neighbor with Fréchet distance $a \in \R_{\geq 0}$. By \cref{def:frechet_distance} we have that \[I \eqdef \{ \tau^\prime(0) \mid \tau^\prime \in \Tau^\prime \} \subset B^I \eqdef \{ p \in \euclideanspace^d \mid \eucl{p}{\tau(0)} \leq a \}\] and \[E \eqdef \{ \tau^\prime(1) \mid \tau^\prime \in \Tau^\prime \} \subset B^E \eqdef \{ p \in \euclideanspace^d \mid \eucl{p}{\tau(1)} \leq a \}.\]
	
	Now by \cref{obs:frechet_endpoints} we know that for two arbitrary line segments, $\sigma$ and $\sigma^{\prime}$, the Fréchet distance is either realized through $\eucl{\sigma(0)}{\sigma^{\prime}(0)}$ or $\eucl{\sigma(1)}{\sigma^{\prime}(1)}$. That means for $\tau^\prime \in \Tau^\prime$ either $\eucl{\tau(0)}{\tau^\prime(0)} = a$ or $\eucl{\tau(1)}{\tau^\prime(1)} = a$ and for any other $\tau^{\prime\prime} \in \Tau^\prime \setminus \{ \tau^\prime \}$ it must hold that if $\eucl{\tau^\prime(0)}{\tau^{\prime\prime}(0)} < a$, then $\eucl{\tau^\prime(1)}{\tau^{\prime\prime}(1)} \geq a$ and vice versa.
	
	Let \[B_o^I \eqdef  \{ p \in \euclideanspace^d \mid \eucl{p}{\tau(0)} < a \}\] and \[ B^E_o \eqdef \{ p \in \euclideanspace^d \mid \eucl{p}{\tau(1)} < a \}. \]
	For an arbitrary $\tau^\prime \in \Tau^\prime$ assume that $\tau^\prime(0) \in B_o^I$. This means that $\tau^\prime(1) \in B^E\setminus B^E_o$ must hold, otherwise $\frechet{\tau}{\tau^\prime} < a$. The same argument holds for $\tau^\prime(1) \in B^E_o$, hence it is detrimental if $\tau^{\prime\prime}(0) \in B_o^I$ or $\tau^{\prime\prime}(1) \in B^E_o$ for any $\tau^{\prime\prime} \in \Tau^\prime$. Further, assume for $\tau^\prime, \tau^{\prime\prime} \in \Tau^\prime$ that $\eucl{\tau^\prime(0)}{\tau^{\prime\prime}(0)} < a$. This implies that $\eucl{\tau^\prime(0)}{\tau^{\prime\prime}(0)} \geq a$ must hold, otherwise $\frechet{\tau^\prime}{\tau^{\prime\prime}} < a$ and then $\tau$ is not a nearest neighbor for both line segments. The same argument holds for $\tau^\prime(1)$ and $\tau^{\prime\prime}(1)$. 
	
	Both arguments hold for all $\tau^\prime \in \Tau^\prime$ at a time. In combination, they imply that $I \subset B^I \setminus B_o^I$ and $E \subset B^E \setminus B^E_o$ must hold and further it is implied that $\vert I \vert \leq \psi_d$, as well as $\vert E \vert \leq \psi_d$ by \cref{theo:kiss}.
		
	Finally, the maximum number of line segments that can be obtained is to connect every point in $I$ to every point in $E$, which gives $\psi_d^2$ line segments. This yields to a contradiction which finishes the proof.
\end{proof}
\cref{coro:kiss_lines} shows that the space $(\Tau, \pfrechet)$ has other characteristics than $\euclideanspace^d$ (which is already thoroughly studied), although these characteristics can be derived from $\euclideanspace^d$. Thus, techniques that depend highly on the underlying geometry of the objects, such as geometric decompositions, will most likely do not work with respect to curve-sets as they do with respect to point-sets. To apply such a technique we will have to check every aspect of the technique carefully regarding their correctness in relation to curves. We will encounter such a situation in \cref{sec:klcenter}, which unfortunately yields very complicated construction techniques that lead to worse results than those with respect to point-sets, in terms of the cardinality of the resulting \coreset[s].

On the other hand, techniques that depend only on the property that the employed distance function is a metric (which is the case with the Fréchet distance) can be applied directly. This will be the case in \cref{sec:klmedian}, which fortunately yields good results in terms of the cardinality of the resulting \coreset[].

In \cref{sec:klmeans} we will again try to tackle the underlying geometry of the clustering problem, using the benefits of \cref{prop:zauberformel}, but this attempt does fail due to the different geometry of the problems at hand.

The prior statements hint that common \coreset constructions through geometric decompositions may have to be highly modified to work with curves under the Fréchet distance. When an approximate solution to the clustering problem is available an adaption of a construction used for point-sets based on a minimum enclosing ball/minimum radius ball cover is possible, though.
\section{The \klcenter Objective}
\label{sec:klcenter}

In this section we are working with \cref{def:klcenter}, \ie the \klcenter objective, and we are proving \cref{theo:algo_coreset_center_ls} and \cref{theo:algo_coreset_center_pc}. We are given a set of curves $\Tau \eqdef \{ \tau_1, \dots, \tau_n \} \subset \eqcfre{m}$ and for $k=1$ we are simply looking for a center-set $\{c\} \subset \eqcfre{l}$ such that the center-curve is the center of a smallest enclosing ball of the curve-set (\conferre \cref{def:ball}). For $k>1$ we are looking for a center-set $C \subset \eqcfre{l}$, such that the center-curves are centers of a minimum radius ball cover. Of course, with respect to the Fréchet distance, the obtained geometric objects are barely similar to (regular) balls in the euclidean space.

\subsection{Moving Curves}

We are going to construct the \coreset[s] by moving curves, or more precisely, by partitioning the input-set with respect to the pairwise distances among the curves and adding only a single curve from each element of the partition to the \coreset[]. This can be seen as a kind of movement as a result of which the elements of the partition collapse to a single curve that is present multiple times, so we have to add only a single instance of the respective curve to the \coreset[]. The more curves an element of the partition contains, the more the input-set is reduced. This idea is very central in \coreset[]-constructions. For example a similar approach for point-sets can be found in the long version of \citet{movement_lemma_origin} or in \citet[Theorem 1]{coresers_methods_history}. We start by bounding the maximum movement for each curve. 
\begin{proposition}
	\label{prop:movement_center}
	Let $\pi \colon \Tau \rightarrow \eqcfre{m}$ be a function with $\forall \tau \in \Tau \colon \frechet{\tau}{\pi(\tau)} \leq \epsilon \cdot \optcost{\Tau}$, then $\forall C \subset \eqcfre{l}: (\vert C \vert = k) \Rightarrow (1-\epsilon) \cost{\Tau}{C} \leq \cost{\pi(\Tau)}{C} \leq (1+\epsilon) \cost{\Tau}{C} $, where $\pi(\Tau) \eqdef \{ \pi(\tau) \mid \tau \in \Tau \}$.
\end{proposition}
\begin{proof}
	Let $\tau \in \Tau$ and $C \subset \eqcfre{l}$, with $\vert C \vert = k$, be arbitrary but fixed and let $c \eqdef \nearest{\tau}{C}$ (\conferre \cref{def:nearest_func}) be a nearest neighbor of $\tau$ in $C$. Further let $c^\prime \eqdef \nearest{\pi(\tau)}{C}$ be a nearest neighbor of $\pi(\tau)$ in $C$. 
	
	First we show that $\frechet{\pi(\tau)}{c^\prime} \leq \frechet{\pi(\tau)}{\tau} + \frechet{\tau}{c}$. Assume that $\frechet{\pi(\tau)}{c^\prime} > \frechet{\pi(\tau)}{\tau} + \frechet{\tau}{c}$, then by the triangle-inequality \[\frechet{\pi(\tau)}{c} \leq \frechet{\pi(\tau)}{\tau} + \frechet{\tau}{c} < \frechet{\pi(\tau)}{c^\prime},\] which is a contradiction, because $c^\prime$ is a nearest neighbor of $\pi(\tau)$ in $C$. 
	
	Now we show that $\frechet{\pi(\tau)}{c^\prime} \geq \frechet{\tau}{c} - \frechet{\tau}{\pi(\tau)}$. Assume that $\frechet{\pi(\tau)}{c^\prime} < \frechet{\tau}{c} - \frechet{\tau}{\pi(\tau)}$, then by the triangle-inequality \[ \frechet{\tau}{c^\prime} \leq \frechet{\tau}{\pi(\tau)} + \frechet{\pi(\tau)}{c^\prime} < \frechet{\tau}{c}, \] which again is a contradiction, because $c$ is a nearest neighbor of $\tau$ in $C$.
	
	Let $\sigma \eqdef \argmax\limits_{\tau \in \Tau} \frechet{\pi(\tau)}{\nearest{\pi(\tau)}{C}}$, we have:
	\begin{align*}
		\cost{\pi(\Tau)}{C} ={} & \frechet{\pi(\sigma)}{\nearest{\pi(\sigma)}{C}} \leq \frechet{\pi(\sigma)}{\sigma} + \frechet{\sigma}{\nearest{\sigma}{C}} \tag{I} \label{movement_ie1} \\
		\leq{} & \max_{\tau \in \Tau} \left[\frechet{\pi(\tau)}{\tau} + \frechet{\tau}{\nearest{\tau}{C}}\right] \tag{II} \label{movement_ie2} \\
		\leq{} & \underbrace{\max_{\tau \in \Tau} \frechet{\tau}{\nearest{\tau}{C}}}_{= \cost{\Tau}{C}} + \epsilon \cdot \optcost{\Tau} \tag{III} \label{movement_ie3} \\
		\leq{} & (1 + \epsilon) \cdot \cost{\Tau}{C} \tag{IV} \label{movement_ie4}
	\end{align*}
	\cref{movement_ie1} follows from the argument above and \cref{movement_ie2} holds, because $\sigma \in \Tau$. Further, \cref{movement_ie3} follows from the definition of $\pi$ and \cref{movement_ie4} follows from the fact, that every clustering hast cost greater than or equal to $\optcost{\Tau}$.
	
	Also we have:
	\begin{align*}
		\cost{\pi(\Tau)}{C} ={} & \frechet{\pi(\sigma)}{\nearest{\pi(\sigma)}{C}} = \max_{\tau \in \Tau} \frechet{\pi(\tau)}{\nearest{\pi(\tau)}{C}} \label{movement_ie5} \tag{V} \\
		\geq{} & \max_{\tau \in \Tau} \left[ \frechet{\tau}{\nearest{\tau}{C}} - \frechet{\pi(\tau)}{\tau} \right] \label{movement_ie6} \tag{VI} \\
		\geq{} & \underbrace{\max_{\tau \in \Tau} \frechet{\tau}{\nearest{\tau}{C}}}_{=\cost{\Tau}{C}} - \epsilon \cdot \optcost{\Tau} \label{movement_ie7} \tag{VII}  \\
		\geq{}	& (1-\epsilon) \cdot \cost{\Tau}{C} \label{movement_ie8} \tag{VIII} 
	\end{align*}
	Here \cref{movement_ie5} follows from the definition of $\sigma$, \cref{movement_ie6} follows from the above arguments, \cref{movement_ie7} follows from the definition of $\pi$ and \cref{movement_ie8} follows from the fact that $\optcost{\Tau} \leq \cost{\Tau}{C}$ for all $C \subset \eqcfre{l}$ with $\vert C \vert = k$.
\end{proof}
Here it becomes obvious that an approximate solution to the clustering is extremely beneficial, because its value yields a lower bound on the value of an optimal solution, which is necessary to construct such a function $\pi$.

At first, we want to make clear, that \cref{prop:movement_center} only works together with a clever partitioning scheme. For an \coreset construction with respect to point-sets an approximate solution to the clustering objective has several tasks. Not only does it yield a lower bound on the optimum value, as it was stated before, but it does also yield an approach for a good partition of the input point-set. Namely, we can partition the radius of its objective value around the obtained center-points with grids. Note that, as grids are defined in this work, they do not yield a proper partitioning of the space, because the cells overlap at their boundaries, thus are not pairwise disjoined. We assume that ties are broken arbitrarily when a point of the input-set lies in two cells at a time, such that we obtain a proper partition of the input point-set. We set the cells edge length such that their diagonal, which gives the maximum distance two points in a cell can have, has length less than or equal to $\epsilon \cdot \optcost{\Tau}$. Then $\pi$ becomes a function that maps every cell to a single point in the cell. Because the number of cells is independent of the cardinality of the input-set, we obtain an \coreset of sub-linear cardinality. With respect to curve-sets the approach is similar.

\subsection{A Constant-Factor Approximation Algorithm}
\label{subsec:center_approx}
As it was stated in the previous section, we need at least a lower bound on the optimal value of the clustering objective to construct the function that is needed for \cref{prop:movement_center} to work.

For this purpose we use a relatively simple constant-factor approximation algorithm from \citet{approx_k_l_center} that is based on \citet{gonzalez} algorithm, endowed with an \el-simplification algorithm. The $\el$-simplification algorithm is used to prevent over-fitting, \ie one does not want to obtain centers that have a number of vertices equal to $m$ or even more. If one does not want to reduce the complexity of the center-curves, $l$ can be set to be equal to $m$.

\begin{algorithm}
	\begin{algorithmic}[1]
		\caption{Compute $6$-Approximate Solution to the \klcenter Objective}\label{algo:klcenter_approx}
		\Procedure{kl-center-approx}{$\Tau$}
		\State $C \gets \emptyset$
		\State $c_1 \gets$ arbitrary $\tau \in \Tau$
		\State $\hat{c}_1 \gets$ simplify$(c_1,l)$
		\State $C \gets C \cup \{ \hat{c}_1 \}$
		\For{$i=2, \dots, k$}
		\State $c_i \gets \argmax\limits_{\tau \in \Tau} \min\limits_{j \in \{1, \dots, i-1\}} \frechet{\tau}{\hat{c}_j}$
		\State $\hat{c}_i \gets$ simplify$(c_i, l)$
		\State $C \gets C \cup \{ \hat{c}_i \}$
		\EndFor
		\State \textbf{return} $C$
		\EndProcedure
	\end{algorithmic}
\end{algorithm}

\cref{algo:klcenter_approx} works as follows: An $l$-simplification of an arbitrary curve becomes the first center. Now in each iteration an $l$-simplification of a curve, that lies farthest from the already chosen curves, is added as the next center. This is repeated until there are $k$ centers. \citeauthor{approx_k_l_center} show that the obtained algorithm is a $6$-approximation for the \klcenter objective with running-time $\On{mn \log(m) + m^3 \log(m)}$, which is stated in the following theorem:
\begin{theorem}[{\citet[Corollary 21]{approx_k_l_center}}]
	\label{theo:approx_algo_center}
	\cref{algo:klcenter_approx} computes a $6$-approximate solution to the \klcenter objective in time $\On{mn \log(m) + m^3 \log(m))}$.
\end{theorem}

Now we introduce a rather simple technique for constructing \coreset[s] for line segments. Building upon this we will introduce a technique for constructing \coreset[s] for general polygonal curves in the subsequent subsection.

\subsection{An $\epsilon$-coreset Construction for Line Segments}

In this work we are dealing with curve sets, which makes it hard to define a simple partitioning scheme that yields to a similar result as those with respect to point-sets. However, this is possible for line segments, though.

\begin{algorithm}
	\caption{Compute $\epsilon$-Coreset for the \ensuremath{(k,2)}-\textsc{Center} Objective for Line Segments}\label{algo:klcenter_coreset_1}
	\begin{algorithmic}[1]
		\Procedure{kl-center-coreset-line-segments}{$\Tau$, $\epsilon$}
			\State $S \gets \emptyset$
			\State $\hat{C} \eqdef \{ \hat{c}_1, \dots, \hat{c}_k \} \gets $\textsc{kl-center-approx}$(\Tau)$ \Comment{\cref{algo:klcenter_approx}}
			\State $\approxcost{\Tau} \gets \cost{\Tau}{\hat{C}}$
			\For{$i = 1, \dots, k$}
				\State $\mathfrak{C}_{i,0} \gets \cube{\hat{c}_i(0), 2\cdot \approxcost{\Tau}}$
				\State $\mathfrak{C}_{i,1} \gets \cube{\hat{c}_i(1), 2\cdot \approxcost{\Tau}}$
				\State $\mathfrak{G}_{i,0} \gets \grid{\mathfrak{C}_{i,0}, \nicefrac{1}{\sqrt{d}}\cdot\epsilon\cdot\nicefrac{\approxcost{\Tau}}{6}}$
				\State $\mathfrak{G}_{i,1} \gets \grid{\mathfrak{C}_{i,1}, \nicefrac{1}{\sqrt{d}}\cdot\epsilon\cdot\nicefrac{\approxcost{\Tau}}{6}}$
				\For{$(i^\prime_1, \dots, i^\prime_d) \in \left(\left\{ -\left\lceil \nicefrac{(6 \sqrt{d})}{\epsilon}\right\rceil, \dots, -1\right\} \cup \left\{1, \dots, \left\lceil \nicefrac{(6 \sqrt{d})}{\epsilon}\right\rceil\right\}\right)^d$}
					\For{$(j_1, \dots, j_d) \in \left(\left\{ -\left\lceil \nicefrac{(6 \sqrt{d})}{\epsilon}\right\rceil, \dots, -1\right\} \cup \left\{1, \dots, \left\lceil \nicefrac{(6 \sqrt{d})}{\epsilon}\right\rceil\right\}\right)^d$}
						\State $\tau \gets \tau^\prime \in \Tau$ with $\tau^\prime(0) \in \cell{\mathfrak{G}_{i,0}}{i^\prime_1, \dots, i^\prime_d}$ and $\tau^\prime(1) \in \cell{\mathfrak{G}_{i,1}}{j_1, \dots, j_d}$
						\State $S \gets S \cup \{ \tau \}$
					\EndFor
				\EndFor
			\EndFor
			\State \textbf{return} $S$
		\EndProcedure
	\end{algorithmic}
\end{algorithm}

\cref{algo:klcenter_coreset_1} employs a partitioning scheme for the given curve-set based on grids: For any center of the approximate clustering returned by \cref{algo:klcenter_approx}, \cref{algo:klcenter_coreset_1} puts a grid around its initial point and around its end point, both of edge length twice the value of the objective function. From \cref{def:frechet_distance}, \cref{obs:ball_grid} and \cref{obs:cube_grid} we know that the initial points and the end points of every curve, that is associated to the respective center by the clustering, are covered by the grids. Now for every center and every combination of the grid cells of its initial point grid and its end point grid, \cref{algo:klcenter_coreset_1} adds a single curve to the \coreset[] that has its initial point in the respective cell of the initial point grid and its end point in the respective cell of the end point grid. Because the number of grids and the number of cells of each grid is independent of the input-set, this yields a number of curves independent of the input-set. If a point lies in two cells at a time, which can be the case, we assume that ties are broken arbitrarily to obtain a proper partition of the input-set.
\newpage
\subsubsection{Correctness Analysis of \cref{algo:klcenter_coreset_1}}

\begin{theorem}
	\label{theo:klcenter_coreset_1}
	Given a set $\Tau \eqdef \{\tau_1, \dots, \tau_n\} \subset \eqcfre{2}$ of line segments and a parameter $\epsilon \in (0,1)$, \cref{algo:klcenter_coreset_1} computes an \coreset for the \ensuremath{(k,2)}-\textsc{Center} objective of cardinality $\On{\frac{1}{\epsilon^{2d}}}$.
\end{theorem}
\begin{proof}
	Let $\hat{C} \eqdef \{ \hat{c}_1, \dots, \hat{c}_k \}$ be the center-set obtained by \cref{algo:klcenter_approx} with $\approxcost{\Tau} \eqdef \cost{\Tau}{C}$. From \cref{def:klcenter} we know that \[\{ \tau(0) \mid \tau \in \Tau \} \subseteq \bigcup_{i=1}^k B^I_i,\] where $B^I_i \eqdef \{ p \in \euclideanspace^d \mid \eucl{p}{\hat{c}_i(0)} \leq \approxcost{\Tau} \}$ is the closed ball of radius $\approxcost{\Tau}$ around the initial point of $\hat{c}_i$. Similarly, we know that \[\{ \tau(1) \mid \tau \in \Tau \} \subseteq \bigcup_{i=1}^k B^E_i,\] where $B^E_i \eqdef \{ p \in \euclideanspace^d \mid \eucl{p}{\hat{c}_i(1)} \leq \approxcost{\Tau} \}$ is the closed ball of radius $\approxcost{\Tau}$ around the end point of $\hat{c}_i$.
	
	For $i = 1, \dots, k$ let $\mathfrak{C}_{i,0} \eqdef \cube{\hat{c}_i(0), 2\cdot \approxcost{\Tau}}$ and $\mathfrak{C}_{i,1} \eqdef \cube{\hat{c}_i(1), 2\cdot \approxcost{\Tau}}$ be the cubes defined in \cref{algo:klcenter_coreset_1} and let \[\mathfrak{G}_{i,0} \eqdef \grid{\mathfrak{C}_{i,0}, \nicefrac{1}{\sqrt{d}}\cdot\epsilon\cdot\nicefrac{\approxcost{\Tau}}{6}}\] and \[\mathfrak{G}_{i,1} \eqdef \grid{\mathfrak{C}_{i,1}, \nicefrac{1}{\sqrt{d}}\cdot\epsilon\cdot\nicefrac{\approxcost{\Tau}}{6}}\] be the associated grids. By \cref{obs:ball_grid} we know that $B^I_i \subseteq \mathfrak{C}_{i,0}$ and $B^E_i \subseteq \mathfrak{C}_{i,1}$. Further, by \cref{obs:cube_grid} we know that $\mathfrak{C}_{i,0} = \mathfrak{G}_{i,0}$ and $\mathfrak{C}_{i,1} = \mathfrak{G}_{i,1}$, therefore the vertices of any curve that has distance less than or equal to $\approxcost{\Tau}$ to any center in $\hat{C}$ are covered by the grids.
	
	For $i = 1, \dots, k$ let $\Tau_i \eqdef \{ \tau \in \Tau \mid \nearest{\tau}{\hat{C}} = \hat{c}_i \}$ be the set of curves whose nearest center is $\hat{c}_i$. We know that the curves in $\Tau_i$ have their initial points in a cell of $\mathfrak{G}_{i,0}$ and their end points in a cell of $\mathfrak{G}_{i,1}$ which have edge length $\frac{1}{\sqrt{d}} \cdot \epsilon \cdot \frac{\approxcost{\Tau}}{6}$, therefore the maximum distance of two points in a cell, which is given by a diagonal of the cell, is:
	\begin{align*}
		\eucl{(0, \dots, 0)}{\left(\frac{1}{\sqrt{d}} \cdot \epsilon \cdot \frac{\approxcost{\Tau}}{6}, \dots, \frac{1}{\sqrt{d}} \cdot \epsilon \cdot \frac{\approxcost{\Tau}}{6}\right)} ={} & \sqrt{\sum_{i=1}^{d} \left(\frac{1}{\sqrt{d}} \cdot \epsilon \cdot \frac{\approxcost{\Tau}}{6} - 0\right)^2} \\
		={} & \sqrt{d \cdot \frac{1}{d} \cdot \epsilon^2 \cdot \frac{\approxcost{\Tau}^2}{6^2}} \\
		={} & \epsilon \cdot \frac{\approxcost{\Tau}}{6} \leq \epsilon \cdot \optcost{\Tau}
	\end{align*}
	Here the inequality follows from \cref{theo:approx_algo_center}, which states that \cref{algo:klcenter_approx} is a $6$-approximation.
	
	By \cref{obs:frechet_endpoints} any two curves that have their initial points in the same cell of $\mathfrak{G}_{i,0}$ and their end points in the same cell of $\mathfrak{G}_{i,1}$, for any $i \in \{1, \dots, k\}$, have Fréchet distance less than or equal to $\epsilon \cdot \optcost{\Tau}$. Now let $\pi\colon \Tau \rightarrow \Tau$ be a function that maps all input-curves to an arbitrary curve that has its initial point in the same cell of a $\mathfrak{G}_{i,0}$ and its end point in the same cell of a $\mathfrak{G}_{i,1}$, with the restriction that every two curves that have their initial point in the same cell and their end point in the same cell are mapped to the same curve, as it is done in \cref{algo:klcenter_coreset_1}. Thus, by \cref{prop:movement_center} the set $S$ returned by \cref{algo:klcenter_coreset_1} is an \coreset for the \ensuremath{(k,2)}-\textsc{Center} objective.
	
	It stays to show that $\vert S \vert \in \On{\frac{1}{\epsilon^{2d}}}$: For every $i \in \{1, \dots, k\}$ the volume of $\mathfrak{G}_{i,0}$ and $\mathfrak{G}_{i,1}$ is equal to $2^d \cdot \approxcost{\Tau}^d$. Every cell has volume $\nicefrac{1}{\sqrt{d}^d} \cdot \epsilon^d \cdot \nicefrac{\approxcost{\Tau}^d}{6^d}$, therefore every grid has number of cells up to:
	\begin{align*}
		\frac{2^d \cdot \approxcost{\Tau}^d}{\frac{1}{\sqrt{d}^d} \cdot \epsilon^d \cdot \frac{\approxcost{\Tau}^d}{6^d}} = \frac{2^d}{\frac{1}{\sqrt{d}^d} \cdot \epsilon^d \cdot \frac{1}{6^d}} = \frac{2^d 6^d d^{\nicefrac{d}{2}}}{\epsilon^d}
	\end{align*}
	Now for every $i = 1, \dots, k$ the function $\pi$ maps to a single curve from every pair of the cells of $\mathfrak{G}_{i,0}$ and $\mathfrak{G}_{i,1}$ which are up to $k \cdot \frac{2^{2d} 6^{2d} d^{d}}{\epsilon^{2d}} = 2^{2d} 6^{2d} d^{d} k \cdot \frac{1}{\epsilon^{2d}}$ curves.
\end{proof}

\subsubsection{Time Complexity Analysis of \cref{algo:klcenter_coreset_1}}

\begin{theorem}
	Given a set of $n$ line segments and a parameter $\epsilon \in (0,1)$, \cref{algo:klcenter_coreset_1} has running-time $\On{\frac{n}{\epsilon^{2d}}}$.
\end{theorem}
\begin{proof}
	Let $\hat{C} \eqdef \{ \hat{c}_1, \dots, \hat{c}_k \}$ be the center-set returned by \cref{algo:klcenter_approx}, which has running-time $\On{mn \log(m) + m^3 \log(m)}$, \conferre \cref{theo:approx_algo_center}. The main part of the running-time of \cref{algo:klcenter_coreset_1} is to construct the grids for the initial point and the end point of every $\hat{c} \in \hat{C}$. To test whether a point $p \eqdef (p_1, \dots, p_d) \in \euclideanspace^d$ is contained in a cell of such a grid it is sufficient to know the intervals, one interval for each dimension, that the cell covers. Recall from \cref{theo:klcenter_coreset_1} that each grid has $\frac{2^d 6^d d^{\nicefrac{d}{2}}}{\epsilon^d}$ cells. All in all \cref{algo:klcenter_coreset_1} calls \cref{algo:klcenter_approx}, then computes $k \cdot 2 \cdot d \cdot \frac{2^d 6^d d^{\nicefrac{d}{2}}}{\epsilon^d}$ intervals, then it checks every pair of cells from the initial point grid and the end point grid of each $\hat{c} \in \hat{C}$ if there is a line segment $\tau \in \Tau$, that has its initial, respective end point in those cells. This yields the running-time, for a sufficiently large $c \in \R_{> 0}$:
	\begin{align*}
		& c \cdot (mn \log(m) + m^3 \log(m)) + k \cdot 2 \cdot d \cdot \frac{2^d 6^d d^{\nicefrac{d}{2}}}{\epsilon^d} + k \cdot 2d \cdot \frac{2^{2d} 6^{2d} d^{d}}{\epsilon^{2d}} \cdot n \\
		\leq{} & c \cdot (mn \log(m) + m^3 \log(m)) + 2 k \cdot 2d \cdot \frac{2^{2d} 6^{2d} d^{d}}{\epsilon^{2d}} \cdot n \\
		\leq{} & (c + 2 k \cdot 2d \cdot 2^{2d} 6^{2d} d^{d} + 2^3 \log(2)) \cdot \left(n + \frac{n}{\epsilon^{2d}}\right)
	\end{align*}
	So additionally to a large constant we have running-time linear in $n$ and polynomial in $\nicefrac{1}{\epsilon}$.
\end{proof}

\subsection{An $\epsilon$-coreset Construction for Polygonal Curves}
\label{subsec:mg2}
The case of polygonal curves of complexity at least $3$ clearly is very different, because the input-curves have more than two vertices, which also need to be covered by grids. The idea is simple: Cover the whole center-curves with grids and do the same construction as before. But there is a little deficit: We have defined grids based on cubes whose edges are axes-parallel, \ie grids based on uniform cubes. Clearly this may not be the case for the edges of the center-curves. The proposed solution is simple: We successively rotate the edges of the center-curves such that they lie on the non-negative section of the axis of the $d$\textsuperscript{th} dimension, then we define grids that cover the edge and rotate both the edge and the grids in reverse order and reverse direction, such that the edge remains unchanged but is covered by the grids. This is possible because all required transformations are isometries that are invertible and their inverses are also isometries. Also, we can efficiently compute the angles by which we have to rotate the edges, with the help of the following definition:
\begin{definition}
	\label{def:axisangle}
	For $\vv{x} \eqdef (x_1, \dots, x_d) \in \euclideanspace^d \setminus \{\vv{o}\}$ and $i \in \{1, \dots, d-1\}$ let 
	\begin{equation*}
		\axisangle{\vv{x}}{i} \eqdef \\
		\begin{cases}
		\frac{\pi}{2} - \angl{\vv{x}^\prime}{\vv{a}_i}, & \text{if } x_{i+1} \geq 0 \\
		\frac{\pi}{2} + \angl{\vv{x}^\prime}{\vv{a}_i}, & \text{else}
		\end{cases},
	\end{equation*}
	where $\vv{x}^\prime \eqdef (x^\prime_1, \dots, x^\prime_d)$ is the projection of $\vv{x}$ onto the plane spanned by the axes of the $i$\textsuperscript{th} and $(i+1)$\textsuperscript{th} dimension, therefore $x^\prime_j \eqdef 0$ for $j \in \{1, \dots, d\}  \setminus \{i, i+1\}$, $x^\prime_i \eqdef x_i$ and $x^\prime_{i+1} \eqdef x_{i+1}$. Also, the unit vector $\vv{a}_i \eqdef (a_1, \dots, a_d)$ is aligned on the $i$\textsuperscript{th} axis, where $a_j \eqdef 0$ for $j \in \{1, \dots, d\} \setminus \{i\}$ and $a_i \eqdef 1$.
\end{definition}

\begin{algorithm}
	\caption{Compute $\epsilon$-Coreset for the \klcenter Objective for Polygonal Curves}\label{algo:klcenter_coreset_mg1}
	\begin{algorithmic}[1]
		\Procedure{kl-center-coreset-polygonal-curves}{$\Tau$, $\epsilon$}
			\State Add dummy vertices to every $\tau \in \Tau$, such that all $\tau \in \Tau$ have $m$ vertices \label{algo:klcenter_coreset_mg1:line1}
			\State $S \gets \emptyset$
			\State $\hat{C} \eqdef \{ \hat{c}_1, \dots, \hat{c}_k \} \gets $\textsc{kl-center-approx}$(\Tau)$ \Comment{\cref{algo:klcenter_approx}}
			\State $\approxcost{\Tau} \gets \cost{\Tau}{\hat{C}}$
			\State $\delta \gets$ length of the longest edge of any $\hat{c} \in \hat{C}$
			\If{$\frac{\delta^m}{\approxcost{\Tau}^m} > \sqrt{n}$}
				\State \textbf{return} Unsuccessful \label{algo:klcenter_coreset_mg1:lineret}
			\EndIf
			\For{$i = 1, \dots, k$}
				\For{$j=1, \dots, \vert \hat{c}_i \vert - 1$} \label{algo:klcenter_coreset_mg1:line5}
					\State\Comment{Here we assemble the grid-cover.}
					\State $\vv{x}_{i,j} \gets \vv{v_{i,j}o}$ \label{algo:klcenter_coreset_mg1:line6}
					\State $v_{i,j,1} \gets \trans{v_{i,j}}{\vv{x}_{i,j}}$ \label{algo:klcenter_coreset_mg1:line7} \Comment{$j$\textsuperscript{th} vertex of $\hat{c}_i$ gets translated to origin}
					\State $v_{i, j+1, 1} \gets  \trans{v_{i,j+1}}{\vv{x}_{i,j}}$ \label{algo:klcenter_coreset_mg1:line8} \Comment{$(j+1)$\textsuperscript{th} vertex of $\hat{c}_i$ undergoes same translation}
					\For{$r = 1, \dots, d-1$} \label{algo:klcenter_coreset_mg1:line9} 
						\State $\alpha_{i,j,r} \gets \axisangle{\vv{o v_{i,j+1,r}}}{r}$
						\State $v_{i,j+1,r+1} \gets \rotate{v_{i,j+1,r}}{r}{\alpha_{i,j,r}}$ \Comment{set $r$\textsuperscript{th} component of $v_{i,j+1,r}$ to $0$}
					\EndFor \label{algo:klcenter_coreset_mg1:line12} \Comment{all components of $v_{i,j+1,d}$, except the $d$\textsuperscript{th}, are $0$ now}
					\For{$s = 1, \dots, \left\lceil 1 + \frac{\eucl{o}{v_{i,j+1,d}}}{2 \approxcost{\Tau}} \right\rceil$} \label{algo:klcenter_coreset_mg1:line13} \Comment{envelop $\overline{o v_{i,j+1,d}}$ with grids}
						\State \begin{varwidth}[t]{\linewidth} $\vv{y}_{i,j,s} \gets (y_{i,j,s,1}, \dots, y_{i,j,s,d})$ \par
							\hskip\algorithmicindent where $y_{i,j,s,r} \eqdef 0$, for $r \in \{1, \dots, d-1\}$, \par
							\hskip\algorithmicindent and $y_{i,j,s,d} \eqdef (s-1)\cdot 2 \approxcost{\Tau}$
						\end{varwidth}
						\State $\mathfrak{C}_{i,j,s} \gets \cube{\trans{o}{\vv{y}_{i,j,s}}, 2 \approxcost{\Tau}}$
						\State $\mathfrak{G}_{i,j,s} \gets \grid{\mathfrak{C}_{i,j,s}, \nicefrac{1}{\sqrt{d}}\cdot\epsilon\cdot\nicefrac{\approxcost{\Tau}}{6}}$ \label{algo:klcenter_coreset_mg1:line16}
					\EndFor \label{algo:klcenter_coreset_mg1:line17}
					\algstore{algo:center_mg2}
	\end{algorithmic}
\end{algorithm}
\begin{algorithm}
	\begin{algorithmic}[1]
					\algrestore{algo:center_mg2}
					\State\Comment{Here we put the grid-cover around the original edges of the center-curves.}
					\State $\mathfrak{G}_{i,j} \gets \emptyset$
					\For{$s=1, \dots, \left\lceil 1 + \frac{\eucl{v_{i,j,d}}{v_{i,j+1,d}}}{2 \approxcost{\Tau}} \right\rceil$} \label{algo:klcenter_coreset_mg1:line19}
						\State $\mathfrak{G}_{i,j,s,1} \gets \mathfrak{G}_{i,j,s}$
						\For{$r=1, \dots, d-1$}
							\State \begin{varwidth}[t]{\linewidth} $\mathfrak{G}_{i,j,s,r+1} \gets \{ \rotate{p}{d-r}{\alpha_{i,j,d-r}} \mid p \in \mathfrak{G}_{i,j,s,r} \}$ \par \hskip\algorithmicindent \Comment{rotate grids in reverse order}
							\end{varwidth}
						\EndFor
						\State \begin{varwidth}[t]{\linewidth} $\mathfrak{G}^\prime_{i,j,s} \gets \{ \trans{p}{-\vv{x}_{i,j}} \mid p \in \mathfrak{G}_{i,j,s,d} \}$ \par \hskip\algorithmicindent \Comment{translate grids to cover original edges}
						\end{varwidth}
						\State $\mathfrak{G}_{i,j} \gets \mathfrak{G}_{i,j} \cup \mathfrak{G}^\prime_{i,j,s}$
					\EndFor \label{algo:klcenter_coreset_mg1:line26}
				\EndFor \label{algo:klcenter_coreset_mg1:line27}
				\State\Comment{Here we build the actual \coreset[].}
				\For{every combination of $m$ cells of $\cup_{j=1}^{\vert \hat{c}_i \vert - 1} \mathfrak{G}_{i,j}$} \label{algo:klcenter_coreset_mg1:line28}
					\For{every permutation of those cells}
						\State \begin{varwidth}[t]{\linewidth} $\tau^\prime \gets$ arbitrary curve $\tau \in \Tau$ that has its vertices in those cells \par
						\hskip \algorithmicindent in the order of the current permutation, if exists, else $\bot$
						\end{varwidth} 
						\If{$\tau^\prime \neq \bot$}
							\State $S \gets S \cup \tau^\prime$
						\EndIf
					\EndFor
				\EndFor \label{algo:klcenter_coreset_mg1:line32}
			\EndFor
			\State \textbf{return} $S$
		\EndProcedure
	\end{algorithmic}
\end{algorithm}

Now \cref{algo:klcenter_coreset_mg1} works as follows: We run \cref{algo:klcenter_approx} on $\Tau$ to obtain a $6$-approximate solution $\hat{C} \eqdef \{ \hat{c}_1, \dots, \hat{c}_k \}$ to the \klcenter objective, with $\approxcost{\Tau} \eqdef \cost{\Tau}{\hat{C}}$. Then, for $i \in \{1, \dots, k\}$ we successively process each $\hat{c}_i$. One after another, we translate each edge of $\hat{c}_i$, such that its initial point is the origin, then we successively rotate the edge such that it lies on the non-negative section of the $d$\textsuperscript{th} axis with the help of \cref{def:axisangle}. Now we cover the edge with axis-parallel grids of edge length $2 \cdot \approxcost{\Tau}$. We rotate both the edge and the grids in reverse order and reverse direction and finally translate everything in the reverse direction of the initial translation. We obtain the original edge and additionally grids that cover the edge. In \cref{lem:envelope_cover} we will see that the vertices of every curve within Fréchet distance $\approxcost{\Tau}$ to $\hat{c}_i$, for $i \in \{1, \dots, k\}$, are covered by these grids. By \cref{prop:frechet_line_conc} we obtain that every two curves from the input-set have Fréchet distance at most $\epsilon \cdot \optcost{\Tau}$ if their vertices, in the order of their occurrence, lie in the same cells. Thus, for every combination of $m$ cells chosen from the grids that cover $\hat{c}_i$ and every permutation of these, for every $i \in \{1, \dots, k\}$, it suffices to pick a single curve from the input-set that has its vertices in these $m$ cells in the order of the permutation, if such a curve exists, to construct an \coreset[].

The construction has an obvious flaw: The cardinality of the resulting \coreset is not independent of the input-set: It does depend on the length of the longest edge of any center-curve. To be precise, it is dependent on the ratio of the longest edge of a center-curve and twice the cost of the approximate clustering $\hat{C}$, therefore we check if this ratio exceeds $\sqrt[2m]{n}$ in advance. If so, the algorithm may not be able to return a \coreset of sub-linear cardinality. In this case the algorithm will stop and return nothing. 

For the sake of simplicity of the following proofs, the algorithm adds dummy vertices to every input-curve that has less than $m$ vertices, such that all input-curves have $m$ vertices. This can for example be done by cloning the first vertex of every curve $\tau$ $m - \vert \tau \vert$ times.

\subsubsection{Correctness Analysis of \cref{algo:klcenter_coreset_mg1}}

\begin{theorem}
	\label{theo:klcenter_coreset_mg1}
	Given a set $\Tau \eqdef \{\tau_1, \dots, \tau_n\} \subset \eqcfre{m}$ of polygonal curves of complexity at least $3$ and a parameter $\epsilon \in (0,1)$, \cref{algo:klcenter_coreset_mg1} computes an \coreset for the \klcenter objective of cardinality $\On{2^{3m} \cdot \sqrt{n} \cdot \frac{l^{12d^2m}}{\epsilon^{dm}} + 2^m m^m}$, if successful.
\end{theorem}
\begin{proof}
	This proof is threefold: 1. We show that every curve within Fréchet distance less than or equal to $\approxcost{\Tau}$ to $\hat{c}_i$, for any $i \in \{1, \dots, k\}$, is covered by the grids defined by \cref{algo:klcenter_coreset_mg1}. 2. We show that for every curve $\tau \in \Tau$ there is a curve $\tau^\prime \in S$ within Fréchet distance less than or equal to $\epsilon \cdot \optcost{\Tau}$ to $\tau$ and therefore $S$ is an \coreset for the \klcenter objective by \cref{prop:movement_center}. 3. We show that $S$ has cardinality $\On{2^{3m} \cdot \sqrt{n} \cdot \frac{l^{12d^2m}}{\epsilon^{dm}} + 2^m m^m}$.
	
	\begin{enumerate}
		\item Without loss of generality assume that $\approxcost{\Tau} > 0$. For $i = 1, \dots, k$ let $b_i \eqdef \vert \hat{c}_i \vert$ be the complexity of $\hat{c}_i$. Further, let $v_{i,1}, \dots, v_{i,b_i}$ be the vertices of $\hat{c}_i$ and $e_{i,1}, \dots, e_{i,b_i-1}$ be the edges of $\hat{c}_i$. For $\tau \in \Tau$ let $t_{\tau,1}, \dots, t_{\tau,m}$ be chosen such that $v_{\tau, 1} = \tau(t_{\tau,1}), \dots, v_{\tau, m} = \tau(t_{\tau, m})$ are the vertices of $\tau$ (recall that every curve has $m$ vertices).
		
		Using these definitions, in \cref{lem:envelope_cover} we will show that the vertices of every curve $\sigma \in \eqcfre{m}$ within Fréchet distance less than or equal to $\approxcost{\Tau}$ to $\hat{c}_i$, for any $i \in \{1, \dots, k\}$, are covered by the grids defined by \cref{algo:klcenter_coreset_mg1}. Thus, every $v_{\tau,j}$, for an arbitrary $\tau \in \Tau$ and $j \in \{1, \dots, m\}$, is contained in at least one cell of a grid (assume that ties are broken arbitrarily when the vertex is contained in more than one cell).
		
		\item Recall from \cref{theo:klcenter_coreset_1} that the maximum distance of two points in the same cell is less than or equal to $\epsilon \cdot \optcost{\Tau}$. Let $i \in \{1, \dots, k\}$ be arbitrary but fixed and let $\mathfrak{G}_i$ be the set of grids \cref{algo:klcenter_coreset_mg1} uses to cover $\hat{c}_i$. For every combination of $m$ cells of $\mathfrak{G}_i$ and every permutation of these cells, in \cref{algo:klcenter_coreset_mg1:line28} to \cref{algo:klcenter_coreset_mg1:line32}, \cref{algo:klcenter_coreset_mg1} adds only one curve to the set $S$, that has its vertices in those cells and has them connected in the order of the permutation at hand, if such a curve exists in the input-set. Assume $\tau, \tau^\prime \in \Tau$ are curves which have their vertices in the same cells of $\mathfrak{G}_i$ and have them connected in the same order, \ie for every $j \in \{1, \dots, m\}$ the vertices $v_{\tau, j}$ and $v_{\tau^\prime, j}$ lie in the same cell. Let $e_{\tau, 1}, \dots, e_{\tau^\prime, m-1}$ be the edges of $\tau$ and $e_{\tau^\prime, 1}, \dots, e_{\tau^\prime, m-1}$ be the edges of $\tau^\prime$. For $j \in \{1, \dots, m-1\}$ let $\tau_{j}\colon [t_{\tau, j}, t_{\tau, j+1}], t \mapsto \tau(t)$ be the $j$\textsuperscript{th} sub-curve of $\tau$ with $\cup_{t \in [t_{\tau,j}, t_{\tau, j+1}]} \{\tau_j(t)\} = e_{\tau, j}$ and let $\tau^\prime_{j}\colon [t_{\tau^\prime, j}, t_{\tau^\prime, j+1}], t \mapsto \tau^\prime(t)$ be the $j$\textsuperscript{th} sub-curve of $\tau^\prime$ with $\cup_{t \in [t_{\tau^\prime,j}, t_{\tau^\prime, j+1}]} \{\tau^\prime_j(t)\} = e_{\tau^\prime, j}$. Recall from \cref{theo:klcenter_coreset_1} that for all $j = 1, \dots, m-1$ the edges $e_{\tau, j}$ and $e_{\tau^\prime, j}$ have Fréchet distance less than or equal to $\epsilon \cdot \optcost{\Tau}$, so $\tau_j$ and $\tau^\prime_j$ do too. Now by \cref{prop:frechet_line_conc} we obtain that $\tau = \tau_1 \oplus \dots \oplus \tau_{m-1}$ and $\tau^\prime = \tau^\prime_1 \oplus \dots \oplus \tau^\prime_{m-1}$ have Fréchet distance less than or equal to $\epsilon \cdot \optcost{\Tau}$. Let $\pi \colon \Tau \rightarrow \Tau$ be a function that maps every $\tau \in \Tau$ to itself if there does not exist such a $\tau^\prime$ and to such a $\tau^\prime \in S$ otherwise, with the restriction that every two such curves are mapped to the same curve as it is done in \cref{algo:klcenter_coreset_mg1}. By \cref{prop:movement_center} the set $S$ returned by \cref{algo:klcenter_coreset_mg1} thus is an \coreset for the \klcenter objective.
		
		\item Now let $\delta$ be the length of the longest edge $e$ of any center $\hat{c} \in \hat{C}$. \cref{algo:klcenter_coreset_mg1} uses $\left\lceil 1 + \frac{\delta}{2 \approxcost{\Tau}} \right\rceil$ cubes to envelope $e$, \conferre \cref{lem:envelope_cover}. Recall from the proof of \cref{theo:klcenter_coreset_1} that an associated grid of such a cube has $\frac{2^d 6^d d^{\nicefrac{d}{2}}}{\epsilon^d}$ cells, thus we have at most
		\begin{align*}
		\kappa \eqdef (l-1) \cdot \left\lceil \left(1 + \frac{\delta}{2 \approxcost{\Tau}}\right) \right\rceil \cdot \frac{2^d 6^d d^{\nicefrac{d}{2}}}{\epsilon^d}
		\end{align*}
		cells to cover any $\hat{c} \in \hat{C}$ (recall that they have complexity $l$). For every curve $\tau \in \eqcfre{m}$ we have at most $\multiset{\kappa}{m} = \binom{\kappa + m - 1}{m}$ possibilities to put the vertices of $\tau$ in the cells of the grids (the order pays no importance and the cells can be reused) and $m!$ possibilities to connect them with edges, \ie to order the vertices. All in all the \coreset returned by \cref{algo:klcenter_coreset_mg1} has maximum cardinality:
		\begin{align*}
		k \cdot \binom{\kappa + m - 1}{m} \cdot m! ={} & k \cdot \frac{(\kappa + m - 1)! \cdot m!}{(\kappa - 1)! \cdot m!} \leq k \cdot (\kappa + m)^m \\
		={} & k \cdot \left((l-1) \cdot \left\lceil\left( 1 + \frac{\delta}{2 \approxcost{\Tau}}\right) \right\rceil \cdot \frac{2^d 6^d d^{\nicefrac{d}{2}}}{\epsilon^d} + m\right)^m \\
		\leq{} & k \cdot \left(l \cdot \left(2 + \frac{\delta}{2 \approxcost{\Tau}}\right) \frac{2^d 6^d d^{\nicefrac{d}{2}}}{\epsilon^d} + m \right)^m \\
		\leq{} & k \cdot 2^{m-1} \left(l^m \left(2 + \frac{\delta}{2 \approxcost{\Tau}}\right)^m \cdot \frac{2^{dm} 6^{dm} d^{\nicefrac{dm}{2}}}{\epsilon^{dm}} + m^m\right) \label{eq:mg2comb1} \tag{I} \\
		\leq{} &  k \cdot 2^{m-1}\left(\left(2^{2m} + \frac{\delta^m}{\approxcost{\Tau}^m}\right) \frac{l^m 2^{dm} 6^{dm} d^{\nicefrac{dm}{2}}}{\epsilon^{dm}} + m^m\right) \label{eq:mg2comb2} \tag{II} \\
		\leq{} & k \cdot \left(2^{3m} \cdot \sqrt{n} \cdot \frac{l^{12d^2m}}{\epsilon^{dm}} + 2^m m^m \right) \label{eq:mg2comb3} \tag{III}
		\end{align*}
		Here \cref{eq:mg2comb1} and \cref{eq:mg2comb2} follow from the fact that $(x+y)^a \leq 2^{a-1} \cdot (x^a + y^a)$ and \cref{eq:mg2comb3} follows from the fact, that the algorithm fails in \cref{algo:klcenter_coreset_mg1:lineret}, if $\frac{\delta^m}{\approxcost{\Tau}^m} > \sqrt{n}$.
	\end{enumerate}
\end{proof}
A thing that comes to mind is the combinatorial term of $(\kappa + m)^m$ and the question if there is a smaller bound on the cardinality of the resulting \coreset[], like $\binom{\kappa}{m}$ or even smaller. A reason why we cannot easily obtain a smaller bound on the cardinality of the resulting \coreset[] is that two or more vertices may lie in the same cell. In fact all but one vertex of a curve may lie in the same cell. We remark that the term of $m!$ may be substituted by something smaller because the vertices may actually not be connected in any permutation. Because we do not know this in advance we may not provide a non-adaptive bound for all possible input-sets, though.
\paragraph{Here we prove some deferred lemmas.}
\begin{lemma}
	\label{lem:envelope_cover}
	For $i \in \{ 1, \dots, k\}$, the grids defined by \cref{algo:klcenter_coreset_mg1} cover all vertices of all $\sigma \in \eqcfre{m}$ with $\frechet{\sigma}{\hat{c}_i} \leq \approxcost{\Tau}$.
\end{lemma}
\begin{proof}
	Let $i \in \{1,\dots, k\}$ and $j \in \{1, \dots, b_i-1\}$ be arbitrary but fixed. We now look at the iterations for the $j$\textsuperscript{th} line segment of $\hat{c}_i$. In \cref{algo:klcenter_coreset_mg1:line6} \cref{algo:klcenter_coreset_mg1} defines $\vv{x}_{i,j} \eqdef \vv{v_{i,j}o}$ and translates both $v_{i,j}$ and $v_{i, j+1}$ by $\vv{x}_{i,j}$ in \cref{algo:klcenter_coreset_mg1:line7} and \cref{algo:klcenter_coreset_mg1:line8}. We obtain $v^\prime_{i,j} \eqdef \trans{v_{i,j}}{\vv{x}_{i,j}} = o$ and $v^\prime_{i,j+1} \eqdef \trans{v_{i,j+1}}{\vv{x}_{i,j}}$. For $p \in \euclideanspace^d$ let $D_{i,j,0}(p) \eqdef \trans{p}{\vv{x}_{i,j}}$ be this transformation and $D_{i,j,0}^{-1}(p) \eqdef \trans{p}{-\vv{x}_{i,j}}$ be its inverse (\conferre \cref{prop:translation_inverse}). Because $e_{i,j} = \overline{v_{i,j} v_{i,j+1}}$ we obtain a translated copy $e^\prime_{i,j} \eqdef \overline{v^\prime_{i,j} v^\prime_{i,j+1}}$ of $e_{i,j}$ that has its initial point at the origin. Now in \cref{algo:klcenter_coreset_mg1:line9} to \cref{algo:klcenter_coreset_mg1:line12} \cref{algo:klcenter_coreset_mg1} successively rotates $v^\prime_{i,j+1}$ in the plane spanned by the $r$\textsuperscript{th} and $(r+1)$\textsuperscript{th} axes, for $r = 1, \dots, d-1$, by $\alpha_{i,j,r} \eqdef \axisangle{\vv{o v^\prime_{i,j+1,r}}}{r}$, where $v^\prime_{i,j+1,r}$ is the result of the previous iteration, $v^\prime_{i,j+1,r+1} \eqdef \rotate{v^\prime_{i,j+1,r}}{r}{\alpha_{i,j,r}}$ is the result of the current iteration and $v^\prime_{i,j+1,1} \eqdef v^\prime_{i,j+1}$. For $p \in \euclideanspace^d$ and $r \in \{1, \dots, d-1\}$ let $D_{i,j,r}(p) \eqdef \rotate{p}{r}{\alpha_{i,j,r}}$ be the respective transformation and $D^{-1}_{i,j,r}(p) \eqdef \rotate{p}{r}{-\alpha_{i,j,r}}$ be its inverse (\conferre \cref{prop:rotation_inverse}). In \cref{lem:rotation} we will show that in each iteration for the resulting $v^\prime_{i,j+1,r+1} \eqdef (v^\prime_{i,j+1,r+1,1}, \dots, v^\prime_{i,j+1,r+1,d})$ it holds that $v^\prime_{i,j+1,r+1,r} = 0$ and $v^\prime_{i,j+1,r+1,r+1} \geq 0$. In conclusion, for the final result $v^\prime_{i,j+1,d} \eqdef (v^\prime_{i,j+1,d,1}, \dots, v^\prime_{i,j+1,d,d})$ it holds that $v^\prime_{i,j+1,d,r} = 0$, for $r \in \{0,\dots,d-1\}$, and $v^\prime_{i,j+1,d,d} \geq 0$. Because we will show in \cref{lem:axisangle} that every other point on $e^\prime_{i,j}$ except $v^\prime_{i,j} = o$ ($v^\prime_{i,j}$ is already aligned on the non-negative section of the $d$\textsuperscript{th} axis) has equal $\paxisangle$ as $v^\prime_{i,j+1}$, \conferre \cref{def:axisangle}, we have that  $e^{\prime\prime}_{i,j} \eqdef \overline{v^\prime_{i,j} v^\prime_{i, j+1, d}}$ is a translated and rotated copy of $e_{i,j}$ that is aligned on the non-negative section of the $d$\textsuperscript{th} axis and therefore has the advantage that it can be covered by grids.
	
	\begin{figure}
		\centering
		\def\svgwidth{0.78\textwidth}
\begingroup%
  \makeatletter%
  \providecommand\color[2][]{%
    \errmessage{(Inkscape) Color is used for the text in Inkscape, but the package 'color.sty' is not loaded}%
    \renewcommand\color[2][]{}%
  }%
  \providecommand\transparent[1]{%
    \errmessage{(Inkscape) Transparency is used (non-zero) for the text in Inkscape, but the package 'transparent.sty' is not loaded}%
    \renewcommand\transparent[1]{}%
  }%
  \providecommand\rotatebox[2]{#2}%
  \newcommand*\fsize{\dimexpr\f@size pt\relax}%
  \newcommand*\lineheight[1]{\fontsize{\fsize}{#1\fsize}\selectfont}%
  \ifx\svgwidth\undefined%
    \setlength{\unitlength}{581.74283814bp}%
    \ifx\svgscale\undefined%
      \relax%
    \else%
      \setlength{\unitlength}{\unitlength * \real{\svgscale}}%
    \fi%
  \else%
    \setlength{\unitlength}{\svgwidth}%
  \fi%
  \global\let\svgwidth\undefined%
  \global\let\svgscale\undefined%
  \makeatother%
  \begin{picture}(1,0.7389506)%
    \lineheight{1}%
    \setlength\tabcolsep{0pt}%
    \put(0.0,0.13){\color[rgb]{0,0,0}\rotatebox{-50.299208}{\makebox(0,0)[lt]{\lineheight{1.25}\smash{\begin{tabular}[t]{l}$2 \cdot \approxcost{\Tau}$\end{tabular}}}}}%
    \put(0.06,0.07){\includegraphics[width=\unitlength]{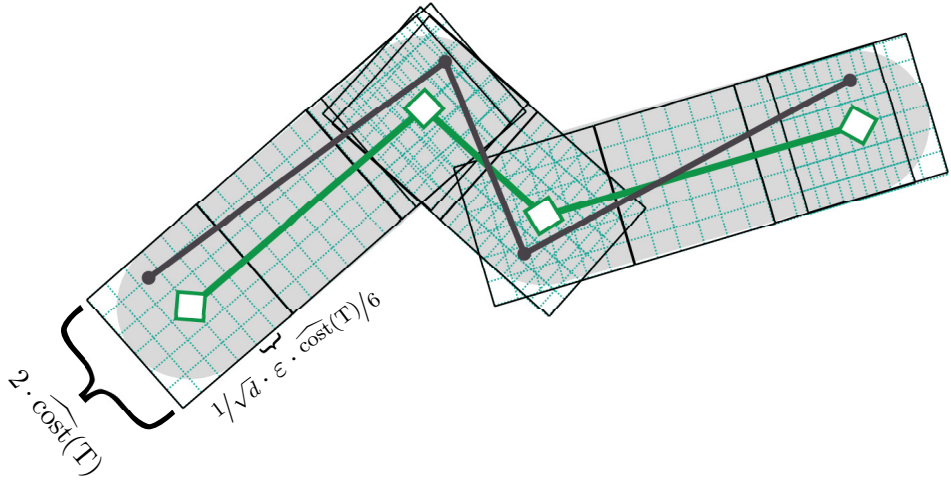}}%
    \put(0.24,0.08){\color[rgb]{0,0,0}\rotatebox{37}{\makebox(0,0)[lt]{\lineheight{1.25}\smash{\begin{tabular}[t]{l}$\nicefrac{1}{\sqrt{d}}\cdot \epsilon \cdot \nicefrac{\approxcost{\Tau}}{6}$\end{tabular}}}}}%
  \end{picture}%
\endgroup%

		\caption{Exemplary grid-cover of \cref{algo:klcenter_coreset_mg1} for $d=2$. A center-curve is depicted in green with cubes in black and associated grids in light blue. The area with the gray background is the union of the so called envelopes, \conferre \cref{lem:envelope_cover}. A curve with Fréchet distance less than $ \approxcost{\Tau}$ is also depicted. It can be observed that the vertices of this curve lie in at least one cell of a grid.}
		\label{fig:partitioning}
	\end{figure}

	For $i \in \{1, \dots, k\}$ and $j \in \{1, \dots, b_i-1\}$ let $E^{\prime\prime}_{i,j} \eqdef \{ p \in \euclideanspace^d \mid \exists q \in e^{\prime\prime}_{i,j}: \eucl{p}{q} \leq \approxcost{\Tau} \}$ be the envelope of $e^{\prime\prime}_{i,j}$ (\conferre \cref{fig:partitioning}) of radius $\approxcost{\Tau}$. We now show that for each $i \in \{1, \dots, k\}$ and $j \in \{ 1, \dots, b_i - 1 \}$ the algorithm covers $E^{\prime\prime}_{i,j}$ with grids: By the definition of $E^{\prime\prime}_{i,j}$ we need at least one grid around the initial point of $e^{\prime\prime}_{i,j}$ and one grid around its end point in addition to at least $\frac{\eucl{v_{i,j}}{v_{i,j+1}}-2 \cdot \approxcost{\Tau}}{2 \cdot \approxcost{\Tau}}$ many, to cover the range of $e^{\prime\prime}_{i,j}$ that remains uncovered, which is of length $\eucl{v_{i,j}}{v_{i,j+1}}-2 \cdot \approxcost{\Tau}$. It can be observed that $\beta \eqdef \left\lceil 1 + \frac{\eucl{v_{i,j}}{v_{i,j+1}}}{2 \cdot \approxcost{\Tau}} \right\rceil$ is an upper bound for this number. Because $e^{\prime\prime}_{i,j}$ is aligned on the non-negative section of the $d$\textsuperscript{th} axis, we cover $E^{\prime\prime}_{i,j}$ by translating the grids along the $d$\textsuperscript{th} axis, starting from the initial point of $e^{\prime\prime}_{i,j}$, which is the origin. In \cref{algo:klcenter_coreset_mg1:line13} to \cref{algo:klcenter_coreset_mg1:line17} \cref{algo:klcenter_coreset_mg1} defines for $i \in \{1, \dots, k\}$, $j \in \{1, \dots, b_i-1\}$ and $s \in \{1, \dots, \beta\}$ the vectors $\vv{y}_{i,j,s} \eqdef (y_{i,j,s,1}, \dots, y_{i,j,s,d})$, where $y_{i,j,s,r} \eqdef 0$, for $r \in \{0, \dots, d-1\}$, and $y_{i,j,s,d} \eqdef (s-1) \cdot 2 \cdot \approxcost{\Tau}$. Also, it defines the cubes $\mathfrak{C}_{i,j,s} \eqdef \cube{\trans{o}{\vv{y}_{i,j,s}}, 2 \cdot \approxcost{\Tau}}$ which are translated along the $d$\textsuperscript{th} axis and the associated grids $\mathfrak{G}_{i,j,s} \eqdef \grid{\mathfrak{C}_{i,j,s}, \nicefrac{1}{\sqrt{d}}\cdot\epsilon\cdot\nicefrac{\approxcost{\Tau}}{6}}$ of cell length $\nicefrac{1}{\sqrt{d}}\cdot\epsilon\cdot\nicefrac{\approxcost{\Tau}}{6}$. Recall that by \cref{obs:cube_grid} it holds that $\mathfrak{C}_{i,j,s} = \mathfrak{G}_{i,j,s}$. Now because it can be observed that $\cup_{s=1}^{\beta} \mathfrak{G}_{i,j,s} = \cup_{p \in e^{\prime\prime}_{i,j}} \cube{p, 2 \cdot \approxcost{\Tau}}$ and $E^{\prime\prime}_{i,j} = \cup_{p \in e^{\prime\prime}_{i,j}} \{ q \in \euclideanspace^d \mid \eucl{p}{q} \leq \approxcost{\Tau} \}$, we have that
	\begin{align*}
	E^{\prime\prime}_{i,j} \subseteq \bigcup_{s=1}^{\beta} \mathfrak{G}_{i,j,s},
	\end{align*}
	by \cref{obs:ball_grid}. Thus, it holds that for every $i \in \{1,\dots,k\}$ and every $j \in \{1, \dots, b_i-1\}$ the envelope $E^{\prime\prime}_{i,j}$ is covered by grids.
	
	For $i \in \{1, \dots, k\}$ and $j \in \{1, \dots, b_i - 1\}$ let $E_{i,j} \eqdef \{ p \in \euclideanspace^d \mid \exists q \in e_{i,j}: \eucl{p}{q} \leq \approxcost{\Tau} \}$ be the envelope of $e_{i,j}$ and for $p \in \euclideanspace^d$ let $D_{i,j}(p) \eqdef (D_{i,j,d} \circ \dots \circ D_{i,j,1} \circ D_{i,j,0})(p)$ be the composition of the transformations \cref{algo:klcenter_coreset_mg1} does in \cref{algo:klcenter_coreset_mg1:line7} to \cref{algo:klcenter_coreset_mg1:line12}, also let $D^{-1}_{i,j}(p) \eqdef (D^{-1}_{i,j,0} \circ D^{-1}_{i,j,1} \circ \dots \circ D^{-1}_{i,j,d})(p)$ be the inverse of $D_{i,j}(p)$. In \cref{algo:klcenter_coreset_mg1:line19} to \cref{algo:klcenter_coreset_mg1:line26} \cref{algo:klcenter_coreset_mg1} applies $D^{-1}_{i,j}(\cdot)$ to every point in every grid. By \cref{prop:translation_isometry} and \cref{prop:rotation_isometry} the points in the grids now have same distances to the points in $E_{i,j}$ as they had to the points in $E^{\prime\prime}_{i,j}$, therefore
	\begin{align*}
	E_{i,j} \subseteq \bigcup_{s=1}^{\beta} D^{-1}_{i,j}\left(\mathfrak{G}_{i,j,s}\right)
	\end{align*}
	holds, where $D^{-1}_{i,j}\left(\mathfrak{G}_{i,j,s}\right) \eqdef \{ D^{-1}_{i,j}(p) \mid p \in \mathfrak{G}_{i,j,s} \}$, hence for every $i \in \{1, \dots, k\}$ and every $j \in \{ 1, \dots, b_i-1\}$ the envelope $E_{i,j}$ is covered by the grids defined by \cref{algo:klcenter_coreset_mg1}. By \cref{def:frechet_distance} it is clear, that for every curve $\sigma \in \eqcfre{m}$ with $\frechet{\sigma}{\hat{c}_i} \leq \approxcost{\Tau}$, for an arbitrary $i \in \{1, \dots,k\}$, it holds that $\{ \sigma(t) \mid t \in [0,1] \} \subseteq \cup_{j=1}^{b_i-1} E_{i,j}$, which finishes the proof.
\end{proof}
\begin{lemma}
	\label{lem:axisangle}
	Let $p \eqdef (p_1, \dots, p_d) \in \euclideanspace^d$ be a point and $\overline{op}$ be the line segment from the origin to $p$. For every two points $q_1, q_2 \in \overline{op} \setminus \{ o \}$ and $i \in \{1, \dots, d-1\}$ it holds that $\axisangle{\vv{oq_1}}{i} = \axisangle{\vv{oq_2}}{i}$.
\end{lemma}
\begin{proof}
	We have that $\overline{op} = \{ (1-\gamma)o + \gamma p \mid \gamma \in [0,1] \} = \{ \gamma p \mid \gamma \in [0,1] \}$. Let $\gamma_1 \eqdef \frac{q_1}{p}$ and $\gamma_2 \eqdef \frac{q_2}{p}$. By definition, we have that $\gamma_1 > 0$ and $\gamma_2 > 0$.
	
	Let $\vv{a_i}$ and $\vv{oq_1}^\prime$, as well as $\vv{oq_2}^\prime$, be defined as in \cref{def:axisangle}. We obtain:
	\begin{align*}
		\axisangle{\vv{oq_1}^\prime}{i} ={} & \frac{\scalarm{\vv{oq_1}^\prime}{\vv{a}_i}}{\norm{\vv{oq_1}^\prime}\cdot\norm{\vv{a}_i}} \\
		={} & \frac{\gamma_1 \cdot p_i}{\sqrt{\gamma_1^2}\sqrt{\sum_{j=1}^{i-1} 0^2 + p_i^2 + p_{i+1}^2 + \sum_{i+2}^{d} 0^2}} = \frac{p_i}{\sqrt{\sum_{j=1}^{i-1} 0^2 + p_i^2 + p_{i+1}^2 + \sum_{i+2}^{d} 0^2}} \\
		={} & \frac{\gamma_2 \cdot p_i}{\sqrt{\gamma_2^2}\sqrt{\sum_{j=1}^{i-1} 0^2 + p_i^2 + p_{i+1}^2 + \sum_{i+2}^{d} 0^2}}
		= \frac{\scalarm{\vv{oq_2}^\prime}{\vv{a}_i}}{\norm{\vv{oq_2}^\prime}\cdot\norm{\vv{a}_i}} = \axisangle{\vv{oq_2}^\prime}{i}
	\end{align*}
\end{proof}
\begin{lemma}
	\label{lem:rotation}
	Let $p \in \euclideanspace^d \setminus \{ o \}$ be an arbitrary point and $i \in \{1, \dots,d-1\}$ be arbitrary. Let \[q \eqdef (q_1, \dots, q_d) \eqdef \rotate{p}{i}{\axisangle{\vv{op}}{i}},\] then it holds that $q_i = 0$ and $q_{i+1} \geq 0$.
\end{lemma}
\begin{proof}
	Consider the plane spanned by the $i$\textsuperscript{th} and $(i+1)$\textsuperscript{th} dimension, which we denote by $\mathcal{P}$. Let $\vv{a}_i \eqdef (a_{i,1}, \dots, a_{i,d})$, where we define $a_{i,j} \eqdef 0$ for $j \in \{1, \dots, d\} \setminus \{i\}$ and $a_{i,i} \eqdef 1$. Also let $\vv{a}_{i+1} \eqdef (a_{i+1,1}, \dots, a_{i+1,d})$, where we define $a_{i,j} \eqdef 0$ for $j \in \{1, \dots, d\} \setminus \{i+1\}$ and $a_{i+1,i+1} \eqdef 1$. Clearly $\vv{a}_i$ lies on the non-negative section of the $i$\textsuperscript{th} axis and $\vv{a}_{i+1}$ lies on the non-negative section of the $(i+1)$\textsuperscript{th} axis (and $\vv{a}_i$ is orthogonal to $\vv{a}_{i+1}$ and vice versa). Now we have four cases:
	\paragraph{Case 1:} The projection of $p$ onto $\mathcal{P}$, denoted by $p^\prime \eqdef (p^\prime_1, \dots, p^\prime_d)$, lies in the first quadrant, \ie $p^\prime_i \geq 0$ and $p^\prime_{i+1} \geq 0$. We have that $\axisangle{\vv{op}}{i}$ is $\frac{\pi}{2}$ minus the angle between the position vector $\vv{op^\prime}$ and $\vv{a}_i$, which gives the angle between $\vv{op^\prime}$ and $\vv{a}_{i+1}$ (\conferre \cref{fig:rot_case1}). By rotating $p$ counter-clockwise by this angle we obtain $q$ for which the position vector $\vv{oq^\prime}$, where $q^\prime$ is the projection of $q$ onto $\mathcal{P}$, is orthogonal to $\vv{a}_i$ and therefore to the $i$\textsuperscript{th} axis (just as the $(i+1)$\textsuperscript{th} axis is). In conclusion $q^\prime$ lies on the non-negative section of the $(i+1)$\textsuperscript{th} axis, hence $q_{i} = 0$ and $q_{i+1} \geq 0$.
	\begin{figure}
		\centering
		\def\svgwidth{0.5\textwidth}
\begingroup%
  \makeatletter%
  \providecommand\color[2][]{%
    \errmessage{(Inkscape) Color is used for the text in Inkscape, but the package 'color.sty' is not loaded}%
    \renewcommand\color[2][]{}%
  }%
  \providecommand\transparent[1]{%
    \errmessage{(Inkscape) Transparency is used (non-zero) for the text in Inkscape, but the package 'transparent.sty' is not loaded}%
    \renewcommand\transparent[1]{}%
  }%
  \providecommand\rotatebox[2]{#2}%
  \newcommand*\fsize{\dimexpr\f@size pt\relax}%
  \newcommand*\lineheight[1]{\fontsize{\fsize}{#1\fsize}\selectfont}%
  \ifx\svgwidth\undefined%
    \setlength{\unitlength}{333.8771913bp}%
    \ifx\svgscale\undefined%
      \relax%
    \else%
      \setlength{\unitlength}{\unitlength * \real{\svgscale}}%
    \fi%
  \else%
    \setlength{\unitlength}{\svgwidth}%
  \fi%
  \global\let\svgwidth\undefined%
  \global\let\svgscale\undefined%
  \makeatother%
  \begin{picture}(1,0.94470604)%
    \lineheight{1}%
    \setlength\tabcolsep{0pt}%
    \put(0,0){\includegraphics[width=\unitlength]{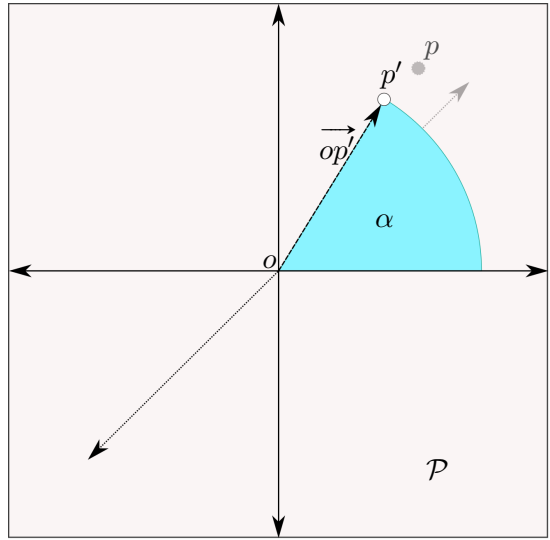}}%
    \put(0.54390726,0.67054846){\color[rgb]{0,0,0}\makebox(0,0)[lt]{\lineheight{1.25}\smash{\begin{tabular}[t]{l}$\vv{op^\prime}$\end{tabular}}}}%
    \put(0.73182624,0.10448845){\color[rgb]{0,0,0}\makebox(0,0)[lt]{\lineheight{1.25}\smash{\begin{tabular}[t]{l}$\mathcal{P}$\end{tabular}}}}%
    \put(0.72885181,0.85636272){\color[rgb]{0.30196078,0.30196078,0.30196078}\makebox(0,0)[lt]{\lineheight{1.25}\smash{\begin{tabular}[t]{l}$p$\end{tabular}}}}%
    \put(0.64228905,0.54951581){\color[rgb]{0,0,0}\makebox(0,0)[lt]{\lineheight{1.25}\smash{\begin{tabular}[t]{l}$\alpha$\end{tabular}}}}%
    \put(0.65239836,0.80412757){\color[rgb]{0,0,0}\makebox(0,0)[lt]{\lineheight{1.25}\smash{\begin{tabular}[t]{l}$p^\prime$\end{tabular}}}}%
    \put(0.44480162,0.4777958){\color[rgb]{0,0,0}\makebox(0,0)[lt]{\lineheight{1.25}\smash{\begin{tabular}[t]{l}$o$\end{tabular}}}}%
  \end{picture}%
\endgroup%

		\caption{Depiction of case 1 in the proof of \cref{lem:rotation} for $d=3$ and $i=1$. $\alpha$ is the angle between the position vector of $p^\prime$, which is the projection of $p$ onto the plane $\mathcal{P}$, and the $i$\textsuperscript{th} axis, therefore $\nicefrac{\pi}{2} - \alpha$ is the angle between the position vector of $p^\prime$ and the $(i+1)$\textsuperscript{th} axis.}
		\label{fig:rot_case1}
	\end{figure}
	The following cases only differ in the angle and direction of the rotation, therefore we just show that for these cases $\vv{oq^\prime}$ is also orthogonal to $\vv{a}_i$ in the direction of $\vv{a}_{i+1}$.
	\paragraph{Case 2:} $p^\prime$ lies in the second quadrant, \ie $p^\prime_i < 0$ and $p^\prime_{i+1} \geq 0$. We have that $\axisangle{\vv{op}}{i}$ is $\frac{\pi}{2}$ minus the angle between $\vv{op^\prime}$ and $\vv{a}_i$, which is greater than $\frac{\pi}{2}$. The resulting angle equals the negative angle between $\vv{op^\prime}$ and $\vv{a}_{i+1}$ (\conferre \cref{fig:rot_case2}), therefore the rotation changes its direction and again we obtain that $\vv{oq^\prime}$ is orthogonal to $\vv{a}_i$ in the direction of $\vv{a}_{i+1}$.
	\begin{figure}
		\centering
		\def\svgwidth{0.5\textwidth}
\begingroup%
  \makeatletter%
  \providecommand\color[2][]{%
    \errmessage{(Inkscape) Color is used for the text in Inkscape, but the package 'color.sty' is not loaded}%
    \renewcommand\color[2][]{}%
  }%
  \providecommand\transparent[1]{%
    \errmessage{(Inkscape) Transparency is used (non-zero) for the text in Inkscape, but the package 'transparent.sty' is not loaded}%
    \renewcommand\transparent[1]{}%
  }%
  \providecommand\rotatebox[2]{#2}%
  \newcommand*\fsize{\dimexpr\f@size pt\relax}%
  \newcommand*\lineheight[1]{\fontsize{\fsize}{#1\fsize}\selectfont}%
  \ifx\svgwidth\undefined%
    \setlength{\unitlength}{333.8771913bp}%
    \ifx\svgscale\undefined%
      \relax%
    \else%
      \setlength{\unitlength}{\unitlength * \real{\svgscale}}%
    \fi%
  \else%
    \setlength{\unitlength}{\svgwidth}%
  \fi%
  \global\let\svgwidth\undefined%
  \global\let\svgscale\undefined%
  \makeatother%
  \begin{picture}(1,0.94470604)%
    \lineheight{1}%
    \setlength\tabcolsep{0pt}%
    \put(0,0){\includegraphics[width=\unitlength]{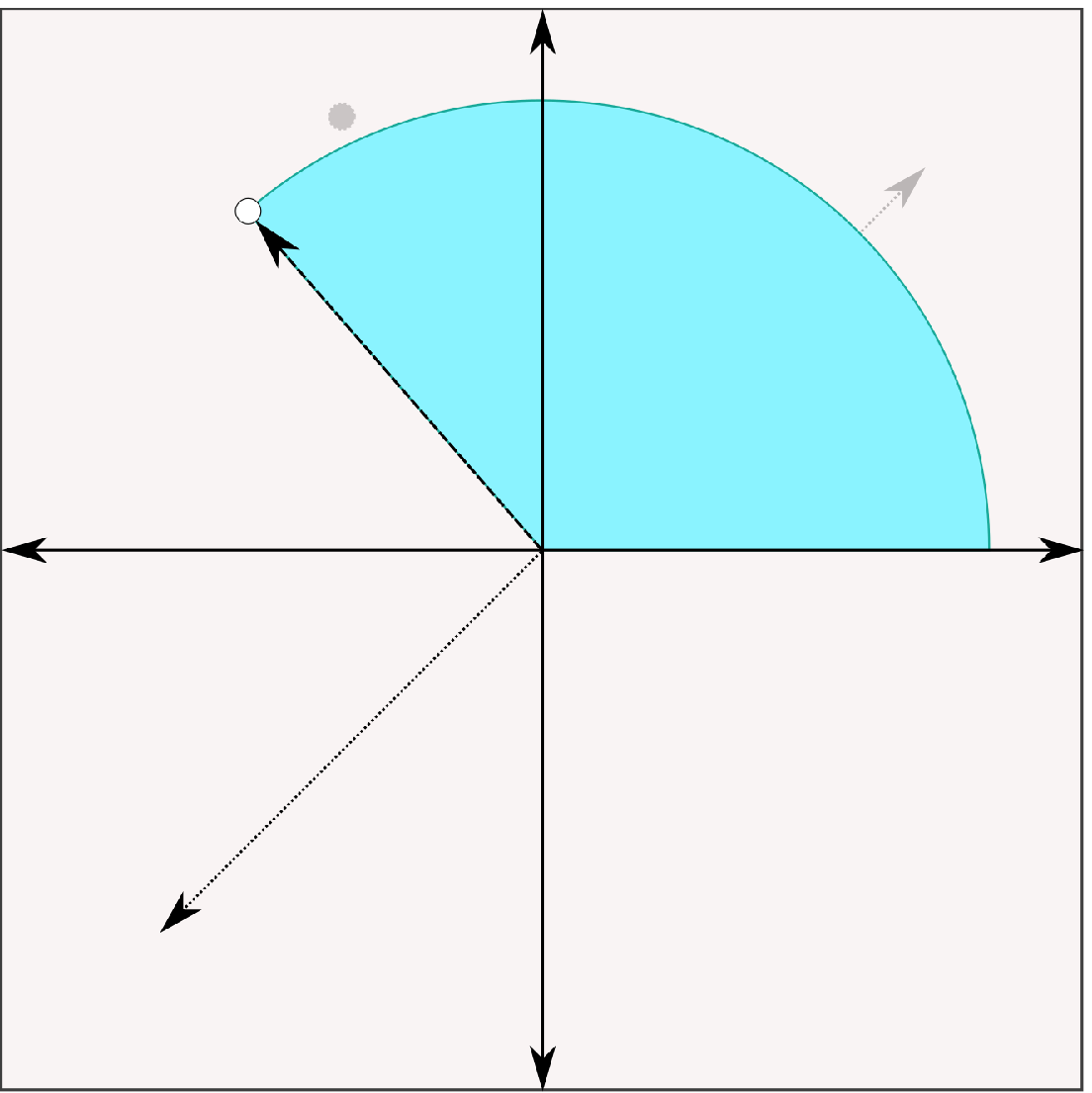}}%
    \put(0.3300854,0.65327175){\color[rgb]{0,0,0}\makebox(0,0)[lt]{\lineheight{1.25}\smash{\begin{tabular}[t]{l}$\vv{op^\prime}$\end{tabular}}}}%
    \put(0.73182624,0.10448845){\color[rgb]{0,0,0}\makebox(0,0)[lt]{\lineheight{1.25}\smash{\begin{tabular}[t]{l}$\mathcal{P}$\end{tabular}}}}%
    \put(0.28259716,0.88018869){\color[rgb]{0.30196078,0.30196078,0.30196078}\makebox(0,0)[lt]{\lineheight{1.25}\smash{\begin{tabular}[t]{l}$p$\end{tabular}}}}%
    \put(0.58564232,0.65753522){\color[rgb]{0,0,0}\makebox(0,0)[lt]{\lineheight{1.25}\smash{\begin{tabular}[t]{l}$\alpha$\end{tabular}}}}%
    \put(0.21125064,0.79691043){\color[rgb]{0,0,0}\makebox(0,0)[lt]{\lineheight{1.25}\smash{\begin{tabular}[t]{l}$p^\prime$\end{tabular}}}}%
    \put(0.47559528,0.44683153){\color[rgb]{0,0,0}\makebox(0,0)[lt]{\lineheight{1.25}\smash{\begin{tabular}[t]{l}$o$\end{tabular}}}}%
  \end{picture}%
\endgroup%

		\caption{Depiction of case 2 in \cref{lem:rotation} for $d=3$ and $i=1$. $\alpha$ is the angle between the position vector of $p^\prime$, which is the projection of $p$ onto the plane $\mathcal{P}$, and the $i$\textsuperscript{th} axis, therefore $\nicefrac{\pi}{2} - \alpha$ is the negative angle between the position vector of $p^\prime$ and the $(i+1)$\textsuperscript{th} axis.}
		\label{fig:rot_case2}
	\end{figure}
	\paragraph{Case 3:} $p^\prime$ lies in the third quadrant, \ie $p_i < 0$ and $p_{i+1} < 0$. We have that $\axisangle{\vv{op}}{i}$ is $\frac{\pi}{2}$ plus the angle between $\vv{op^\prime}$ and $\vv{a}_i$, which again is the angle between $\vv{op^\prime}$ and $\vv{a}_{i+1}$. Rotating by this angle again yields that $\vv{oq^\prime}$ is orthogonal to $\vv{a}_i$ in the direction of $\vv{a}_{i+1}$.
	\paragraph{Case 4:} $p^\prime$ lies in the fourth quadrant, \ie $p_i \geq 0$ and $p_{i+1} < 0$ (\conferre \cref{fig:rot_case4}). This case is completely analogous to the previous case.
	\begin{figure}
		\centering
		\def\svgwidth{0.5\textwidth}
\begingroup%
  \makeatletter%
  \providecommand\color[2][]{%
    \errmessage{(Inkscape) Color is used for the text in Inkscape, but the package 'color.sty' is not loaded}%
    \renewcommand\color[2][]{}%
  }%
  \providecommand\transparent[1]{%
    \errmessage{(Inkscape) Transparency is used (non-zero) for the text in Inkscape, but the package 'transparent.sty' is not loaded}%
    \renewcommand\transparent[1]{}%
  }%
  \providecommand\rotatebox[2]{#2}%
  \newcommand*\fsize{\dimexpr\f@size pt\relax}%
  \newcommand*\lineheight[1]{\fontsize{\fsize}{#1\fsize}\selectfont}%
  \ifx\svgwidth\undefined%
    \setlength{\unitlength}{333.8771913bp}%
    \ifx\svgscale\undefined%
      \relax%
    \else%
      \setlength{\unitlength}{\unitlength * \real{\svgscale}}%
    \fi%
  \else%
    \setlength{\unitlength}{\svgwidth}%
  \fi%
  \global\let\svgwidth\undefined%
  \global\let\svgscale\undefined%
  \makeatother%
  \begin{picture}(1,0.94470604)%
    \lineheight{1}%
    \setlength\tabcolsep{0pt}%
    \put(0,0){\includegraphics[width=\unitlength]{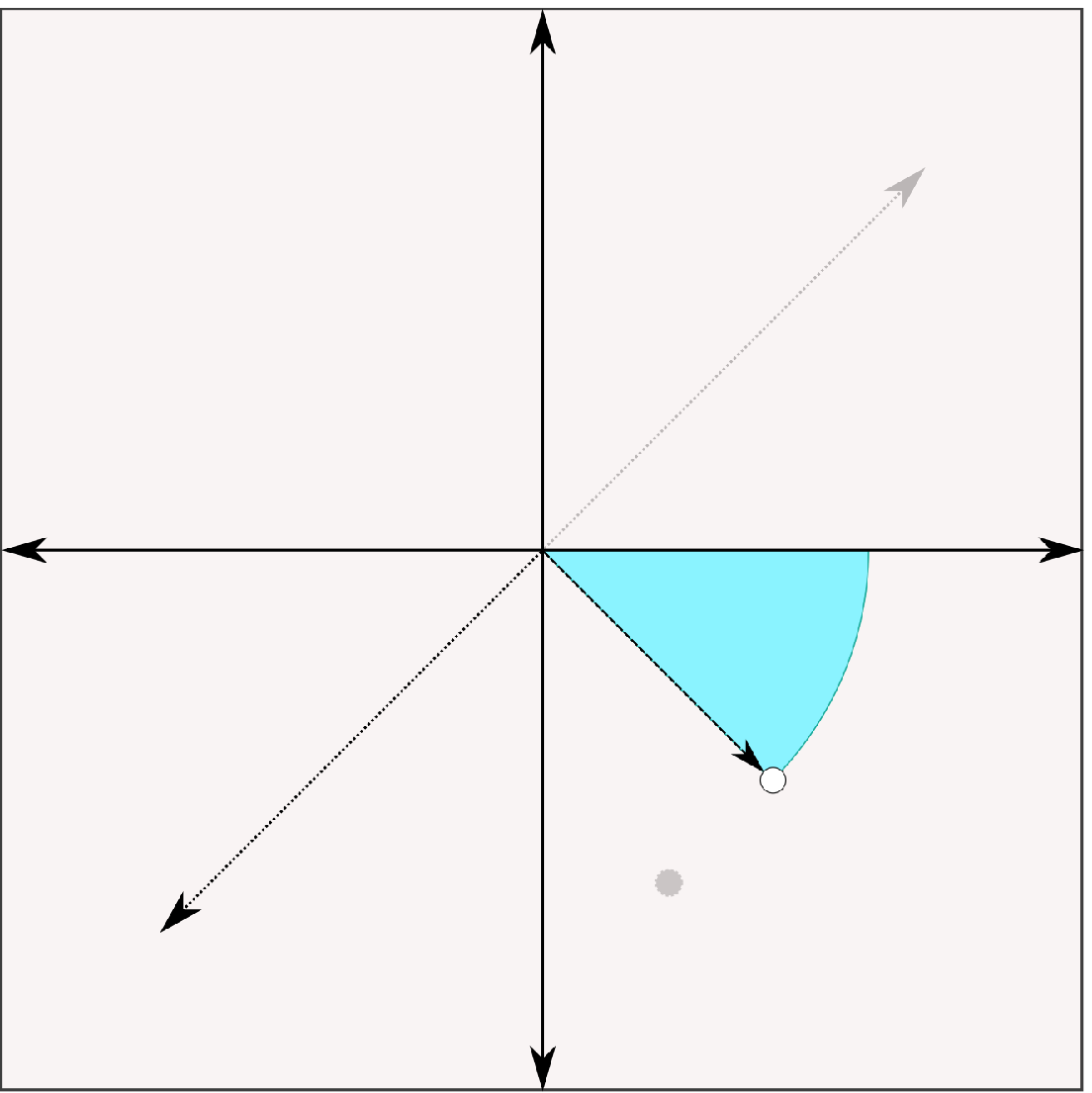}}%
    \put(0.44480104,0.47779975){\color[rgb]{0,0,0}\makebox(0,0)[lt]{\lineheight{1.25}\smash{\begin{tabular}[t]{l}$o$\end{tabular}}}}%
    \put(0.48346904,0.3358108){\color[rgb]{0,0,0}\makebox(0,0)[lt]{\lineheight{1.25}\smash{\begin{tabular}[t]{l}$\vv{op^\prime}$\end{tabular}}}}%
    \put(0.73182624,0.10448845){\color[rgb]{0,0,0}\makebox(0,0)[lt]{\lineheight{1.25}\smash{\begin{tabular}[t]{l}$\mathcal{P}$\end{tabular}}}}%
    \put(0.56059219,0.13855825){\color[rgb]{0.30196078,0.30196078,0.30196078}\makebox(0,0)[lt]{\lineheight{1.25}\smash{\begin{tabular}[t]{l}$p$\end{tabular}}}}%
    \put(0.6169782,0.41030632){\color[rgb]{0,0,0}\makebox(0,0)[lt]{\lineheight{1.25}\smash{\begin{tabular}[t]{l}$\alpha$\end{tabular}}}}%
    \put(0.66041637,0.22091373){\color[rgb]{0,0,0}\makebox(0,0)[lt]{\lineheight{1.25}\smash{\begin{tabular}[t]{l}$p^\prime$\end{tabular}}}}%
  \end{picture}%
\endgroup%

		\caption{Depiction of case 4 in \cref{lem:rotation} for $d=3$ and $i=1$. $\alpha$ is the angle between the position vector of $p^\prime$, which is the projection of $p$ onto the plane $\mathcal{P}$, and the $i$\textsuperscript{th} axis, therefore $\nicefrac{\pi}{2} + \alpha$ is the angle between the position vector of $p^\prime$ and the $(i+1)$\textsuperscript{th} axis.}
		\label{fig:rot_case4}
	\end{figure}
	Together these cases prove the claim.
\end{proof}

\subsubsection{Time Complexity Analysis of \cref{algo:klcenter_coreset_mg1}}
\begin{theorem}
	Given a set $\Tau \subset \eqcfre{m}$ of $n$ polygonal curves and a parameter $\epsilon \in (0,1)$, \cref{algo:klcenter_coreset_mg1} has running-time $\On{\left(2^{3m} \cdot n^{1.5} \cdot \frac{l^{12d^2m}m}{\epsilon^{dm}} + 2^m m^{m+1} n \right) + nm \log(m) + m^3 \log(m)}$, if successful, else $\On{mn \log(m) + m^3 \log(m)}$.
\end{theorem}
\begin{proof}
	Recall that \cref{algo:klcenter_approx}, which is initially run to obtain an approximate clustering $\hat{C} \eqdef \{ \hat{c}_1, \dots, \hat{c}_k \}$, has running-time $\On{mn \log(m) + m^3 \log(m)}$, \conferre \cref{theo:approx_algo_center}. Further, recall that \cref{algo:klcenter_coreset_mg1} fails after this step, if $\frac{\delta^m}{\approxcost{\Tau}^m} > \sqrt{n} \Leftrightarrow \frac{\delta}{\approxcost{\Tau}} > \sqrt[2m]{n}$, where $\approxcost{\Tau} \eqdef \cost{\Tau}{\hat{C}}$ and $\delta$ is the length of a longest edge of any center $\hat{c} \in \hat{C}$, \conferre \cref{algo:klcenter_coreset_mg1:lineret}. In \cref{algo:klcenter_coreset_mg1:line5} to \cref{algo:klcenter_coreset_mg1:line27} \cref{algo:klcenter_coreset_mg1} computes the grids required for covering the center-curves $\hat{c}_1, \dots, \hat{c}_k$. For every cell of every of those grids we save the following in an array $A_{c, i}$ (where $i \in \{1, \dots, k\}$ is the index of the center-curve of the current iteration), one entry per cell: The intervals (one for each dimension) the cell covers when it is \emph{axis-parallel} (at \cref{algo:klcenter_coreset_mg1:line17}) and the $\alpha_{i,j,1}, \dots, \alpha_{i,j,d-1}$, as well as $\vv{x}_{i,j}$, \conferre \cref{lem:envelope_cover}. Recall from \cref{theo:klcenter_coreset_mg1} that there are up to $k \cdot (l-1) \cdot \left\lceil 1 + \frac{\delta}{2 \approxcost{\Tau}} \right\rceil \cdot \frac{2^d 6^d d^{\nicefrac{d}{2}}}{\epsilon^d}$ many cells. Thus, \cref{algo:klcenter_coreset_mg1} takes time $\On{k \cdot (l-1)} = \On{1}$ for \cref{algo:klcenter_coreset_mg1:line6} to \cref{algo:klcenter_coreset_mg1:line8}, $\On{k \cdot (l-1) \cdot d} = \On{1}$ for \cref{algo:klcenter_coreset_mg1:line9} to \cref{algo:klcenter_coreset_mg1:line12} and $\On{k \cdot (l-1) \cdot \frac{\delta}{\approxcost{\Tau}} \cdot \frac{1}{\epsilon^d} \cdot d} = \On{\sqrt[2m]{n} \cdot \frac{1}{\epsilon^d}}$ for \cref{algo:klcenter_coreset_mg1:line13} to \cref{algo:klcenter_coreset_mg1:line17}. In conclusion \cref{algo:klcenter_coreset_mg1:line6} to \cref{algo:klcenter_coreset_mg1:line17} take time $\On{\sqrt[2m]{n} \cdot \frac{1}{\epsilon^d}}$.
	
	Now \cref{algo:klcenter_coreset_mg1:line19} to \cref{algo:klcenter_coreset_mg1:line26} are done implicitly in an implementation. We only need those to show that we can indeed cover the center-curves with grids, \conferre \cref{lem:envelope_cover}. Instead we do the following: For every $\tau \in \Tau$ and every vertex $v_{\tau, r} \in \{v_{\tau, 1}, \dots, v_{\tau, m}\}$ of $\tau$ (recall that every curve has $m$ vertices) we go through the list of the cells and apply $D_{i,j}$ (recall that we saved $\alpha_{i,j,1}, \dots, \alpha_{i,j,d-1}$, as well as $\vv{x}_{i,j}$) to $v_{\tau, r}$ (\conferre \cref{lem:envelope_cover}). Now we can check if the components of $D_{i,j}(v_{\tau, r})$ lie in the intervals of the respective cell we are looking at. If so, we save the id (respective to the list of cells) of the cell in a field $A_{\tau}[r]$ of an array $A_{\tau}[1, \dots, m]$. As the notation suggests, there is one such array per input-curve. This takes time $\On{n \cdot m \cdot \underbrace{d \cdot k \cdot (l-1) \cdot \frac{\delta}{2 \cdot \approxcost{\Tau}} \cdot \frac{2^d 6^d d^{\nicefrac{d}{2}}}{\epsilon^d} \cdot d}_{\text{Apply }D_{i,j}\text{ and check in which cell the vertex lies.}}} = \On{n \cdot m \cdot \sqrt[2m]{n} \cdot \frac{1}{\epsilon^d}}$.
	
	We can now realize \cref{algo:klcenter_coreset_mg1:line28} to \cref{algo:klcenter_coreset_mg1:line32} by iterating over all combinations of $m$ cells, of all grids that cover the center-curve we are currently looking at, and all permutations of those $m$ cells and then looking through the input-set whether there is a curve that has its vertices in those cells, connected in the order of the current permutation. We do this by looking up the ids of the current cells in order of the permutation in the $A_\tau$, where $\tau \in \Tau$ is the input curve we are currently looking at. Recall from \cref{theo:klcenter_coreset_mg1} that there are up to $\kappa = (l-1) \cdot \left\lceil\left( 1 + \frac{\delta}{2 \cdot \approxcost{\Tau}}\right)\right\rceil \cdot \frac{2^d 6^d d^{\nicefrac{d}{2}}}{\epsilon^d}$ cells per center-curve. Also recall, that we have at most $\frac{(\left(\kappa + m\right))^m}{m!}$ ways to put the vertices of $\tau$ into these cells and $m!$ possibilities to connect them with edges. Thus, \cref{algo:klcenter_coreset_mg1:line28} to \cref{algo:klcenter_coreset_mg1:line32} take time $\On{k \cdot (\kappa + m)^m \cdot n \cdot m \cdot d} = \On{\left(2^{3m} \cdot \sqrt{n} \cdot \frac{l^{12d^2m}}{\epsilon^{dm}} + 2^m m^m \right) \cdot n \cdot m}$, \conferre \cref{theo:klcenter_coreset_mg1} and recall that \cref{algo:klcenter_coreset_mg1} fails if $\frac{\delta^m}{\approxcost{\Tau}^m} > \sqrt{n}$. Together with the running-time of \cref{algo:klcenter_approx} this dominates the running-time of \cref{algo:klcenter_coreset_mg1}, which is then $\On{\left(2^{3m} \cdot n^{1.5} \cdot \frac{l^{12d^2m}m}{\epsilon^{dm}} + 2^m m^{m+1}n \right) + nm \log(m) + m^3 \log(m)}$.
\end{proof}
\clearpage
\section{The \klmedian Objective}
\label{sec:klmedian}
In this section we are working with \cref{def:klmedian}, \ie the \klmedian objective, and we are proving \cref{theo:algo_coreset_median}. We are given a set of curves $\Tau \eqdef \{ \tau_1, \dots, \tau_n \} \subset \eqcfre{m}$ and we are looking for a center-set $C \subseteq \Tau$, of cardinality $k$, that minimizes the sum of the distances between the curves in $\Tau$ and a respective nearest center. For the structure of the objective function sampling-techniques show to be beneficial. Therefore, we examine the probability space $(\Tau, \powerset{\Tau}, \varphi)$, where $\varphi$ is the uniform distribution over $\Tau$, thus for $\tau \in \Tau$ we define $\varphi(\tau) \eqdef \frac{1}{n}$ and for $\Tau^\prime \subseteq \Tau$ we define $\varphi(\Tau^\prime) \eqdef \sum_{\tau \in \Tau^\prime} \varphi(\tau)$. It is easy to define estimators for the cost of center-sets under the \klmedian objective. When the sampling is done with respect to $\varphi$ these estimators have high variance, though. We are going to apply the sensitivity sampling framework, \conferre \cref{subsec:sensitivity}, to obtain estimators with low variance. Note that, in contrast to the \klcenter and \klmeans objectives, we restrict ourselves to finding a subset of $\Tau$ instead of finding a subset of $\eqcfre{l}$, of cardinality \kk, \ie we are working with the \emph{discrete} median objective. Mainly we do so because our techniques do not work properly when we define the objective with respect to the latter. To be precise, when there are infinitely many possible center-sets then the probability that a sample is an \coreset for all these center-sets at a time (just as \cref{def:coreset} requires) tends to zero.

\subsection{An $\epsilon$-coreset Construction}

\cref{algo:klmedian_coreset} works as follows: It builds the probability distribution $\psi$ as it is defined in terms of the sensitivity sampling, \conferre \cref{subsec:sensitivity}. To obtain a set of curves that is an \coreset with constant probability it takes an independent sample from $\Tau$ of cardinality $\Om{\frac{\ln(n)}{\epsilon^2}}$ with respect to $\psi$ and weighs every curve by $\nicefrac{n}{\ell}$, where $\ell$ is the cardinality of the sample. This weighted sample-set is then a weighted \coreset with probability at least $\nicefrac{2}{3}$.

\begin{algorithm}
	\caption{Compute Weighted $\epsilon$-Coreset for the \klmedian Objective}\label{algo:klmedian_coreset}
	\begin{algorithmic}[1]
		\Procedure{k-median-coreset}{$\Tau, \epsilon$}
			\State $n \gets \vert \Tau \vert$
			\State $\psi \gets $\textsc{compute-psi}$(\Tau, n)$
			\State $\ell \gets  \left\lceil 640 \cdot (-\ln(\rho) + \ln(2)) \cdot \frac{k^2 \ln(n)}{\epsilon^2} \right\rceil$
			\State $S \gets$ Sample $\ell$ curves from $\Tau$ independently with respect to $\psi$  \label{algo:klmedian_coreset_line_sampling}
			\State Weigh every curve $s$ in $S$ by $\nicefrac{1}{\ell} \cdot \nicefrac{1}{\psi(s)}$
			\State \textbf{return} $S$
		\EndProcedure
		\Function{compute-psi}{$\Tau, n$}\label{func:compute_psi}
		\State $\hat{C} \eqdef \{ \hat{c}_1, \dots, \hat{c}_k \} \gets $\textsc{k-median-approx}$(\Tau, \nicefrac{1}{(k\cdot 3 \cdot n)} )$ \label{algo:klmedian_coreset_approx_line} \Comment{\cref{algo:klmedian_approx}}
		\For{$j=1, \dots, k$}
		\State $U_j \gets \cluster{\Tau}{\hat{C}}{\hat{c}_j}$
		\State $m_j \gets \frac{n}{\vert U_j \vert} \sum\limits_{\tau \in U_j} \frechet{\tau}{\hat{c}_j} \cdot \nicefrac{1}{n}$
		\EndFor
		\State $\Phi \gets \frac{1}{6} \sum_{j=1}^{k} \sum\limits_{\tau \in U_j} \frechet{\tau}{\hat{c}_j} \cdot \nicefrac{1}{n} $
		\For{$j=1, \dots, k$}
		\For{$\tau \in U_j$}
		\State $\psi(\tau) \gets \frac{1}{32kn} \cdot \left( \frac{2m_j + \frechet{\tau}{\hat{c}_j}}{\nicefrac{3}{4}\Phi} + \frac{8n}{\vert U_j \vert} \right)$
		\EndFor
		\EndFor
		\State \textbf{return} $\psi$
		\EndFunction
	\end{algorithmic}
\end{algorithm}

To build $\psi$, \cref{algo:klmedian_coreset} calls \cref{algo:klmedian_approx} (we will introduce the algorithm in the next subsection) to obtain a center-set which yields an approximate solution to the \klmedian objective. This center-set is then used to compute the clusters $U_1, \dots, U_k$ and the values $m_1, \dots, m_k$, as well as $\Phi$. These values suffice to construct $\psi$, as will be shown in \cref{theo:klmedian_coreset}.

\subsubsection{Correctness Analysis of \cref{algo:klmedian_coreset}}
\begin{theorem}[\citet{coresers_methods_history}]
	\label{theo:klmedian_coreset}
	Given a set $\Tau \eqdef \{ \tau_1, \dots, \tau_n \} \subset \eqcfre{m}$ of polygonal curves and a parameter $\epsilon \in (0,1)$, \cref{algo:klmedian_coreset} computes a set of cardinality $\On{\frac{\ln(n)}{\epsilon^2}}$, that is a weighted \coreset for the \klmedian objective with probability at least $\nicefrac{2}{3}$.
\end{theorem}
\begin{proof}
	Let $C_1, \dots, C_w$ be a sequence of all possible center-sets of cardinality $k$, that are subsets of $\Tau$, \ie $w \leq n^k$. For $i \in \{1, \dots, w\}$ and $\tau \in \Tau$ we define the random variables (\conferre \cref{def:random_variable}) \[X_i(\tau) \eqdef \min\limits_{c \in C_i} \frechet{\tau}{c},\] that are estimators for the cost of a curve with respect to $C_i$. $W \eqdef \{ X_1, \dots, X_w \}$ is the set of these random variables. Because the $X_i$ have high variance when we sample from $\Tau$ with respect to $\varphi$, we define according to the sensitivity framework: For a curve $\tau \in \Tau$ the sensitivity with respect to $W$ is \[\xi(\tau) \eqdef \max\limits_{i \in \{1, \dots, w\}} \frac{X_i(\tau)}{\E_\varphi[X_i]}\] and the total sensitivity of $W$ is \[\Xi \eqdef \sum\limits_{\tau \in \Tau} \xi(\tau) \cdot \varphi(\tau).\] For $\tau \in \Tau$ we define an upper bound $\lambda(\tau)$ on $\xi(\tau)$ later in the proof, thus \[\Lambda \eqdef \sum_{\tau \in \Tau} \lambda(\tau) \cdot \varphi(\tau)\] is then an upper bound on $\Xi$. Finally, for $\tau \in \Tau$, respective $\Tau^\prime \subseteq \Tau$ we define the probability function (\conferre \cref{prop:sum_one}) \[\psi(\tau) \eqdef \frac{\lambda(\tau)}{\Lambda} \varphi(\tau),\] \[ \psi(\Tau^\prime) \eqdef \sum_{\tau \in \Tau^\prime} \psi(\tau), \] which takes the sensitivities into account. Thus, for $i \in \{ 1, \dots, w\}$ and $\tau \in \Tau$ the random variables \[Y_i(\tau) \eqdef X_i(\tau) \frac{\varphi(\tau)}{\psi(\tau)}\] have low variance with respect to $\psi$, \conferre \cref{prop:var_sens}. 
	
	Now let $i \in \{1, \dots,w\}$ be arbitrary but fixed and for $j=1, \dots, \ell$ let $Z_j \eqdef Y_i$ be independent random variables with respect to $\psi$, that are copies of $Y_i$. Let $Z \eqdef \sum\limits_{j=1}^\ell Z_j$ be the sum of these random variables and further let $\hat{Z} \eqdef \frac{n}{\ell} \cdot Z$ be a random variable that is a reweighing of $Z$ by $\nicefrac{n}{\ell}$. We show that $\hat{Z}$ is an estimator for $\cost{\Tau}{C_i}$. By \cref{prop:exp_equal} and the linearity of expectation we obtain (recall that for $\tau \in \Tau$: $\varphi(\tau) = \frac{1}{n}$):
	\begin{align*}
		\E_{\psi}[\hat{Z}] ={} & \frac{n}{\ell} \cdot \E_\psi[Z] = \frac{n}{\ell} \sum_{j=1}^\ell \E_\psi[Z_j] = \frac{n}{\ell} \sum_{j=1}^\ell \E_\varphi[X_i] = \frac{1}{\ell} \sum_{j=1}^\ell \sum_{r=1}^n \min\limits_{c \in C_i} \frechet{\tau_r}{c} \\
		={} & \cost{\Tau}{C_i}
	\end{align*} 
	Thus, $\hat{Z}$ indeed is an unbiased estimator for $\cost{\Tau}{C_i}$. We now prove that a sample of $\Tau$, of cardinality $\ell$ with respect to $\psi$, where the curves are weighted by $\frac{n}{\ell}$, indeed is a weighted \coreset for the \klmedian objective with constant probability.
	
	The random variables $Z_j-\E_{\psi}[Z_j]$ and $-Z_j + \E_{\psi}[Z_j]$ have zero mean. Because $Z_j(\tau) = Y_i(\tau) = X_i(\tau) \frac{\varphi(\tau)}{\psi(\tau)}$ we can bound $Z_j(\tau)$ as follows: \[ 0 \leq Z_j(\tau) = X_i(\tau) \frac{\Lambda}{\lambda(\tau)} \leq X_i(\tau) \frac{\Lambda}{\max\limits_{r \in \{ 1, \dots, w \}} \frac{X_r(\tau)}{\E_\varphi[X_r]} } \leq X_i(\tau) \frac{\Lambda}{\frac{X_i(\tau)}{\E_\varphi[X_i]}} = \E_\varphi[X_i] \Lambda \] Therefore $Z_j - \E_{\psi}[Z_j]$ takes values in the interval \[[-\E_\psi[Z_j], (\Lambda-1)\E_\psi[Z_j]]\] and $-Z_j + \E_{\psi}[Z_j]$ takes values in the interval \[[(-\Lambda + 1) \E_{\psi}[Z_j], \E_{\psi}[Z_j]],\] hence it holds for all $j \in \{1,\dots,\ell\}$ that \[\vert Z_j - \E_\psi[Z_j] \vert \leq \Lambda\E_\psi[Z_j].\] Also, it holds for all $j \in \{1,\dots,\ell\}$ that \[\vert -Z_j + \E_\psi[Z_j] \vert \leq \Lambda\E_\psi[Z_j].\] Let $\overline{E_i} \eqdef \vert \hat{Z} - \cost{\Tau}{C_i} \vert > \epsilon \cdot \cost{\Tau}{C_i}$ be the event, that $\hat{Z} > (1+\epsilon) \cost{\Tau}{C_i}$ or $\hat{Z} < (1-\epsilon) \cost{\Tau}{C_i}$ and let $E_i$ be the contrary event, \ie the sample, which is the set of simple events (that are the curves in $\Tau$) that lead to the value of $\hat{Z}$, is a weighted \coreset[] (\conferre \cref{def:coreset}). We can now apply Bernstein's inequality (\conferre \cref{theo:bern}).
	\begin{align*}
		\Pr[\overline{E_i}] ={} & \Pr[\vert \hat{Z} - \cost{\Tau}{C_i} \vert > \epsilon \cdot \cost{\Tau}{C_i}] \\
		={} &  \Pr[\vert \frac{1}{\ell} Z - \frac{1}{\ell} \E_\psi[Z] \vert > \epsilon \cdot \frac{1}{\ell} \E_\psi[Z] ] = \Pr[\vert Z - \E_\psi[Z] \vert > \epsilon \cdot \E_\psi[Z] ] \\
		\leq{} & 2\exp\left(- \frac{\epsilon^2 \E_{\psi}[Z]^2}{2\sum_{j=1}^\ell \Var_\psi[Z_j] + \nicefrac{2}{3} \Lambda \E_{\psi}[Z_j] \epsilon \E_{\psi}[Z]}\right) \\
		\leq{} & 2\exp\left(-\frac{\epsilon^2 \ell^2 \E_{\psi}[Z_j]^2}{2\ell \cdot (\Lambda-1)\E_{\psi}[Z_j]^2 + \nicefrac{2}{3} \ell \epsilon \Lambda \E_{\psi}[Z_j]^2}\right) \label{mt:eq1} \tag{I} \\
		={} & 2\exp\left(- \frac{\epsilon^2\ell}{2(\Lambda-1) + \nicefrac{2}{3} \epsilon \Lambda}\right) \\
		={} & 2 \exp\left(- \frac{\epsilon^2 \ell}{(2 + \nicefrac{2\epsilon}{3}) \Lambda - 2} \right)
	\end{align*}
	Here \cref{mt:eq1} holds because by \cref{prop:var_sens} it holds that $\Var_\psi[Y_i] \leq (\Lambda-1) \cdot \E_\psi[Y_i]^2$ for all $i \in \{1, \dots, w\}$.
	
	We bound $\ell$ to obtain a weighted \coreset with probability at least $\nicefrac{2}{3}$.
	\begin{align*}
		&&2 \exp\left(\frac{-\epsilon^2 \ell}{(2 + \nicefrac{2\epsilon}{3}) \Lambda - 2} \right) \leq{} & \frac{\rho}{n^k} \\
		\Leftrightarrow && \frac{\epsilon^2 \ell}{(2 + \nicefrac{2\epsilon}{3}) \Lambda - 2} \geq{} & -\ln(\rho) + k \ln(n) + \ln(2) \\
		\Leftrightarrow && \ell \geq{} & \frac{\left(-\ln(\rho) + k \ln(n) + \ln(2)\right) \cdot \left((2 + \nicefrac{2\epsilon}{3}) \Lambda - 2\right)}{\epsilon^2} \\
		\Leftarrow && \ell \geq{} & \frac{\left(-\ln(\rho) + k \ln(n) + \ln(2)\right) \cdot \left(\frac{6 + 2\epsilon}{3}\right) \cdot (3 \alpha + 2\sqrt{6\alpha} + 2) k}{\epsilon^2} \label{mt:eq2} \tag{II} \\
		\Leftarrow && \ell \geq{} & \frac{\left(-\ln(\rho) + k \ln(n) + \ln(2)\right) \cdot \left(\frac{6 + 2\epsilon}{3}\right) \cdot 32 k}{\epsilon^2} \label{mt:eq3} \tag{III} \\
		\Leftrightarrow && \ell \geq{} & \frac{(-32k\ln(\rho) + 32k^2 \ln(n) + 32k \ln(2)) \cdot \left(\frac{6 + 2\epsilon}{3}\right)}{\epsilon^2} \\
		\Leftrightarrow && \ell \geq{} & \frac{-64k\ln(\rho)(1 + \nicefrac{1}{3} \epsilon) + 64k^2 \ln(n)(1 + \nicefrac{1}{3} \epsilon) + 64k \ln(2)(1 + \nicefrac{1}{3} \epsilon)}{\epsilon^2} \\
		\Leftarrow && \ell \geq{} & 64 \cdot (-\ln(\rho) + \ln(2)) \cdot \frac{1 + \epsilon + k^2 \ln(n) + k^2 \ln(n) \epsilon + k + \epsilon k}{\epsilon^2} \\
		\Leftarrow && \ell \geq{} & 640 \cdot (-\ln(\rho) + \ln(2)) \cdot k^2 \cdot \frac{\ln(n)}{\epsilon^2}
	\end{align*}
	Here \cref{mt:eq2} holds, because in \cref{lem:total_sens_bound} we will show that for an $\alpha$-approximate center-set $\hat{C} \eqdef \{ \hat{c}_1, \dots, \hat{c}_k \}$ for $\Tau$, \ie $\cost{\Tau}{\hat{C}} \leq \alpha \cdot \optcost{\Tau}$, $\lambda$ can be defined such that $\Lambda = (3 \alpha + 2\sqrt{6\alpha} + 2) k$. \cref{mt:eq3} holds because we will show in \cref{theo:klmedian_approx_corr} that we can choose the input of \cref{algo:klmedian_approx} such that it returns a $6$-approximate solution to the \klmedian objective. 
	
	Finally, we apply a union bound (\conferre \cref{prop:union_bound}) to show that the sample indeed is a weighted \coreset for all $C_1, \dots, C_w$ at a time, as \cref{def:coreset} requires. We have that $\Pr[\cap_{i \in \{1, \dots, w\}} E_i] = 1- \Pr[\cup_{i \in \{1, \dots, w\}} \overline{E_i}]$. We obtain:
	\begin{align*}
		\Pr[\cap_{i \in \{1, \dots, w\}} E_i] = 1- \Pr[\cup_{i \in \{1, \dots, w\}} \overline{E_i}] \geq 1 - \sum_{i=1}^{w} \Pr[\overline{E_i}] \geq 1 - \sum_{i=1}^w \frac{\rho}{n^k} \geq  1 - \rho
	\end{align*} 
	Set $\rho \leq \frac{1}{3}$, then the sample is a strong weighted \coreset with probability at least $\frac{2}{3}$.
	
	It remains to show, that \cref{algo:klmedian_coreset} correctly computes the probability distribution $\psi$. Let $\hat{C} \eqdef \{ \hat{c}_1, \dots, \hat{c}_k \}$ be the $\alpha$-approximate center-set returned by \cref{algo:klmedian_approx} in \cref{algo:klmedian_coreset_approx_line} of \cref{algo:klmedian_coreset}. In \cref{lem:total_sens_bound} we will show in detail that for $j \in \{1, \dots, k\}$ and $\tau \in U_j \eqdef \cluster{\Tau}{\hat{C}}{\hat{c}_j}$ we can define $\lambda(\tau) \eqdef \frac{2m_j + \frechet{\tau}{\hat{c}_j}}{(1-\gamma)\Phi} + \frac{2}{\gamma \varphi(U_j)}$, where \[m_j \eqdef \frac{1}{\varphi\left( U_j \right)} \cdot \sum\limits_{\tau \in U_j} \frechet{\tau}{\hat{c}_j} \cdot \varphi(\tau),\] \[\Phi \eqdef \frac{1}{\alpha} \sum\limits_{j=1}^k \sum\limits_{\tau \in U_j} \frechet{\tau}{\hat{c}_j} \varphi(\tau)\] and $\gamma \eqdef \frac{1}{\frac{\sqrt{3 \alpha k}}{\sqrt{2k}} + 1}$. Thus, we have that $\Lambda = (3 \alpha + 2 \sqrt{6 \alpha} + 2)k$. By \cref{theo:klmedian_approx_corr} we have that $\alpha = 6$. Putting everything together (recall that $\varphi$ is the uniform distribution) we obtain for $j \in \{1, \dots, k\}$ and $\tau \in U_j$: \[ \psi(\tau) \eqdef \frac{\lambda(\tau)}{\Lambda} \varphi(\tau) = \frac{\varphi(\tau)}{32k} \cdot \left( \frac{2m_j + \frechet{\tau}{\hat{c}_j}}{\nicefrac{3}{4}\Phi} + \frac{2}{\nicefrac{1}{4} \varphi(U_j)} \right) = \frac{1}{32kn} \cdot \left( \frac{2m_j + \frechet{\tau}{\hat{c}_j}}{\nicefrac{3}{4}\Phi} + \frac{8n}{\vert U_j \vert} \right). \]
\end{proof}
Note that we could also define $Z \eqdef \frac{1}{\ell} \sum_{j=1}^{\ell} Z_j$ in \cref{theo:klmedian_coreset} and then apply Hoeffding's inequality (\conferre \cref{theo:hoeff}). However, this would lead to a weighted \coreset of cardinality $\On{\frac{k^3 \cdot \ln(n)}{\epsilon^2}}$, thus it is more beneficial to apply Bernstein's inequality.
\paragraph{Here we prove some deferred lemmas.}
\begin{lemma}[\citet{langbergschulman}]
	\label{lem:total_sens_bound}
	Let $\hat{C} \eqdef \{ \hat{c}_1, \dots, \hat{c}_k \}$ be an $\alpha$-approximate center-set for $\Tau$, \ie $\cost{\Tau}{\hat{C}} \leq \alpha \cdot \optcost{\Tau}$. Then for $\tau \in \Tau$ we can define $\lambda(\tau)$ such that $\Lambda = (3 \alpha + 2\sqrt{6\alpha} + 2) k$.
\end{lemma}
\begin{proof}
	For $j = 1, \dots, k$ let $U_j \eqdef \cluster{\Tau}{\hat{C}}{\hat{c}_j}$ be the clusters of $\hat{C}$ with respect to $\Tau$ and let \[m_j \eqdef \frac{1}{\varphi\left( U_j \right)} \cdot \sum\limits_{\tau \in U_j} \frechet{\tau}{\hat{c}_j} \cdot \varphi(\tau).\] Also let \[\Phi \eqdef \frac{1}{\alpha} \sum\limits_{j=1}^k \sum\limits_{\tau \in U_j} \frechet{\tau}{\hat{c}_j} \cdot \varphi(\tau)\] be a lower bound on all $\E_\varphi[X_1], \dots, \E_\varphi[X_w]$ (recall that $\hat{C}$ is an $\alpha$-approximate center-set and $\varphi$ is the uniform distribution). Using these definitions, we will show in \cref{lem:xi_bound} that for an arbitrary constant $\gamma \in [0,1]$ it holds that \[\forall j \in \{1,\dots,k\}\forall \tau \in U_j: \xi(\tau) \leq \frac{2m_j + \frechet{\tau}{\hat{c}_j}}{(1-\gamma)\Phi} + \frac{2}{\gamma \varphi(U_j)}.\] Thus, for $j \in \{1,\dots,k\}$ and $\tau \in U_j$ we can define \[\lambda(\tau) \eqdef \frac{2m_j + \frechet{\tau}{\hat{c}_j}}{(1-\gamma)\Phi} + \frac{2}{\gamma \varphi(U_j)}.\]
	Now we obtain:
	\begin{align*}
		\Lambda ={} & \sum_{\tau \in \Tau} \lambda(\tau) \cdot \varphi(\tau) = \sum_{j=1}^k \sum_{\tau \in U_j} \lambda(\tau) \cdot \varphi(\tau) = \sum_{j=1}^k \sum_{\tau \in U_j} \left(\frac{2m_j + \frechet{\tau}{\hat{c}_j}}{(1-\gamma)\Phi} + \frac{2}{\gamma \varphi(U_j)}\right) \cdot \varphi(\tau) \\
		={} & \sum_{j=1}^k \left[\frac{1}{(1-\gamma)\Phi} \cdot \left(2m_j \cdot \sum_{\tau \in U_j} \varphi(\tau) + \sum_{\tau \in U_j} \frechet{\tau}{\hat{c}_j} \cdot \varphi(\tau)\right) + \frac{2}{\gamma \varphi(U_j)}\sum_{\tau \in U_j} \varphi(\tau)\right] \\
		={} & \sum_{j=1}^k \left[\frac{2m_j \varphi(U_j)}{(1-\gamma)\Phi} + \frac{m_j \varphi(U_j)}{(1-\gamma)\Phi} + \frac{2}{\gamma}\right] = \frac{\sum_{j=1}^k 2m_j \varphi(U_j)}{(1-\gamma)\Phi} + \frac{\sum_{j=1}^k m_j \varphi(U_j)}{(1-\gamma)\Phi} + \frac{2k}{\gamma} \\
		={} & \frac{3\alpha k}{(1-\gamma)} + \frac{2k}{\gamma}
	\end{align*}
	Here we mainly use the fact that $\sum_{j=1}^{k}\varphi(U_j) \cdot m_j = \alpha \cdot \Phi$ by definition.
	
	We now want to minimize the obtained bound with respect to $\gamma$. We search for a local minimum for $\gamma \in [0,1]$. Let $h(\gamma) \eqdef \frac{3\alpha k}{(1-\gamma)} + \frac{2k}{\gamma}$ be the obtained bound $\Lambda$ as a function of $\gamma$. We have: $\frac{\partial h}{\partial \gamma} = \frac{3 \alpha k}{(1-\gamma)^2} - \frac{2k}{\gamma^2}$.
	We obtain a value for $\gamma$:
	\begin{align*}
		\frac{3 \alpha k}{(1-\gamma)^2} - \frac{2k}{\gamma^2} \stackrel{\text{!}}{=} 0 \Leftrightarrow \frac{3 \alpha k}{(1-\gamma)^2} = \frac{2k}{\gamma^2} \Leftrightarrow \left(\frac{\sqrt{3 \alpha k}}{\sqrt{2 k}} + 1\right) \gamma = 1 \Leftrightarrow \gamma = \frac{1}{\frac{\sqrt{3 \alpha k}}{\sqrt{2 k}} + 1}
	\end{align*}
	Because the denominator in the last equality is positive and greater than or equal to the numerator, the obtained value lies in $(0,1]$ and because $\frac{\partial^2 h}{\partial \gamma^2} = \frac{6 \alpha k}{(1-\gamma)^2} + \frac{4 k}{\gamma^2}$ the obtained value is at a local minimum. We insert the value for $\gamma$ into $\Lambda = \frac{3\alpha k}{(1-\gamma)} + \frac{2k}{\gamma}$, which yields to:
	\begin{align*}
		\Lambda = \frac{3\alpha k}{\left(1-\frac{1}{\frac{\sqrt{3 \alpha k}}{\sqrt{2 k}} + 1}\right)} + \frac{2k}{\frac{1}{\frac{\sqrt{3 \alpha k}}{\sqrt{2 k}} + 1}} ={} & \frac{3\alpha k}{\left(1-\frac{1}{\frac{\sqrt{3 \alpha k}}{\sqrt{2 k}} + 1}\right)} + 2k\left(\frac{\sqrt{3 \alpha k}}{\sqrt{2 k}} + 1\right) \\
		={} & \frac{3 \alpha k}{\frac{\frac{\sqrt{3 \alpha k} + \sqrt{2k}}{\sqrt{2 k}}-1}{\frac{\sqrt{3 \alpha k} + \sqrt{2k}}{\sqrt{2 k}}}} + 2k\left(\frac{\sqrt{3 \alpha k}}{\sqrt{2 k}} + 1\right) \\
		={} & \frac{3 \alpha k\left(\frac{\sqrt{3 \alpha k} + \sqrt{2k}}{\sqrt{2 k}}\right)}{\frac{\sqrt{3 \alpha k} + \sqrt{2k}}{\sqrt{2 k}}-1} + 2k\left(\frac{\sqrt{3 \alpha k}}{\sqrt{2 k}} + 1\right) \\
		={} & \frac{3 \alpha k\left(\sqrt{3\alpha k}+ \sqrt{2k}\right)}{\sqrt{3 \alpha k}}+ 2k\left(\frac{\sqrt{3 \alpha k}}{\sqrt{2 k}} + 1\right) \\
		={} & 3 \alpha k + \sqrt{2k}\sqrt{3 \alpha k} + 2k\left(\frac{\sqrt{3 \alpha k}}{\sqrt{2 k}}\right) + 2k \\
		={} & \left(3 \alpha + 2\sqrt{6\alpha} + 2\right) k 
	\end{align*}
	The last equality proves the claim.
\end{proof}
\begin{lemma}[\citet{langbergschulman}]
	\label{lem:xi_bound}
	The following holds: \[\forall \gamma \in [0,1] \forall j \in \{1,\dots,k\}\forall \tau \in U_j: \xi(\tau) \leq \frac{2m_j + \frechet{\tau}{\hat{c}_j}}{(1-\gamma)\Phi} + \frac{2}{\gamma \varphi(U_j)}.\]
\end{lemma}
\begin{proof}
	In \cref{lem:phi_bound_voronoi} we will show that for all $j \in \{1, \dots, k\}$ it holds that (\conferre \cref{def:ball}) \[\varphi\left(U_j \cap \ball{\hat{c}_j}{2 m_j}\right) \geq \frac{\varphi(U_j)}{2}.\] Also in \cref{lem:tau_bound} we will show that for all $i \in \{1, \dots, w\}$, $j \in \{1, \dots, k\}$ and $\tau \in U_j$ it holds that \[\frechet{\tau}{\nearest{\tau}{C_i}} \geq \max\{0, \frechet{\nearest{\hat{c}_j}{C_i}}{\hat{c}_j} - 2m_j\}.\]
	
	For $i \in \{1, \dots, w\}$ we now have the bounds \[ \E_\varphi[X_i] \geq \Phi = \frac{1}{\alpha} \sum\limits_{j=1}^k \sum\limits_{\tau \in U_j} \frechet{\tau}{\hat{c}_j} \cdot \varphi(\tau), \] (recall that $\hat{C}$ is an $\alpha$-approximate center-set and $\varphi$ is the uniform distribution) and for an arbitrary $j \in \{1, \dots, k\}$ we have \[ \E_\varphi[X_i] \geq \sum_{\tau \in U_j \cap \ball{\hat{c}_j}{2m_j}} X_i(\tau) \cdot \varphi(\tau) \geq \max\{0, \frechet{\nearest{\hat{c}_j}{C_i}}{\hat{c}_j} - 2m_j\} \cdot \frac{\varphi(U_j)}{2}, \] by the previous arguments (recall that $X_i(\tau) = \frechet{\tau}{\nearest{\tau}{C_i}}$). 
	
	From now on let $j \in \{1, \dots, k\}$ be arbitrary but fixed. For $\tau \in U_j$ and $\gamma \in [0,1]$ we have:
	\begin{align*}
	\xi(\tau) ={} & \max_{i \in \{1, \dots, w\}} \frac{X_i(\tau)}{\E_\varphi[X_i]} \leq \max_{i \in \{1, \dots, w\}} \frac{\frechet{\tau}{\nearest{\hat{c}_j}{C_i}}}{\E_\varphi[X_i]} \label{ml:eq1} \tag{I} \\
	\leq{} & \max_{i \in \{1, \dots, w\}} \frac{\frechet{\tau}{\hat{c}_j} + \frechet{\hat{c}_j}{\nearest{\hat{c}_j}{C_i}}}{\E_\varphi[X_i]} \label{ml:eq2} \tag{II} \\
	\leq{} & \max_{i \in \{1,\dots,w\}} \frac{\frechet{\tau}{\hat{c}_j} + \frechet{\hat{c}_j}{\nearest{\hat{c}_j}{C_i}}}{\gamma \cdot(\max\{0, \frechet{\nearest{\hat{c}_j}{C_i}}{\hat{c}_j} - 2m_j\}) \frac{\varphi(U_j)}{2} + (1-\gamma) \cdot \Phi} \label{ml:eq3} \tag{III} \\
	\leq{} & \max_{i \in \{1,\dots,w\}, \frechet{\nearest{\hat{c}_j}{C_i}}{\hat{c}_j} \geq 2m_j} \frac{\frechet{\tau}{\hat{c}_j} + \frechet{\hat{c}_j}{\nearest{\hat{c}_j}{C_i}}}{\gamma \cdot(\frechet{\nearest{\hat{c}_j}{C_i}}{\hat{c}_j} - 2m_j) \frac{\varphi(U_j)}{2} + (1-\gamma) \cdot \Phi} \label{ml:eq4} \tag{IV}
	\end{align*}
	Here \cref{ml:eq1} holds because $\frechet{\tau}{\nearest{\tau}{C_i}} \leq \frechet{\tau}{\nearest{\hat{c}_j}{C_i}}$ by \cref{def:nearest_func}, \cref{ml:eq2} follows from the triangle-inequality, \cref{ml:eq3} holds because of the statements above and \cref{ml:eq4} holds because for $i \in \{1, \dots, w\}$ it can be observed, that the term assigns smaller values for $\frechet{\nearest{\hat{c}_j}{C_i}}{\hat{c}_j} < 2m_j$, than for $\frechet{\nearest{\hat{c}_j}{C_i}}{\hat{c}_j} \geq 2m_j$. To be precise for all $\frechet{\nearest{\hat{c}_j}{C_i}}{\hat{c}_j} \leq 2m_j$ we have that $\max\{0, \frechet{\nearest{\hat{c}_j}{C_i}}{\hat{c}_j} - 2m_j\} = 0$, so if the term is maximal for $\frechet{\nearest{\hat{c}_j}{C_i}}{\hat{c}_j} \leq 2m_j$, it clearly is when $\frechet{\nearest{\hat{c}_j}{C_i}}{\hat{c}_j} = 2m_j$, because $\frechet{\nearest{\hat{c}_j}{C_i}}{\hat{c}_j}$ is also present in the numerator.
	
	Now to obtain a bound that is independent of $i$ we substitute $\frechet{\nearest{\hat{c}_j}{C_i}}{\hat{c}_j}$ by a variable $\rho$ and define $h(\rho) \eqdef \frac{\frechet{\tau}{\hat{c}_j} + \rho}{\gamma \cdot(\rho - 2m_j) \frac{\varphi(U_j)}{2} + (1-\gamma) \cdot \Phi}$ to be the bound of \cref{ml:eq4} as a function of $\rho$. We show that $h$ is a monotone function, therefore it is maximized at one of the boundaries of the domain. We calculate the slope of $h(\rho)$:
	\begin{align*}
	\frac{\partial h}{\partial \rho} ={} & \frac{(\gamma \cdot(\rho - 2m_j) \frac{\varphi(U_j)}{2} + (1-\gamma) \cdot \Phi) - (\gamma \cdot \frac{\varphi(U_j)}{2} \cdot (\frechet{\tau}{\hat{c}_j} + \rho))}{(\gamma \cdot(\rho - 2m_j) \frac{\varphi(U_j)}{2} + (1-\gamma) \cdot \Phi)^2} \\
	={} & \frac{\frac{\gamma \rho \varphi(U_j)}{2} - \frac{\gamma 2 m_j \varphi(U_j)}{2} + \Phi - \gamma \Phi - \frac{\gamma \varphi(U_j) \frechet{\tau}{\hat{c}_j}}{2} - \frac{\gamma \rho \varphi(U_j)}{2}}{(\gamma \cdot(\rho - 2m_j) \frac{\varphi(U_j)}{2} + (1-\gamma) \cdot \Phi)^2} \\
	={} & \frac{- \frac{\gamma 2 m_j \varphi(U_j)}{2} + \Phi - \gamma \Phi - \frac{\gamma \varphi(U_j) \frechet{\tau}{\hat{c}_j}}{2}}{(\gamma \cdot(\rho - 2m_j) \frac{\varphi(U_j)}{2} + (1-\gamma) \cdot \Phi)^2}.
	\end{align*}
	The sign of $\frac{\partial h}{\partial \rho}$ is independent of $\rho$ since $\rho$ is not present in the numerator and is squared in the denominator. $h$ must be a monotone function that is either maximized at $\rho = 2m_j$, or $\rho \rightarrow \infty$. Thus, for $j \in \{1, \dots, k\}$ and $\tau \in U_j$ we obtain a bound on $\xi(\tau)$ that is independent of $i$:
	\begin{align*}
	\xi(\tau) \leq \max\left\{ \frac{2m_j + \frechet{\tau}{\hat{c}_j}}{(1-\gamma)\Phi}, \frac{2}{\gamma \varphi(U_j)} \right\} \leq \frac{2m_j + \frechet{\tau}{\hat{c}_j}}{(1-\gamma)\Phi} + \frac{2}{\gamma \varphi(U_j)}.
	\end{align*}
	Here the second term is obtained through an application of l'Hôspital's rule.
\end{proof}
\begin{lemma}
	\label{lem:phi_bound_voronoi}
	The following holds: \[\forall j \in \{1,\dots,k\}: \varphi\left(U_j \cap \ball{\hat{c}_j}{2 m_j}\right) \geq \frac{\varphi(U_j)}{2}.\]
\end{lemma}
\begin{proof}
	Let $j \in \{1, \dots, k\}$ be arbitrary but fixed and for $\tau \in \Tau$ let \[V(\tau) \eqdef \left\{
	\begin{array}{ll}
	\frechet{\tau}{\hat{c}_j}, & \text{if } \tau \in U_j \\
	0, & \text{else}
	\end{array} \right.\] be a random variable. By Markov's inequality (\conferre \cref{theo:markov}) we obtain (recall that $\varphi$ is a probability function):
	\begin{align*}
	\varphi(U_j \setminus \ball{\hat{c}_j}{2 m_j}) ={} & \varphi( V > 2 m_j ) \leq \varphi(V \geq 2 m_j) \leq \frac{\E_\varphi[V]}{2 m_j} \\
	={} & \frac{\left(\sum\limits_{\tau \in U_j} \frechet{\tau}{\hat{c}_j} \cdot \varphi(\tau)\right) + \left(\sum\limits_{\tau \in \Tau \setminus U_j} 0 \cdot \varphi(\tau)\right)}{2 m_j} \\
	={} & \frac{\varphi(U_j)}{2}
	\end{align*}
	Here the first equality holds by the definition of $V$ and \cref{def:ball}.
	
	Because $\varphi(U_j) = \varphi(U_j \setminus \ball{\hat{c}_j}{2 m_j}) + \varphi\left(U_j \cap \ball{\hat{c}_j}{2 m_j}\right)$ (\conferre \cref{def:probability_space}), this implies that $\varphi\left(U_j \cap \ball{\hat{c}_j}{2 m_j}\right) \geq \frac{\varphi(U_j)}{2}$.
\end{proof}
\begin{lemma}
	\label{lem:tau_bound}
	The following holds: 
	\begin{align*}
		\forall i \in \{1, \dots, w\} & \forall j \in \{1,\dots,k\} \forall \tau \in U_j \cap \ball{\hat{c}_j}{2m_j}: \\
		& \frechet{\tau}{\nearest{\tau}{C_i}} \geq \max\{0, \frechet{\nearest{\hat{c}_j}{C_i}}{\hat{c}_j} - 2m_j\}.
	\end{align*}
\end{lemma}
\begin{proof}
	Let $i \in \{1, \dots, w\}$, $j \in \{1, \dots, k\}$ and $\tau \in U_j \cap \ball{\hat{c}_j}{2m_j}$ be arbitrary but fixed and let $c_i \eqdef \nearest{\tau}{C_i} $ (\conferre \cref{def:nearest_func}) be a nearest center at hand. Let $c^\prime_{i} \eqdef \nearest{\hat{c}_j}{C_i}$ be a center, which lies nearest to $\hat{c}_j$. The triangle-inequality gives: \[\frechet{c^\prime_{i}}{\hat{c}_j} \leq \frechet{c_i}{\hat{c}_j} \leq \frechet{c_{i}}{\tau} + \frechet{\tau}{\hat{c}_j}\]
	We examine two cases.
	\paragraph{Case 1:} In this case $c_{i}^\prime \not\in \ball{\hat{c}_j}{2m_j}$ (\conferre \cref{def:ball}). We have:
	\begin{align*}
	\frechet{c^\prime_{i}}{\hat{c}_j} \leq \frechet{c_{i}}{\tau} + 2 m_j \Leftrightarrow \frechet{c_{i}}{\tau} \geq \frechet{c_{i}^\prime}{\hat{c}_j} - 2 m_j
	\end{align*}  
	\paragraph{Case 2:} In this case $c_{i}^\prime \in \ball{\hat{c}_j}{2m_j}$. We can only say that $\frechet{\tau}{c_{i}} \geq 0 $.
	
	Putting everything together we obtain $\frechet{\tau}{c_{i}} \geq \max\{0, \frechet{c_{i}^\prime}{\hat{c}_j} - 2m_j\}$.
\end{proof}
\subsubsection{Time Complexity Analysis of \cref{algo:klmedian_coreset}}
\begin{theorem}
	Given a set $\Tau \subset \eqcfre{m}$ of $n$ polygonal curves and a parameter $\epsilon \in (0,1)$, \cref{algo:klmedian_coreset} has running-time $\On{n^2 \cdot m^2 \log(m) + \frac{\ln^2(n)}{\epsilon^2}}$.
\end{theorem}
\begin{proof}
	At first, we analyze the running-time of the function \textsc{compute-psi} in \cref{func:compute_psi} of \cref{algo:klmedian_coreset}. Let $\hat{C} \eqdef \{ \hat{c}_1, \dots, \hat{c}_k \}$ be the center-set returned by \cref{algo:klmedian_approx}, which has running-time $\On{n^2 \cdot m^2 \log(m)}$ (this will be shown in \cref{theo:median_approx_running-time}). To compute $U_j$ and $m_j$, for $j = 1, \dots, k$, we compute the Fréchet distances between every center and every $\tau \in \Tau$ and store them in a two-dimensional array, such that they can be accessed in constant time. This takes time $\On{n \cdot m^2 \log(m)}$. We use these values to compute $\Phi$ and eventually to compute $\psi$, which can then be done in time $\On{n}$. Thus, all in all, the function \textsc{compute-psi} has running-time $\On{n^2 \cdot m^2 \log(m)}$. We assume that the values of $\psi$ are stored in an array, such that \cref{algo:klmedian_coreset} can access them in constant time. To be able to sample from $\psi$ we store, for $i = 1, \dots, n$, the cumulative probabilities that $\tau_1, \dots, \tau_i$ occur, in $A[i]$ which is a field of an array $A[0, \dots, n+1]$. We set $A[0] \eqdef 0$ and $A[n+1] \eqdef 1$. For every curve that shall be sampled we sample a real number $a$ uniformly from $[0,1]$. We assume this can be done in constant time. Then we use binary-search to find the $A[i]$, for which $A[i-1] \leq a$ and $A[i] \geq a$ and then return $\tau_i$. Thus, every $\tau \in \Tau$ is sampled with probability $\psi(\tau)$. 
	
	In this way, the sampling-step in \cref{algo:klmedian_coreset_line_sampling} can be done in time $\On{\log(n) \cdot \frac{k^2 \ln(n)}{\epsilon^2}} = \On{\frac{\ln^2(n)}{\epsilon^2}}$ and the weighting-step in $\On{\frac{\ln(n)}{\epsilon^2}}$. All in all \cref{algo:klmedian_coreset} has running-time $\On{n^2 \cdot m^2 \log(m) + \frac{\ln^2(n)}{\epsilon^2}}$.
\end{proof}
\subsection{A Constant-Factor Approximation Algorithm}
\label{subsec:median_approx}
We introduce \cref{algo:klmedian_approx}, the constant-factor approximation algorithm that is used in \cref{algo:klmedian_coreset}. But first we introduce \cref{algo:klcenter_approx_mod}, which is used in \cref{algo:klmedian_approx} for the purpose of reducing the running-time of \cref{algo:klmedian_approx} to a polynomial in $n$ and $m$. (Otherwise the algorithm has running-time exponential in $n$.)

\cref{algo:klcenter_approx_mod} computes a $3$-approximate solution to the \klcenter objective and also to a modified version of the \klcenter objective, where the center-set stems from $\Tau$ instead of $\eqcfre{l}$. The algorithm is a marginally modified version of \cref{algo:klcenter_approx}, which is presented in \cite[Section 7.2]{approx_k_l_center} and which we already know from \cref{subsec:center_approx}. It originates from \citet{gonzalez} algorithm.

\begin{algorithm}
	\caption{Compute 3-Approximate Solution to the $k$-\textsc{center} Objective}\label{algo:klcenter_approx_mod}
	\begin{algorithmic}[1]
		\Procedure{k-center-approx}{$\Tau$}\label{proc:k_center_approx}
		\State $C \gets \emptyset$
		\State $c_1 \gets$ arbitrary $\tau \in \Tau$
		\State $C \gets C \cup \{c_1\}$
		\For{$i=2, \dots, k$}
		\State $c_i \gets \argmax\limits_{\tau \in \Tau} \min\limits_{j \in \{1, \dots, i-1\}} \frechet{\tau}{c_j}$
		\State $C \gets C \cup \{ c_i \}$
		\EndFor
		\State \textbf{return} $C$
		\EndProcedure
	\end{algorithmic}
\end{algorithm}

\cref{algo:klcenter_approx_mod} works as follows: At first it picks an arbitrary curve from $\Tau$ as the first center. In the subsequent $k-1$ steps the algorithm picks a curve that maximizes the minimal distance of the curve to a nearest center of the current center-set as next center. Afterwards it returns the resulting set.

The following theorem states the correctness and running-time of \cref{algo:klcenter_approx}:
\begin{theorem}[{\citet[Theorem 20]{approx_k_l_center}}]
	\label{theo:approx_klcenter_raw}
	Given a set of polygonal curves $\Tau \eqdef \{ \tau_1, \dots, \tau_n \} \subset \eqcfre{m}$ and $k,l \in \mathbb{N}_{>0}$, \cref{algo:klcenter_approx} computes a $(c+2)$-approximation to the \klcenter objective in time $\On{kn \cdot lm \log(l+m)+k\cdot T_l(m)}$, where $T_l(m)$ is the running-time to compute a $c$-approximate $l$-simplification of a polygonal curve of complexity at most $m$.
\end{theorem}

It is clear, that \cref{algo:klcenter_approx_mod} computes a center-set that is a subset of $\Tau$ and also that it does not do any $l$-simplification-step, thus $c=1$ and $T_l(m) = 0$.

We will use the following corollary to \cref{theo:approx_klcenter_raw} throughout this section.
\begin{corollary}
	\label{coro:median_3n_approx}
	Given a set $\Tau \eqdef \{\tau_1, \dots, \tau_n\} \subset \eqcfre{m}$ of polygonal curves, \cref{algo:klcenter_approx_mod} computes a $3$-approximate solution to the discrete $k$-\textsc{center} objective, \ie \par $\min\limits_{C \subseteq \Tau, \vert C \vert = k} \max\limits_{\tau \in \Tau} \frechet{\tau}{\nearest{\tau}{C}}$, in time $\On{n \cdot m^2 \log(m)}$.
\end{corollary}
\begin{proof}
	Let $\hat{C}$ be the center-set returned by \cref{algo:klcenter_approx_mod}. From \cref{theo:approx_klcenter_raw} we know that \[ \min_{C \subset \eqcfre{l}, \vert C \vert = k} \max_{\tau \in \Tau} \frechet{\tau}{\nearest{\tau}{C}} \leq \max_{\tau \in \Tau} \frechet{\tau}{\nearest{\tau}{\hat{C}}} \leq (c+2) \cdot \min_{C \subset \eqcfre{l}, \vert C \vert = k} \max_{\tau \in \Tau} \frechet{\tau}{\nearest{\tau}{C}}. \] Because \cref{algo:klcenter_approx_mod} directly adds curves that are picked from $\Tau$ to the center-set, instead of an $\el$-simplification of these, we have that $c = 1$ and $\el = m$, hence:
	\begin{align*}
	\min_{C \subset \eqcfre{m}, \vert C \vert = k} \max_{\tau \in \Tau} \frechet{\tau}{\nearest{\tau}{C}} \leq{} & \min_{C \subseteq \Tau, \vert C \vert = k} \max_{\tau \in \Tau} \frechet{\tau}{\nearest{\tau}{C}} \leq \max_{\tau \in \Tau} \frechet{\tau}{\nearest{\tau}{\hat{C}}} \\
	\leq{} & 3 \cdot \min_{C \subset \eqcfre{m}, \vert C \vert = k} \max_{\tau \in \Tau} \frechet{\tau}{\nearest{\tau}{C}} \\
	\leq{} & 3 \cdot \min_{C \subseteq \Tau, \vert C \vert = k} \max_{\tau \in \Tau} \frechet{\tau}{\nearest{\tau}{C}}
	\end{align*}
	These inequalities hold because $\Tau \subset \eqcfre{m}$.
	
	This proves the approximation-factor. The running-time follows from the fact, that we do not use \el-simplifications.
\end{proof}

The following proposition is also crucial for \cref{algo:klmedian_approx} to have running-time polynomial in $n$ and $m$. It states, that the solution of \cref{algo:klcenter_approx_mod} is a $(3n)$-approximate solution for the \klmedian objective. We will use such a solution as initial guess for \cref{algo:klmedian_approx}.
\begin{proposition}
	\label{prop:center_median_rel}
	Let $\Tau \eqdef \{ \tau_1, \dots, \tau_n \} \subset \eqcfre{m}$ be a set of polygonal curves and let $\approxcost{\Tau}$ be a $3$-approximate solution to the discrete $k$-\textsc{center} objective, \ie \par $\min\limits_{C \subseteq \Tau, \vert C \vert = k} \max\limits_{\tau \in \Tau} \frechet{\tau}{\nearest{\tau}{C}}$, with center-set $\hat{C}$. Then $\sum\limits_{\tau \in \Tau} \frechet{\tau}{\nearest{\tau}{\hat{C}}}$ is a $(3\cdot n)$-approximate solution to the \klmedian objective.
\end{proposition}
\begin{proof}
	Let \[C^\ast_m \eqdef \argmin_{C \subseteq \Tau, \vert C \vert =k} \sum_{\tau \in \Tau} \frechet{\tau}{\nearest{\tau}{C}}\] denote an optimal center-set for the \klmedian objective and \[C^\ast_c \eqdef \argmin_{C \subseteq \Tau, \vert C \vert =k} \max_{\tau \in \Tau} \frechet{\tau}{\nearest{\tau}{C}}\] denote an optimal center-set for the $k$-\textsc{center} objective (the center-set is a subset of $\Tau$ instead of $\eqcfre{l}$).
	By the definition of $C^\ast_m$ we have: \[ \sum_{\tau \in \Tau} \frechet{\tau}{\nearest{\tau}{C^\ast_m}} \leq \sum_{\tau \in \Tau} \frechet{\tau}{\nearest{\tau}{C^\ast_c}} \leq n \cdot \max_{\tau \in \Tau} \frechet{\tau}{\nearest{\tau}{C^\ast_c}} \leq n \cdot \max_{\tau \in \Tau} \frechet{\tau}{\nearest{\tau}{\hat{C}}} \] Similarly by the definition of $C^\ast_c$: \[ \frac{1}{3} \max_{\tau \in \Tau} \frechet{\tau}{\nearest{\tau}{\hat{C}}} \leq \max_{\tau \in \Tau} \frechet{\tau}{\nearest{\tau}{C^\ast_c}} \leq \max_{\tau \in \Tau} \frechet{\tau}{\nearest{\tau}{C^\ast_m}} \leq \sum_{\tau \in \Tau} \frechet{\tau}{\nearest{\tau}{C^\ast_m}} \]
	This yields to:
	\begin{align*}
	\frac{1}{3 \cdot n} \sum_{\tau \in \Tau} \frechet{\tau}{\nearest{\tau}{C^\ast_m}} \leq{} & \frac{1}{3 \cdot n} \sum_{\tau \in \Tau} \frechet{\tau}{\nearest{\tau}{\hat{C}}} \leq \frac{1}{3} \max_{\tau \in \Tau} \frechet{\tau}{\nearest{\tau}{\hat{C}}} \\
	\leq{} & \max_{\tau \in \Tau} \frechet{\tau}{\nearest{\tau}{C^\ast_c}} \leq \sum_{\tau \in \Tau} \frechet{\tau}{\nearest{\tau}{C^\ast_m}}
	\end{align*}
	We obtain $\sum\limits_{\tau \in \Tau} \frechet{\tau}{\nearest{\tau}{C^\ast_m}} \leq \sum\limits_{\tau \in \Tau} \frechet{\tau}{\nearest{\tau}{\hat{C}}} \leq n \cdot 3 \cdot \sum\limits_{\tau \in \Tau} \frechet{\tau}{\nearest{\tau}{C^\ast_m}}$, which finishes the proof.
\end{proof}
Now we are ready to state the constant-factor approximation algorithm for the \klmedian objective, that is used in \cref{algo:klmedian_coreset}. \cref{algo:klmedian_approx} as well as the following proofs of correctness and running-time are adapted from \citet[Section 4.3]{har_peled_geo_approx} and originate from \citet{arya_local_search}.

\begin{algorithm}
	\caption{Compute Approximate Solution to the \klmedian Objective}\label{algo:klmedian_approx}
	\begin{algorithmic}[1]
		\Procedure{k-median-approx}{$\Tau$,$\gamma$}\label{proc:kmedian_approx}
		\State $C \gets $\textsc{k-center-approx}$(\Tau)$ \Comment{\cref{algo:klcenter_approx_mod}}
		\State $\approxcost{\Tau} \gets \cost{\Tau}{C}$
		\While{$\exists c \in C \exists \tau \in \Tau\setminus C: \cost{\Tau}{C} - \gamma \approxcost{\Tau} > \cost{\Tau}{(C \cup \{ \tau \})\setminus \{c\}}$}
		\State $C \gets (C \cup \{\tau\})\setminus\{c\}$
		\EndWhile
		\State \textbf{return} $C$
		\EndProcedure
	\end{algorithmic}
\end{algorithm}

\cref{algo:klmedian_approx} works as follows: It runs \cref{algo:klcenter_approx_mod} on $\Tau$ and uses the resulting solution as initial guess. Then it tries to locally improve the current solution by \emph{swap-operations}, \ie it looks for a $\tau \in \Tau \setminus C$ and a $c \in C$ such that $\cost{\Tau}{C} - \gamma \approxcost{\Tau} > \cost{\Tau}{(C \cup \{ \tau \})\setminus \{c\}}$. If such a $c$ and $\tau$ exist, it swaps $c$ and $\tau$, \ie $C$ becomes $\left(C \setminus \{c\}\right) \cup \{\tau\}$. Here $\approxcost{\Tau}$ is the value of the approximate solution returned by \cref{algo:klcenter_approx_mod} and $\gamma$ is a scaling-factor, which is used to obtain a fraction of a lower bound of $\optcost{\Tau}$. This prevents \cref{algo:klmedian_approx} from doing a number of steps, that is greater than any polynomial in $n$ and $m$. The algorithm returns $C$, when no more swap-operation can be done.

\subsubsection{Correctness Analysis of \cref{algo:klmedian_approx}}
\begin{theorem}[\citet{har_peled_geo_approx}]
	\label{theo:klmedian_approx_corr}
	Given a set $\Tau \eqdef \{ \tau_1, \dots, \tau_n \} \subset \eqcfre{m}$ of polygonal curves and a rational number $\gamma \eqdef \frac{1}{k\cdot 3n}$, \cref{algo:klmedian_approx} returns a $6$-approximate solution to the \klmedian objective.
\end{theorem}
\begin{proof}
	It is clear that \cref{algo:klmedian_approx} terminates, because there are at most $n^k$ possible center-sets and the algorithm has to do a proper improvement in every swap-operation. Therefore, the algorithm will eventually get stuck with a solution that can not be improved.
	
	Let $\hat{C}$ be a $(3n)$-approximate solution to the \klmedian objective returned by \cref{algo:klcenter_approx_mod}, with $\approxcost{\Tau} \eqdef \cost{\Tau}{\hat{C}}$, \conferre \cref{coro:median_3n_approx,prop:center_median_rel}. Let $C$ be the center-set returned by \cref{algo:klmedian_approx}, which we call the local center-set and let $C^\ast \subseteq \Tau$, with $\vert C^\ast \vert = k$, be an optimal center-set, \ie $\cost{\Tau}{C^\ast} = \optcost{\Tau}$. For $C \subseteq \Tau$ and $c,c^\prime \in \Tau$ let \[\swap{C}{c}{c^\prime} \eqdef (C \cup \{c^\prime\})\setminus \{c\}\] be the center-set obtained by a swap-operation of \cref{algo:klmedian_approx}. Because the algorithm terminated, we know:
	\begin{equation*}
	\forall c \in C \forall c^\prime \in \Tau\setminus C: \cost{\Tau}{C} \leq \cost{\Tau}{\swap{C}{c}{c^\prime}} + \gamma \approxcost{\Tau}. \label{eq:algo_prop} \tag{I}
	\end{equation*}
	For $c^\ast \in C^\ast$ the optimal cost of a cluster $\cluster{\Tau}{C^\ast}{c^\ast}$ is denoted by \[\mathfrak{o}(c^\ast) \eqdef \sum_{\tau \in \cluster{\Tau}{C^\ast}{c^\ast}} \frechet{\tau}{c^\ast}\] and we denote the cost with respect to $C$ by \[\mathfrak{l}(c^\ast) \eqdef \sum_{\tau \in \cluster{\Tau}{C^\ast}{c^\ast}} \frechet{\tau}{\nearest{\tau}{C}}.\] 
	\cref{eq:algo_prop} is then equivalent to:
	
	$\forall c \in C \forall c^\prime \in \Tau \setminus C:$
	\begin{equation}
	\sum_{c^\ast \in C^\ast} \mathfrak{l}(c^\ast) \leq \gamma \approxcost{\Tau} + \sum_{c^\ast \in C^\ast} \sum_{\tau \in \cluster{\Tau}{C^\ast}{c^\ast}} \frechet{\tau}{\nearest{\tau}{\swap{C}{c}{c^\prime}}} \label{eq:localopt} \tag{II}
	\end{equation}
	For $c^\ast \in C^\ast$ and $c \in C$ let \[\mathcal{C}^{+}(c^{\ast}, c) \eqdef \cluster{\Tau}{C^\ast}{c^{\ast}} \cap \cluster{\Tau}{C}{c}\] denote the intersection of an optimal cluster with a local cluster, also let \[\mathcal{C}^{-}(c^{\ast}, c) \eqdef \cluster{\Tau}{C^\ast}{c^{\ast}} \setminus \cluster{\Tau}{C}{c}\] denote the difference of an optimal cluster with respect to a local cluster. Let $c^\ast \in C^\ast$ be arbitrary but fixed. \cref{eq:localopt} yields:
	
	$\forall c \in C \forall c^\prime \in \Tau \setminus C:$
	\begin{align*}
	\mathfrak{l}(c^\ast) + \sum_{c^{\ast}_o \in C^\ast\setminus \{c^\ast\}} \mathfrak{l}(c^{\ast}_o) \leq{} & \gamma \approxcost{\Tau} + \sum_{c^{\ast}_o \in C^\ast} \sum_{\tau \in \cluster{\Tau}{C^\ast}{c^{\ast}_o}} \frechet{\tau}{\nearest{\tau}{\swap{C}{c}{c^\prime}}} \\
	\leq{} & \gamma \approxcost{\Tau} + \sum_{\tau \in \cluster{\Tau}{C^\ast}{c^\ast}} \frechet{\tau}{\nearest{\tau}{\swap{C}{c}{c^\prime}}} \\
	& + \sum_{c^{\ast}_o \in C^\ast\setminus\{c^\ast\}} \sum_{\tau \in \mathcal{C}^{-}(c^{\ast}_o, c)} \frechet{\tau}{\nearest{\tau}{C}} \\
	& + \sum_{c^{\ast}_o \in C^\ast\setminus\{c^\ast\}} \sum_{\tau \in \mathcal{C}^{+}(c^{\ast}_o, c)} \frechet{\tau}{\nearest{\tau}{\swap{C}{c}{c^\prime}}} \\
	={} & \gamma \approxcost{\Tau} + \sum_{\tau \in \cluster{\Tau}{C^\ast}{c^\ast}} \frechet{\tau}{\nearest{\tau}{\swap{C}{c}{c^\prime}}} \\
	& + \sum_{c^{\ast}_o \in C^\ast\setminus\{c^\ast\}} \sum_{\tau \in \mathcal{C}^{+}(c^{\ast}_o, c)} \left[\frechet{\tau}{\nearest{\tau}{\swap{C}{c}{c^\prime}}} - \frechet{\tau}{\nearest{\tau}{C}}\right] \\
	& + \sum_{c^{\ast}_o \in C^\ast\setminus\{c^\ast\}} \mathfrak{l}(c^{\ast}_o) \tag{III} \label{eq:clusterchange}
	\end{align*}
	In words the obtained bound expresses that the following happens if we make a swap: All curves $\tau \in \Tau \setminus \cluster{\Tau}{C}{c}$ remain unaffected while the $\tau \in \cluster{\Tau}{C}{c}$ do now contribute cost that is equal to the distance to a nearest center in $\swap{C}{c}{c^\prime}$. By the structure of \cref{eq:clusterchange}, \ie we have $\sum_{c^{\ast}_o \in C^\ast\setminus \{c^\ast\}} \mathfrak{l}(c^{\ast}_o)$ on the left-hand side and on the right-hand side, we immediately obtain the following inequalities through equivalence:
	
	$\forall c^\ast \in C^\ast \forall c \in C \forall c^\prime \in \Tau \setminus C:$
	\begin{align*}
	\mathfrak{l}(c^\ast) \leq{} & \gamma \approxcost{\Tau} + \sum_{\tau \in \cluster{\Tau}{C^\ast}{c^\ast}} \frechet{\tau}{\nearest{\tau}{\swap{C}{c}{c^\prime}}} \\
	& + \sum_{c^{\ast}_o \in C^\ast\setminus\{c^\ast\}} \sum_{\tau \in \mathcal{C}^{+}(c^{\ast}_o, c)} \left[\frechet{\tau}{\nearest{\tau}{\swap{C}{c}{c^\prime}}} - \frechet{\tau}{\nearest{\tau}{C}}\right] \\
	\leq{} & \gamma \approxcost{\Tau} + \sum_{\tau \in \cluster{\Tau}{C^\ast}{c^\ast}} \frechet{\tau}{\nearest{\tau}{\swap{C}{c}{c^\prime}}} \\
	& + \sum_{c^{\ast}_o \in C^\ast\setminus\{c^\ast\}} \sum_{\tau \in \mathcal{C}^{+}(c^{\ast}_o, c)} \left[\frechet{\tau}{\nearest{\nearest{\tau}{C^\ast}}{\swap{C}{c}{c^\prime}}} - \frechet{\tau}{\nearest{\tau}{C}}\right] \tag{IV} \label{eq:single_opt_cluster}
	\end{align*}
	Here the last inequality holds by \cref{def:nearest_func}. Now we partition $C$ to obtain better bounds from \cref{eq:single_opt_cluster}, therefore consider the bipartite directed graph \[G_\eta \eqdef (C \cup C^\ast, E_\eta),\] where $E_\eta \eqdef \{ (c^\ast, c) \mid c^\ast \in C^\ast, c = \eta(c^\ast, C) \}$.
	We define three types of vertices of $C$:
	\begin{enumerate}
		\item Drifters $ C_d \eqdef \{ c \in C \mid \indegree{c} = 0 \}.$ That are local centers that are no nearest neighbor to any optimal center.
		\item Anchors $ C_a \eqdef \{ c \in C \mid \indegree{c} = 1 \} $. These are local centers that are nearest neighbors to exactly one optimal center.
		\item Tyrants $ C_t \eqdef \{ c \in C \mid \indegree{c} > 1 \}$. These are local centers that are nearest neighbors to multiple optimal centers.
	\end{enumerate}
	Similarly, we define two types of vertices of $C^\ast$:
	\begin{enumerate}
		\item Optimal centers $C^\ast_a \eqdef \{ c^\ast \in C^\ast \mid \forall (c^\ast, c) \in E: \indegree{c} = 1 \}$ whose nearest neighbor in $C$ is no nearest neighbor to any other $c^\ast_o \in C^\ast\setminus\{c^\ast\}$.
		\item Optimal centers $C^\ast_t \eqdef \{ c^\ast \in C^\ast \mid \forall (c^\ast, c) \in E: \indegree{c} > 1 \}$ whose nearest neighbor in $C$ is a nearest neighbor to some other $c^\ast_o \in C^\ast\setminus\{c^\ast\}$.
	\end{enumerate}
	Now we have:
	\begin{align*}
	\cost{\Tau}{C} ={} & \sum_{c^\ast \in C^\ast} \mathfrak{l}(c^\ast) = \sum_{c^\ast \in C_t^\ast} \mathfrak{l}(c^\ast) + \sum_{c^\ast \in C_a^\ast} \mathfrak{l}(c^\ast) \\
	\leq{} & \sum\limits_{c^\ast \in C^\ast} \gamma \approxcost{\Tau} + \sum\limits_{c^\ast \in C^\ast} \mathfrak{o}(c^\ast) + \sum\limits_{c \in C} \sum\limits_{c^{\ast}_o \in C^\ast} \sum\limits_{\tau \in \mathcal{C}^{+}(c^{\ast}_o, c)} 4 \frechet{\tau}{\nearest{\tau}{C^\ast}} \tag{V} \label{eq:lemmaeq} \\
	={} & \optcost{\Tau} + k \gamma \approxcost{\Tau} + \sum\limits_{c \in C} \sum\limits_{c^{\ast}_o \in C^\ast} \sum\limits_{\tau \in \mathcal{C}^{+}(c^{\ast}_o, c)} 4 \frechet{\tau}{\nearest{\tau}{C^\ast}} \\
	={} & \optcost{\Tau} + k \gamma \approxcost{\Tau} + \sum\limits_{\tau \in \Tau} 4 \frechet{\tau}{\nearest{\tau}{C^\ast}} \tag{VI} \label{eq:union} \\
	={} & 5 \optcost{\Tau} +  k \gamma \approxcost{\Tau}.
	\end{align*}
	Here \cref{eq:lemmaeq} holds, because we show in \cref{lem:tyrant_bound} that \[\sum\limits_{c^\ast \in C^\ast_t} \mathfrak{l}(c^\ast)  \leq  \sum\limits_{c^\ast \in C^\ast_t} \gamma \approxcost{\Tau} + \sum\limits_{c^\ast \in C^\ast_t} \mathfrak{o}(c^\ast) + 2 \sum\limits_{c \in C_d} \sum\limits_{c^{\ast}_o \in C^\ast} \sum\limits_{\tau \in \mathcal{C}^{+}(c^{\ast}_o, c)} 2 \frechet{\tau}{\nearest{\tau}{C^\ast}}\] holds, and we show in \cref{lem:anchor_bound} that \[\sum\limits_{c^\ast \in C^\ast_a} \mathfrak{l}(c^\ast)  \leq  \sum\limits_{c^\ast \in C^\ast_a} \gamma \approxcost{\Tau} + \sum\limits_{c^\ast \in C^\ast_a} \mathfrak{o}(c^\ast) + \sum\limits_{c \in C_a} \sum\limits_{c^{\ast}_o \in C^\ast} \sum\limits_{\tau \in \mathcal{C}^{+}(c^{\ast}_o, c)} 2 \frechet{\tau}{\nearest{\tau}{C^\ast}}\] holds. \cref{eq:union} holds, because \[ \bigcup_{c \in C} \bigcup_{c^\ast \in C^\ast} \mathcal{C}^{+}(c^{\ast}, c) = \Tau \] holds by definition, therefore $\sum\limits_{c \in C} \sum\limits_{c^{\ast}_o \in C^\ast} \sum\limits_{\tau \in \mathcal{C}^{+}(c^{\ast}_o, c)} 4 \frechet{\tau}{\nearest{\tau}{C^\ast}} = \sum\limits_{\tau \in \Tau} 4 \frechet{\tau}{\nearest{\tau}{C^\ast}}$.
	
	Now by the definition of $\gamma = \frac{1}{k\cdot 3n}$ we have that $k \gamma \approxcost{\Tau} \leq \optcost{\Tau}$ (recall that $\approxcost{\Tau} \leq 3\cdot n \cdot \optcost{\Tau}$), thus we obtain \[ \cost{\Tau}{C} \leq 6 \optcost{\Tau}. \]
\end{proof}
\paragraph{Here we prove some deferred lemmas.}
\begin{lemma}[\citet{har_peled_geo_approx}]
	\label{lem:tyrant_bound}
	The following holds: \[\sum\limits_{c^\ast \in C^\ast_t} \mathfrak{l}(c^\ast)  \leq  \sum\limits_{c^\ast \in C^\ast_t} \gamma \approxcost{\Tau} + \sum\limits_{c^\ast \in C^\ast_t} \mathfrak{o}(c^\ast) + 2 \sum\limits_{c \in C_d} \sum\limits_{c^{\ast}_o \in C^\ast} \sum\limits_{\tau \in \mathcal{C}^{+}(c^{\ast}_o, c)} 2 \frechet{\tau}{\nearest{\tau}{C^\ast}}.\]
\end{lemma}
\begin{proof}
	Without loss of generality assume $\vert C^\ast_t \vert > 0$. Let $c^\ast \in C^\ast_t$ and $c \in C_d$ be arbitrary. Consider swapping $c$ and $c^\ast$. We immediately get that for all $\tau \in \Tau$ it holds that $\frechet{\tau}{\nearest{\nearest{\tau}{C^\ast}}{\swap{C}{c}{c^\ast}}} \leq \frechet{\tau}{\nearest{\nearest{\tau}{C^\ast}}{C}}$, because for any optimal center $c^\ast \in C$ it holds that $\nearest{c^\ast}{C} \neq c$.
	Combining this fact with \cref{eq:single_opt_cluster} in the proof of \cref{theo:klmedian_approx_corr} we obtain:
	\begin{align*}
	\mathfrak{l}(c^\ast) \leq{} & \gamma \approxcost{\Tau} + \mathfrak{o}(c^\ast) + \sum_{c^{\ast}_o \in C^\ast\setminus\{c^\ast\}} \sum_{\tau \in \mathcal{C}^{+}(c^{\ast}_o, c)} \left[\frechet{\tau}{\nearest{\nearest{\tau}{C^\ast}}{C}} - \frechet{\tau}{\nearest{\tau}{C}}\right] \\
	\leq{} & \gamma \approxcost{\Tau} + \mathfrak{o}(c^\ast) + \sum_{c^{\ast}_o \in C^\ast\setminus\{c^\ast\}} \sum_{\tau \in \mathcal{C}^{+}(c^{\ast}_o, c)} 2 \frechet{\tau}{\nearest{\tau}{C^\ast}} \\
	\leq{} & \gamma \approxcost{\Tau} + \mathfrak{o}(c^\ast) + \sum_{c^{\ast}_o \in C^\ast} \sum_{\tau \in \mathcal{C}^{+}(c^{\ast}_o, c)} 2 \frechet{\tau}{\nearest{\tau}{C^\ast}}
	\end{align*}
	The first inequality holds, because the curves in $\cluster{\Tau}{C^\ast}{c^\ast}$ pay at most $\mathfrak{o}(c^\ast)$ and the second inequality holds because $ \frechet{\tau}{\nearest{\nearest{\tau}{C^\ast}}{C}} - \frechet{\tau}{\nearest{\tau}{C}} \leq 2 \frechet{\tau}{\nearest{\tau}{C^\ast}}$ for every $\tau \in \Tau$, what will be shown in \cref{lem:cluster_switch}.
	
	When we consider such a swap for every $c^\ast \in C^\ast_t$ then we obtain by adding up the respective inequalities:
	\begin{align*}
	\sum_{c^\ast \in C^\ast_t} \mathfrak{l}(c^\ast) \leq \sum_{c^\ast \in C^\ast_t} \gamma \approxcost{\Tau} + \sum_{c^\ast \in C^\ast_t} \mathfrak{o}(c^\ast) + 2 \sum_{c \in C_d} \sum_{c^{\ast}_o \in C^\ast} \sum_{\tau \in \mathcal{C}^{+}(c^{\ast}_o, c)} 2 \frechet{\tau}{\nearest{\tau}{C^\ast}}
	\end{align*}
	This inequality holds, because every drifter $c \in C_d$ has to be used at most twice, what will be shown in \cref{lem:number_drifter}.
\end{proof}
\begin{lemma}[\citet{har_peled_geo_approx}]
	\label{lem:anchor_bound}
	The following holds: \[\sum\limits_{c^\ast \in C^\ast_a} \mathfrak{l}(c^\ast)  \leq  \sum\limits_{c^\ast \in C^\ast_a} \gamma \approxcost{\Tau} + \sum\limits_{c^\ast \in C^\ast_a} \mathfrak{o}(c^\ast) + \sum\limits_{c \in C_a} \sum\limits_{c^{\ast}_o \in C^\ast} \sum\limits_{\tau \in \mathcal{C}^{+}(c^{\ast}_o, c)} 2 \frechet{\tau}{\nearest{\tau}{C^\ast}}.\]
\end{lemma}
\begin{proof}
	Let $c^\ast \in C_a^\ast$ be arbitrary but fixed and let $c \eqdef \nearest{c^\ast}{C}$ be the anchor of $c^\ast$. We know that for all $c^\ast_o \in C^\ast \setminus \{c^\ast\}$ and for all $\tau \in \cluster{\Tau}{C^\ast}{c^\ast_o}$ it holds that $\nearest{\tau}{C^\ast} \neq c^\ast$, therefore $ \nearest{\nearest{\tau}{C^\ast}}{C} \neq c$. We conclude that 
	$\frechet{\tau}{\nearest{\nearest{\tau}{C^\ast}}{\swap{C}{c}{c^\ast}}} \leq \frechet{\tau}{\nearest{\nearest{\tau}{C^\ast}}{C}}$ holds in such a setting.
	
	Again, combining this fact with \cref{eq:single_opt_cluster} in the proof of \cref{theo:klmedian_approx_corr} we obtain:
	\begin{align*}
	\mathfrak{l}(c^\ast) \leq{} & \gamma \approxcost{\Tau} + \mathfrak{o}(c^\ast) + \sum_{c^{\ast}_o \in C^\ast\setminus\{c^\ast\}} \sum_{\tau \in \mathcal{C}^{+}(c^{\ast}_o, c)} \left[\frechet{\tau}{\nearest{\nearest{\tau}{C^\ast}}{C}} - \frechet{\tau}{\nearest{\tau}{C}}\right] \\
	\leq{} & \gamma \approxcost{\Tau} + \mathfrak{o}(c^\ast) + \sum_{c^{\ast}_o \in C^\ast\setminus\{c^\ast\}} \sum_{\tau \in \mathcal{C}^{+}(c^{\ast}_o, c)} 2 \frechet{\tau}{\nearest{\tau}{C^\ast}} \\
	\leq{} & \gamma \approxcost{\Tau} + \mathfrak{o}(c^\ast) + \sum_{c^{\ast}_o \in C^\ast} \sum_{\tau \in \mathcal{C}^{+}(c^{\ast}_o, c)} 2 \frechet{\tau}{\nearest{\tau}{C^\ast}}
	\end{align*}
	Similarly to \cref{lem:tyrant_bound} the first inequality holds, because the curves in $\cluster{\Tau}{C^\ast}{c^\ast}$ pay at most $\mathfrak{o}(c^\ast)$ and the second inequality holds because $ \frechet{\tau}{\nearest{\nearest{\tau}{C^\ast}}{C}} - \frechet{\tau}{\nearest{\tau}{C}} \leq 2 \frechet{\tau}{\nearest{\tau}{C^\ast}}$ for every $\tau \in \Tau$, what will be shown in \cref{lem:cluster_switch}.
	
	Now we consider swapping every $c^\ast \in C^\ast_a$ with the respective anchor $\nearest{c^\ast}{C}$. By adding up the respective inequalities we obtain:
	\begin{align*}
	\sum_{c^\ast \in C^\ast_a} \mathfrak{l}(c^\ast) \leq \sum_{c^\ast \in C^\ast_a} \gamma \approxcost{\Tau} + \sum_{c^\ast \in C^\ast_a} \mathfrak{o}(c^\ast) + \sum_{c \in C_a} \sum_{c^{\ast}_o \in C^\ast} \sum_{\tau \in \mathcal{C}^{+}(c^{\ast}_o, c)} 2 \frechet{\tau}{\nearest{\tau}{C^\ast}}
	\end{align*}
	This inequality holds because trivially $\vert C^\ast_a \vert = \vert C_a \vert$ by definition.
\end{proof}
\begin{lemma}[\citet{har_peled_geo_approx}]
	\label{lem:cluster_switch}
	The following holds: \[\forall \tau \in \Tau: \frechet{\tau}{\nearest{\nearest{\tau}{C^\ast}}{C}} - \frechet{\tau}{\nearest{\tau}{C}} \leq 2 \frechet{\tau}{\nearest{\tau}{C^\ast}}.\]
\end{lemma}
\begin{proof}
	Let $\tau \in \Tau$ be arbitrary but fixed and let $c^\ast \eqdef \nearest{\tau}{C^\ast}$, $c \eqdef \nearest{\tau}{C}$ and finally $c^\prime \eqdef \nearest{c^\ast}{C}$. By definition $\frechet{c^\ast}{c^\prime} \leq \frechet{c^\ast}{c}$. The triangle-inequality gives:
	\begin{align*}
	\frechet{\tau}{c^\prime} \leq \frechet{\tau}{c^\ast} + \frechet{c^\ast}{c^\prime} \leq \frechet{\tau}{c^\ast} + \frechet{c^\ast}{c} \leq \frechet{\tau}{c^\ast} + \frechet{c^\ast}{\tau} + \frechet{\tau}{c}
	\end{align*}
	The last inequality proves the claim.
\end{proof}
\begin{lemma}[\citet{har_peled_geo_approx}]
	\label{lem:number_drifter}
	The following holds: \[2 \vert C_d \vert \geq \vert C^\ast_t \vert.\]
\end{lemma}
\begin{proof}
	By definition $C^\ast_t \cup C^\ast_a = C^\ast$ and $\vert C_t \vert \leq \frac{1}{2} \vert C^\ast_t \vert$, as well as $\vert C_a \vert = \vert C^\ast_a \vert$. We have:
	\begin{align*}
	\vert C_t \vert + \vert C_a \vert + \vert C_d \vert = \vert C^\ast_t \vert + \vert C^\ast_a \vert \Leftrightarrow \vert C_t \vert + \vert C_d \vert = \vert C^\ast_t \vert \Rightarrow \vert C_d \vert \geq \frac{1}{2} \vert C^\ast_t \vert
	\end{align*}
\end{proof}
\subsubsection{Time Complexity Analysis of \cref{algo:klmedian_approx}}
\begin{theorem}
	\label{theo:median_approx_running-time}
	Given a set of $n$ polygonal curves and a rational number $\gamma \eqdef \frac{1}{k \cdot 3n}$, \cref{algo:klmedian_approx} has running-time $\On{n^2 \cdot m^2 \log(m)}$.
\end{theorem}
\begin{proof}
	Let $C$ be the center-set returned by \cref{algo:klcenter_approx_mod}, which has running-time \par $\On{nm^2 \log(m)}$, \conferre \cref{coro:median_3n_approx}. Let $\approxcost{\Tau} \eqdef \cost{\Tau}{C}$. We know that $\optcost{\Tau} \leq \approxcost{\Tau} \leq 3n \cdot \optcost{\Tau}$, \conferre \cref{prop:center_median_rel}. Further we know that \cref{algo:klmedian_approx} improves the solution in every step by at least $\gamma \cdot \approxcost{\Tau}$. We conclude that \cref{algo:klmedian_approx} does at most $\frac{1}{\gamma \cdot \approxcost{\Tau}} \cdot (3n \cdot \optcost{\Tau} - \optcost{\Tau})$ steps. This yields to a maximum number of steps:
	\begin{align*}
	\frac{3n \cdot \optcost{\Tau} - \optcost{\Tau}}{\gamma \cdot \approxcost{\Tau}} \leq \frac{3n \cdot \optcost{\Tau} - \optcost{\Tau}}{\frac{1}{3n \cdot k} \cdot \approxcost{\Tau}} \leq \frac{3n - 1}{\frac{1}{k}} = 3nk - k.
	\end{align*}
	For every step the algorithm checks for each $c \in C$ and each $\tau \in \Tau\setminus C$ whether they can be swapped. This takes time $\On{n \cdot m^2 \log(m)}$, because for every possible swap, which are up to $n \cdot k$ many, the value of the objective function has to be evaluated, which takes time $\On{m^2 \log(m)}$, \conferre \cref{theo:frechet_algo}.
	
	Taking the number of steps into account we have the overall running-time $\On{n^2 \cdot m^2 \log(m)}$.
\end{proof}
\clearpage
\section{The \klmeans Objective}
\label{sec:klmeans}
\FloatBarrier
In this section we are working with \cref{def:klmeans}, \ie the \klmeans objective. We are given a set of curves $\Tau \eqdef \{ \tau_1, \dots, \tau_n \} \subset \eqcfre{m}$ and we are looking for a center-set $C \subseteq \Tau$, of cardinality $k$, that minimizes the sum of the squared distances between the curves in $\Tau$ and a respective nearest center. 

With respect to point-sets the $k$-\textsc{means} objective has a peculiarity. Namely, for $k = 1$ we can compute the optimal center analytically with \cref{def:centroid}, \ie the \emph{centroid} of the point-set, as it is shown in \cref{prop:zauberformel}. This benefit is used in most clustering algorithms for the $k$-\textsc{means} problem, such as Lloyd's algorithm, \conferre \cite{lloyd}. It is also used in \coreset constructions for the $k$-\textsc{means} problem, such as in \cite{bico}. One may assume that this peculiarity extends to polygonal curves, \ie for the $(1,l)$-\textsc{means} one can obtain an optimal center-curve by connecting the centroid of the first vertices of the input-curves to the centroid of the second vertices of the input-curves and so on. Of course, these considerations would require the input-curves to have the exact same complexity. We show that such a construction does not necessarily yield an optimal center, even for line segments.

\begin{figure}
	\centering
	\def\svgwidth{0.85\textwidth}
\begingroup%
  \makeatletter%
  \providecommand\color[2][]{%
    \errmessage{(Inkscape) Color is used for the text in Inkscape, but the package 'color.sty' is not loaded}%
    \renewcommand\color[2][]{}%
  }%
  \providecommand\transparent[1]{%
    \errmessage{(Inkscape) Transparency is used (non-zero) for the text in Inkscape, but the package 'transparent.sty' is not loaded}%
    \renewcommand\transparent[1]{}%
  }%
  \providecommand\rotatebox[2]{#2}%
  \newcommand*\fsize{\dimexpr\f@size pt\relax}%
  \newcommand*\lineheight[1]{\fontsize{\fsize}{#1\fsize}\selectfont}%
  \ifx\svgwidth\undefined%
    \setlength{\unitlength}{393.93995729bp}%
    \ifx\svgscale\undefined%
      \relax%
    \else%
      \setlength{\unitlength}{\unitlength * \real{\svgscale}}%
    \fi%
  \else%
    \setlength{\unitlength}{\svgwidth}%
  \fi%
  \global\let\svgwidth\undefined%
  \global\let\svgscale\undefined%
  \makeatother%
  \begin{picture}(1,0.44431684)%
    \lineheight{1}%
    \setlength\tabcolsep{0pt}%
    \put(0,0){\includegraphics[width=\unitlength]{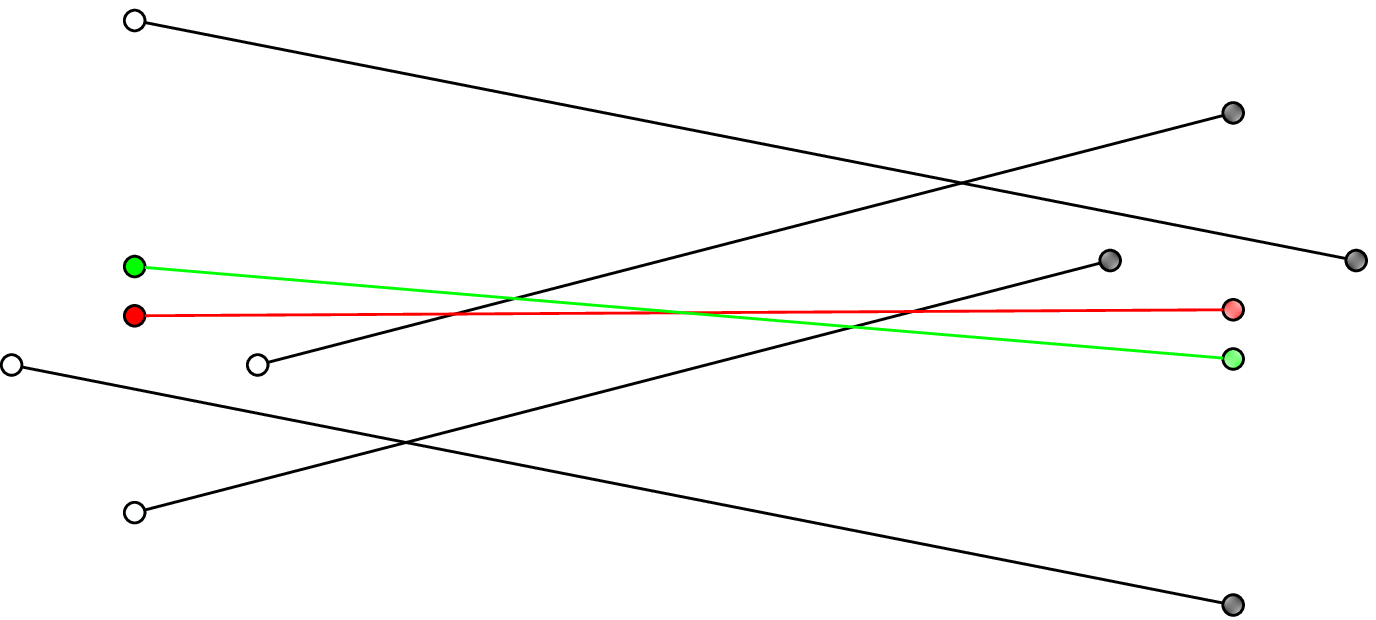}}%
    \put(0.4276539,0.26318767){\color[rgb]{0,0,0}\makebox(0,0)[lt]{\lineheight{1.25}\smash{\begin{tabular}[t]{l}$\tau_1$\end{tabular}}}}%
    \put(0.59202421,0.08263743){\color[rgb]{0,0,0}\makebox(0,0)[lt]{\lineheight{1.25}\smash{\begin{tabular}[t]{l}$\tau_2$\end{tabular}}}}%
    \put(0.4318977,0.38386511){\color[rgb]{0,0,0}\makebox(0,0)[lt]{\lineheight{1.25}\smash{\begin{tabular}[t]{l}$\tau_3$\end{tabular}}}}%
    \put(0.4088839,0.17230756){\color[rgb]{0,0,0}\makebox(0,0)[lt]{\lineheight{1.25}\smash{\begin{tabular}[t]{l}$\tau_4$\end{tabular}}}}%
    \put(0.6204602,0.23706007){\color[rgb]{1,0,0}\makebox(0,0)[lt]{\lineheight{1.25}\smash{\begin{tabular}[t]{l}$\mu_l$\end{tabular}}}}%
    \put(0.18447475,0.26127244){\color[rgb]{0,1,0}\makebox(0,0)[lt]{\lineheight{1.25}\smash{\begin{tabular}[t]{l}$\nu_l$\end{tabular}}}}%
  \end{picture}%
\endgroup%

	\caption{Construction of $\Tau$ in \cref{prop:counterexample_line_center} for $d = 2$.}
	\label{fig:counterexample_line_center}
\end{figure}

\begin{proposition}
	\label{prop:counterexample_line_center}
	Let $\Tau \eqdef \{ \tau_1, \dots, \tau_n \} \subset \eqcfre{2}$ be a set of line segments and $\mu_0 \eqdef \centroid{\{ \tau(0) \mid \tau \in \Tau \}}$ be the centroid of their initial points, as well as $\mu_1 \eqdef \centroid{\{ \tau(1) \mid \tau \in \Tau \}}$ be the centroid of their end points. The center with minimum cost under $(1,2)$-\textsc{means} is not necessarily $\mu_l \eqdef \overline{\mu_0 \mu_1}$.
\end{proposition}
\begin{proof}
	Let $n \eqdef 2 \cdot r$, for an arbitrary integer $r \geq 1$ and $\Tau^\prime \eqdef \{ \tau_1, \dots, \tau_{\nicefrac{n}{2}} \}$. We define $\nu_0 \eqdef \centroid{\{\tau(0) \mid \tau \in \Tau^\prime\}}$, $\nu_1 \eqdef \centroid{\{\tau(1) \mid \tau \in \Tau \setminus \Tau^\prime\}}$ and $\nu_l \eqdef \overline{\nu_0 \nu_1}$. Let $\Tau$ be chosen, such that the following holds (\conferre \cref{fig:counterexample_line_center} for an example):
	\begin{alignat}{1}
		\forall \tau \in \Tau^\prime: & \max\{ \eucl{\tau(0)}{\mu_0}, \eucl{\tau(1)}{\mu_1}\} = \eucl{\tau(0)}{\mu_0} \label{vorb1} \tag{I} \\
		\forall \tau \in \Tau^\prime: & \max\{ \eucl{\tau(0)}{\nu_0}, \eucl{\tau(1)}{\nu_1}\} = \eucl{\tau(0)}{\nu_0} \label{vorb2} \tag{II}\\
		\forall \tau \in \Tau \setminus \Tau^\prime: & \max\{ \eucl{\tau(0)}{\mu_0}, \eucl{\tau(1)}{\mu_1} \} = \eucl{\tau(1)}{\mu_1} \label{vorb3} \tag{III} \\
		\forall \tau \in \Tau \setminus \Tau^\prime: & \max\{ \eucl{\tau(0)}{\nu_0}, \eucl{\tau(1)}{\nu_1} \} = \eucl{\tau(1)}{\nu_1} \label{vorb4} \tag{IV} \\
		& \mu_0 \neq \nu_0 \wedge \mu_1 \neq \nu_1 \label{vorb5} \tag{V}
	\end{alignat}
	Let $c \eqdef \mu_l$, from \cref{obs:frechet_endpoints} we know that:
	\begin{alignat*}{3}
		\cost{\Tau}{\{c\}} ={} & \sum_{\tau \in \Tau} \frechet{\tau}{c}^2 = \sum_{\tau \in \Tau} \max\{ \eucl{\tau(0)}{c(0)}^2, \eucl{\tau(1)}{c(1)}^2 \} \\
		={} & \sum\limits_{\tau \in \Tau^\prime} \eucl{\tau(0)}{c(0)}^2 + \sum\limits_{\tau \in \Tau \setminus \Tau^\prime} \eucl{\tau(1)}{c(1)}^2 \label{eq:vorb} \tag{1} \\
		={} & \frac{\vert \Tau \vert}{2} \cdot \eucl{c(0)}{\nu_0}^2 + \frac{\vert \Tau \vert}{2} \cdot \eucl{c(1)}{\nu_1}^2 \\
		& + \sum\limits_{\tau \in \Tau^\prime} \eucl{\tau(0)}{\nu_0}^2 + \sum\limits_{\tau \in \Tau \setminus \Tau^\prime} \eucl{\tau(1)}{\nu_1}^2 \label{eq:centroidline} \tag{2}
	\end{alignat*}
	Here \cref{eq:vorb} follows from \cref{vorb1} and \cref{vorb3}. \cref{eq:centroidline} follows from \cref{prop:zauberformel}, \cref{vorb2} and \cref{vorb4}. Since \cref{vorb5}, it can be observed that $\nu_l$ is a better center than $\mu_l$.
\end{proof}
Because of \cref{prop:counterexample_line_center} and the fact that the means objective is more sensitive towards outliers we see no benefits in studying \coreset constructions for the \klmeans objective.
	\chapter{Discussion}
In this thesis we have provided construction methods for \coreset[s] for the \klcenter clustering problem and the \klmedian clustering problem. Also, we have proven that the benefits of the $k$-\textsc{means} objective with respect to point-sets do not extend to the \klmeans objective, even for sets of line segments.

These results are more or less satisfying. The most satisfying results are presented in \cref{sec:klmedian}, \ie the \coreset construction for the \klmedian objective. Because we restricted ourselves to the discrete median objective we were able to adapt the sensitivity sampling framework from \citeauthor{langbergschulman}, which was originally stated for point-sets from $\euclideanspace^d$ endowed with some $\ell_p$ norm, to curves in $\euclideanspace^d$ endowed with the Fréchet distance. A benefit in this adaption is that the initial random variables, which are used as estimators for the cost of the clusterings (one estimator per possible center-set), are defined with respect to the uniform distribution. Thus, we could use a $c$-approximate solution to the clustering to obtain a lower bound on the expected values of these random variables and therefore also for the random variables used by the sensitivity sampling framework. This bound is needed to build the probability distribution, which is yielded by the framework and that we used to obtain an \coreset with constant probability. For general probability distributions this is not the case, here we need at least a bi-criteria approximation on a minimal expected value, where minimal means minimal under the choice of the center-set. The crucial point is that the bi-criteria optimization, where the parameters are a scaling factor on $k$ and the approximation-guarantee, does not necessarily yield a $c$-approximate solution on the optimal value of the clustering and vice versa. We can use such bi-criteria approximations for general metric spaces, though, \conferre \cite{median_bi}. Because these approximations are often realized through a greedy approach that is an extension of a local-search scheme, our results are roughly similar in terms of the approximation guarantee and running-time: We use a modified version of the constant-factor approximation algorithm by \citeauthor{arya_local_search}, which is a well-known local-search heuristic. This algorithm was developed for general metric spaces, therefore, it can be used for curves under the Fréchet distance. We provide a straightforward proof of the correctness and approximation-guarantee of the algorithm. Further we use the constant-factor approximation algorithm for \klcenter clustering by \citeauthor{approx_k_l_center} to obtain an initial guess for the local-search algorithm, which then has polynomial running-time, thus obtaining a construction technique for \coreset[s] with polynomial running-time.

Another satisfying result is \cref{prop:embedding}. Here we are able to point out a difference between the $d$-dimensional euclidean space and the Fréchet space $(\eqcfre{m}, \pfrechet)$, which has the $d$-dimensional euclidean space as ambient space. We were even able to derive \cref{coro:kiss_lines} from this result, hinting that the Fréchet spaces (under the restriction of polygonal curves) have own kissing-numbers, which depend only on the maximum number of vertices of the curves. It is conceivable that the kissing number of $(\eqcfre{m}, \pfrechet)$ is bounded by $\psi_d^m$ and that this can be proven by similar arguments as those we use in the proof of \cref{coro:kiss_lines}. If this proves to be true there may be the chance that one can construct an isometry from a Fréchet space, with ambient space $\euclideanspace^d$, to a $d^\prime$-dimensional euclidean space. In this way one can at least obtain the value of a clustering. If there is also interest in the center-curves, a post-processing is imaginable which takes the pairwise distances between the points and their center-points and uses these to construct meaningful center-curves.

The results that are less satisfying are presented in \cref{sec:klcenter}, \ie the \coreset constructions for the \klcenter objective. Here we are working with geometric decompositions based on grids, which work well for line segments, but less good for polygonal curves with complexity at least $3$, because then the cardinality of the \coreset is dependent on the ratio of the length of the longest edge of a center-curve and the cost of the approximate clustering we use to construct the grids. For line segments, \cref{algo:klcenter_coreset_1} provides \coreset[s] of cardinality $\On{\nicefrac{1}{\epsilon^{2d}}}$ in all cases while, for polygonal curves of complexity at least $3$, \cref{algo:klcenter_coreset_mg1} provides \coreset[s] of cardinality exponential in $m$ and sub-linear in $n$, if the ratio of a longest edge of an approximate clustering returned by \cref{algo:klcenter_approx} and the clustering objective value does not exceed $\sqrt[m]{n}$. Otherwise, \cref{algo:klcenter_coreset_mg1} is not able to provide an \coreset at all. A depiction of such a case can be observed in \cref{fig:wc_pc_center}, where the number of input-curves is rigorously outnumbered by the number of cells of the constructed grids. Additional to this disadvantage, the construction is very complicated and requires several very technical proofs, although the idea behind it is very simple, which makes it tedious to improve it.

\begin{figure}
	\centering
	\includegraphics[width=\textwidth]{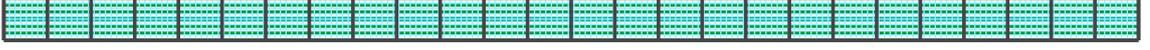}
	\caption{Exemplary (part of) worst case curve-set for the \coreset construction in \cref{subsec:mg2}. An edge of a center-curve is depicted in light green with six edges of curves from the input-set with equidistant distances in dark green, three on each side of the center-edge. The cubes are depicted in black and the associated grids are depicted in light blue. It can be observed that there are orders of magnitude between the number of curves and the number of cells.}
	\label{fig:wc_pc_center}
\end{figure}
	\bibliography{bib}
\end{document}